\documentclass[aps,pra,showpacs,reprint,onecolumn,groupedaddress,floatfix]{revtex4-1}
\usepackage{graphicx}
\usepackage{amsmath}
\usepackage{latexsym}
\usepackage{amsfonts}
\usepackage{mathtools}
\usepackage{MnSymbol}
\usepackage{amsthm}
\usepackage{wasysym}
\usepackage{textcomp}
\usepackage{color} 
\usepackage{multirow}
\usepackage{pifont}
\usepackage[T1]{fontenc}
\theoremstyle{definition}
\newtheorem{definition}{Definition}[section]
\newtheorem{theorem}{Theorem}[section]
\newtheorem{lemma}{Lemma}[section]

\begin{document}

\def\crta{\vrule height1.41ex depth-1.27ex width0.34em}
\def\dj{d\kern-0.36em\crta}
\def\Crta{\vrule height1ex depth-0.86ex width0.4em}
\def\Dj{D\kern-0.73em\Crta\kern0.33em}
\dimen0=\hsize \dimen1=\hsize \advance\dimen1 by 40pt

\definecolor{darkgreen}{rgb}{0.1,0.6,0.3}
\font\1=cmss8
\font\2=cmssdc8
\font\3=cmr8
\font\4=cmss7


\title{Arbitrarily Exhaustive Hypergraph Generation of 4-, 6-, 8-,
  16-, and 32-Dimensional Quantum Contextual Sets}

\author{Mladen Pavi\v ci\'c}

\email{mpavicic@physik.hu-berlin.de}

\affiliation{Department of Physics---Nanooptics, 
Faculty of Math. and Natural Sci.~I,
Humboldt University of Berlin, Germany and\\
Center of Excellence for Advanced Materials
and Sensing Devices (CEMS), Photonics and Quantum Optics Unit,
Ru{\dj}er Bo\v skovi\'c Institute, Zagreb, Croatia}

\date{\today}

\begin{abstract}
Quantum contextuality turns out to be a necessary resource for
universal quantum computation and important in the field of quantum
information processing. It is therefore of interest both for
theoretical considerations and for experimental implementation to
find new types and instances of contextual sets and develop methods
of their optimal generation. We present an arbitrary exhaustive
hypergraph-based generation of the most explored contextual
sets---Kochen-Specker (KS) ones---in 4, 6, 8, 16, and 32 dimensions.
We consider and analyse twelve KS classes and obtain numerous
properties of theirs, which we then compare with the results
previously obtained in the literature. We generate several thousand
times more types and instances of KS sets than previously known.
All KS sets in three of the classes and in the upper part of a fourth 
are novel. We make use of the MMP hypergraph language, algorithms, and 
programs to generate KS sets strictly following their definition from
the Kochen-Specker theorem. This approach proves to be particularly
advantageous over the parity-proof-based ones (which prevail in the
literature), since it turns out that only a very few KS sets have a
parity proof (in six KS classes \textless\ 0.01\% and in one of
them 0\%). MMP hypergraph formalism enables a translation of an
exponentially complex task of solving systems of nonlinear
equations, describing KS vector orthogonalities, into a
statistically linearly complex task of evaluating vertex states of
hypergraph edges, thus exponentially speeding up the generation of
KS sets and enabling us to generate billions of novel instances of
them. The MMP hypergraph notation also enables us to graphically
represent KS sets and to visually discern their features.
\end{abstract}


\maketitle

\section{Introduction}
\label{sec:intro}

An assumed property of a classical system is that any of its 
measurements has values independent of other compatible 
measurements that might have been carried out on the system 
previously or counterfactually simultaneously, i.e., that the 
values are predetermined. The property is called the 
non-contextuality. This is in contrast to quantum mechanical 
systems whose measurements might be contextual, i.e., dependent 
on the context of previous or counterfactually simultaneous 
measurements. Such property of quantum systems is called 
(quantum) contextuality. 

The so-called Kochen-Specker (KS) sets provide constructive 
proofs of quantum contextuality and therefore provide
straightforward blueprints for their implementation and
experimental setups. KS sets are likely to find applications
in the field of quantum information, similarly to ones recently
found for the Bell setups in implementing entanglements
\cite{heinbriegel04,cabello-moreno-09}. The assumption is
supported by a recent result of A.~Cabello \cite{cabello-10},
according to which local contextuality can be used to reveal
quantum nonlocality.

Along that road, it has been most recently ``demonstrate[d] that
\dots\ contextuality is the source of a quantum computer's
power'' \cite{bartlett-nature-14}. In particular, Howard, Wallman,
Veitech, and Emerson \cite{magic-14} ``uncover a remarkable
connection between the power of quantum computers and \dots
contextuality'' \cite{bartlett-nature-14} and prove that 
``contextuality is a necessary resource for universal quantum
computation via magic state distillation'' \cite[p.~354]{magic-14}. 
(``The way of initializing the quantum bits [by means of] \dots
superposition \dots is called {\em magic}'' \cite{bartlett-nature-14}.) 
The scheme of Howard {\em et al.\/} \cite{magic-14} has been
extended by Delfosse, Guerin, Bian, and Raussendorf so as to include
Wigner function negativity \cite{delfosse-raussendorf-15}. 

It has also been recently shown by Raussendorf that ``the
measurement-based quantum computations which compute a nonlinear
Boolean function with a high probability are contextual''
\cite{raussendorf-13}. 

A contextual kind of quantum gates---indispensable ingredients of
quantum computational circuits---can be straightforwardly constructed
from the scheme which served Waegell and Aravind to build 4-dim
complex KS sets \cite{waeg-aravind-jpa-11}.

On the other hand, Pavi{\v c}i{\'c}, Mc{K}ay, Megill, and Fresl have
shown that KS sets can serve as a generator of a new kind of
lattices within Hilbert lattice representation of the Hilbert
space \cite[Fig.~8]{bdm-ndm-mp-fresl-jmp-10}, where a Hilbert lattice
is an algebra underlying every Hilbert space. In addition, Megill
and Pavi{\v c}i{\'c} have shown how new generalized orthoarguesian
equations---the only known equations, apart from the orthomodularity
equation itself, holding in the algebra of closed subspaces of a
Hilbert space---can be generated from KS sets \cite{mp-7oa}. 

Another quantum information contextual KS set application is
a quantum cryptography protection, as outlined by Cabello,
{D'A}mbrosio, Nagali, and Sciarrino \cite{cabello-dambrosio-11}.
It has even been shown by Nagata that the KS theorem is a precondition
for secure quantum key distribution (QKD) in the sense that in
each QKD protocol KS non-contextuality is violated
\cite{nagata-05}.

A series of KS experiments have been carried out during the last
ten years. They were implemented for 4-dim systems with photons 
\cite{simon-zeil00,michler-zeil-00,amselem-cabello-09,liu-09,d-ambrosio-cabello-13,ks-exp-03},
neutrons
\cite{h-rauch06,cabello-fillip-rauch-08,b-rauch-09}
trapped ions \cite{k-cabello-blatt-09}, and 
molecular nuclear spins in the solid states \cite{moussa-09},
for 6-dim systems via six path possibilities for
the photon transmission through a diffractive aperture
\cite{lisonek-14,canas-cabello-14}, and for 8-dim systems by means
of the linear transverse momentum of single photons transmitted by
diffractive apertures addressed in spatial light
modulators \cite{canas-cabello-8d-14}. 

The aforementioned role of contextual sets in  ``supply[ing] `magic'
for quantum computation'' \cite{magic-14} would require numerous
instances of contextual sets and here KS sets as the most numerous
contextual sets are likely to have an important role in designing
appropriate schemes for implementations and applications.
Then, in order to test different quantum gates for KS sets we
should be able to engineer sufficiently large number of vectors
for them, i.e., KS sets of different complexities. For
constructing new algebraic structures and equations for the
Hilbert space we should also have an arbitrary increasing number of
KS sets as explicitly shown in \cite{mp-7oa}.
Finally, it is of theoretical significance to know the structure,
features, and sizes of various KS sets. Taken together, it is
important to find new classes and new instances of non-redundant
non-isomorphic KS sets as well as different coordinatizations
for them. It is also of importance to design algorithms and programs
with the help of which we can generate an arbitrary number of
different KS sets. 

In this paper, we describe the discovery of large numbers (billions
of them) of critical non-redundant non-isomorphic KS sets in 4-, 6-, 
8-, 16-, and 32-dim Hilbert spaces. ``Critical'' means that they are
minimal in the sense that a removal of any $n$-tuple of mutual
orthogonalities, of $n$ vectors from an $n$-dim Hilbert space, turns
a KS set into a non-KS set. In other words, they represent a KS setup
that has no redundancy.

We describe the features of KS sets within particular KS classes
which emerge as we generate the sets. We also outline patterns
of distribution and generation and compare them with the other
methods of generation in the literature. For instance, huge
blocks of KS sets and even whole classes of KS sets turn out to
be completely invisible with the latter methods. 

The paper is organised as follows.

In Sec.~\ref{sec:for} we provide the reader with a constructive
version of the KS theorem, define KS sets as well as the critical
KS sets, define the parity proof for KS sets, and present the
formalism, algorithms, and programs we make the use of in the paper.

In Secs.~\ref{sec:2424}, \ref{sec:6074}, \ref{sec:60105},
\ref{sec:300675}, and \ref{sec:300675} we deal with KS sets in
4-dim Hilbert space. 

In Sec.~\ref{sec:2424} we review the oldest KS class, the 24-24 one,
which is actually a subclass of the 60-105 class we
introduce in Sec.~\ref{sec:60105}. In Sec.~\ref{sec:6074} we obtain
three orders of magnitude more sets from the class 60-74 then in our
previous paper \cite{mfwap-s-11} and 15 times more than reported in
other literature (the class is also known as the 60-75 and/or 600-cell
based KS class); we denoted the 60-74 class as ``tentative'' in the
title of the section because it is particular subclass of the class
300-675 we introduce in Sec.~\ref{sec:300675}. In Sec.~\ref{sec:60105}
we elaborate on the 60-105 class defined by means of Pauli operators
for two qubits in the complex Hilbert space and obtain ca.~2.5 more
types of KS criticals and $3\times 10^4$ more instances of them than
known from the literature. In Sec.~\ref{sec:300675} we
analyze the recently discovered highly complex and extremely interwoven
300-675 class and find an important subclasses which have been
completely overlooked in the literature at the higher end of the class.
In Sec.~\ref{sec:4dim-witt} we generate a completely new class of
ca.~250,000 KS criticals from the so-called Witting's master set,
recently found in the literature; none of the criticals has a parity
proof and therefore all the obtained sets from class are completely
invisible in the standard approach via parity-proof-based algorithms
and programs.

In the 6-dimensional Hilbert space, the so-called ``seven context''
star-like 21-7 KS critical set has recently been discovered and a
challenge was issued to find bigger 6-dim KS sets in response to which
we in Sec.~\ref{sec:6dim} generate $3.7\times 10^6$ 6-dim KS criticals
in the novel 236-1216 class; all but 8 of the criticals lack a parity
proof; we also show that the vector components of the seven-context-star
KS set can be simplified and that the set itself is not contained in the
latter class. 

In Sec.~\ref{sec:8dim} we generate ten times more types of and KS sets
themselves, from the Lie algebra E8 based 120-2024 master set, than
previously achieved in the literature---due to the very low number of
the parity proofs (0.1\permil); we also construct a real star-like
KS critical set and show that it is not contained in the 120-2024 class.

In Sec.~\ref{sec:16dim} we enter a sparsely charted territory of
16-dim 4-qubit KS sets and generate ca.~$2.5\times 10^6$ more sets and
ca.~70 more types of sets than known from the literature from an 80-265
master set.

In Sec.~\ref{sec:32dim} we generate ca.~$2.5\times 10^5$ more instances
and ca.~153 more types of 32-dim 5-qubit KS criticals (from a 160-661
master set) than known from the literature.

In Sec.~\ref{sec:3d} we revisit the only four known 3-dim KS criticals
and show that recently spotted 13-vector set does not prove the 
Kochen-Specker theorem. 

In Sec.~\ref{sec:discussion} we discuss and compare, both
mutually and with those in the literature, all the KS sets we generated.

In Appendix \ref{app:0} we give all hypergraph strings we refer to
in the main body of the paper.

In Appendix \ref{app:B} we provide the reader with chosen KS
critical sets from most of the types from all classes we considered.

\section{KS Sets, MMP Hypergraphs, Formalism, Algorithms,
  and Programs}
\label{sec:for}

Our aim is to present new results in the realm of contextual setups
and KS sets, methods that served us to generate them, and we introduce
formalism and representation that enable us to handle them. The input
and output data are extremely massive and numerous and they contain all
known (from the previous literature, including our own previous
papers) setups, sets, figures, hypergraphs, and diagrams as a very
special and tiny portion of the ones obtained here.
Before we dwell on details of the formalism we will make use of, we
briefly introduce contextuality vs.~non-contextuality features, quote
the KS theorem, and define a KS set.

The notion of non-contextuality of a system, whose observables we
measure after its passing through a device, boils down to a statement
that measurements of a system corresponds to predetermined values of
the observables during the interaction of the system with the device.
A stronger statement, which is usually called the KS theorem is that
non-contextual theories assume that a predetermined result of a
particular measurement of an observable of a system does not depend
on measurements simultaneously carried out on other observables of
the system, while quantum, contextual theories do not assume any
predetermined values for outcomes of measurements, {\em clicks\/},
0-1s, and might depend on simultaneous measurements.

\begin{theorem}\label{th:ks} ({\em Kochen-Specker} 
\cite{gleason,koch-speck,zimba-penrose})
In ${\cal H}^n$, $n\ge 3$, there are sets of $n$-tuples of mutually
orthogonal vectors to which it is impossible to assign 1s and 0s
in such a way that
\begin{enumerate}
\item No two orthogonal vectors are both
assigned the value 1;
\item In any group of $n$ mutually orthogonal vectors, not all of
the vectors are assigned the value 0.
\end{enumerate}
The sets of such vectors are called {\em KS sets\/} and the vectors
themselves are called {\em KS vectors\/}.
\end{theorem}

Any KS set defined for a quantum system provides a constructive proof
of the KS theorem and of the contextuality of quantum mechanics.
A collection of related measurements provides an experimental
verification of the theorem. 

Within quantum mechanics we can formalize KS set properties in the
following manner. To every quantum observable of a quantum system
there corresponds a linear Hermitian operator in a Hilbert space and
to every state of the system associated to the observable there
corresponds an eigenvector of the operator in the same space. The
result of a measurement of the observable is associated with the
eigenvalue of the operator. Any KS set is represented by a
collection of $n$-tuples of mutually orthogonal (eigen)vectors
from $n$-dim Hilbert spaces. 

In this paper we consider 3-, 4-, 5-, 6-, 8-, 16-, and 32-dim KS
sets. They can be implemented in a laboratory in two different ways.
By means of qubits in an $n$-dim (where $n=2^k$, where $k$ is a
natural number $\ge 2$) Hilbert space
${\cal H}^n={\cal H}^2\otimes\cdots(k)\cdots\otimes{\cal H}^2$
and by means of  spin-{$\frac{n-1}{2}$} systems. The examples of
the former way are KS sets in 4-dim ${\cal H}^4$ by means of 2
qubits from the class 60-105 in Sec.~\ref{sec:60105} and from the
24-24 class \cite{pmm-2-10} in Sec.~\ref{sec:2424}, in 8-dim
${\cal H}^8$ by means of 3 qubits, or in 16- and 32-dim spaces
via 4 and 5 qubits in Secs.~\ref{sec:16dim} and \ref{sec:32dim},
respectively. The examples of the latter way are 4-dim 60-74 class
in Sec.~\ref{sec:6074}, 6-dim star/triangle set and the 236-1216
class in Sec.~\ref{sec:6dim}, and the star/triangle set in
Sec~\ref{sec:8dim}. In our hypergraphs approach, the calculational
treatment of and the elaboration on all classes are the same, though.
Only experimental implementations differ and we will discuss them
when needed.

General formalism of $n$-dim ($n\ge 3;\ n\in\mathbb{N}$)
KS sets and their implementation via spin-{$\frac{n-1}{2}$} particles
(say via, e.g., generalized Stern-Gerlach devices with a simultaneous
usage of magnetic and electric fields by means of which it is possible
to generate an arbitrary spin state \cite{anti-shimony}), covers any
possible experimental implementation in contrast to qubit approach
which covers only $n$-dim$=2^k$-dim cases ($k\in\mathbb{N},\ k\ge 2$). 

We represent KS sets by hypergraphs in the {\em MMP hypergraph\/}
notation specified below. In a KS set, the vectors correspond to
vertices of an MMP hypergraph. Vertices representing $n$-tuples of
orthogonal eigenvectors are organized in edges of MMP hypergraphs
\cite{pmmm04a-arXiv}.

\begin{definition}\label{def:mmp}MMP hypergraphs are hypergraphs in 
which
\begin{enumerate}
\item[(i)] Every vertex belongs to at least one edge;
\item[(ii)] Every edge contains at least 3 vertices;
\item[(iii)] Edges that intersect each other in $n-2$
         vertices contain at least $n$ vertices.
\end{enumerate}
\end{definition}

A KS set with $n$ vertices and $m$ edges is denoted as $n$-$m$.  

Only minimal KS sets, called {\em critical\/} KS set are relevant 
for experimental implementations since their supersets just contain
additional orthogonalities that do not change the KS property of
the smallest critical set. 

\begin{definition}\label{def:crit} 
KS sets that do not properly contain any KS subset, meaning
that if any of its edges were removed, they would stop being KS 
sets, are called {\em critical\/} KS sets. 
\end{definition}

Some authors make use of a coarser notion of
{\em [vertex-]critical\/} KS sets: ``A KS [set] is termed critical
iff it cannot be made smaller by deleting the [vertices]''
\cite{ruuge12}. However, this definition lacks operationality
in identifying a huge number of critical sets which turn into a
non-KS set when an edge of theirs is removed while the number of
vertices remains unaltered as allowed by Def.~\ref{def:crit}. On
the other hand, deleting a vertex means a removal of at least one
edge.  


We encode MMP hypergraphs by means of alphanumeric and other
printable ASCII characters. Each vertex is denoted by one of the
following characters: {{\tt 1 2 \dots 9 A B \dots Z a b 
\dots z ! " \#} {\$} \% \& ' ( ) * - / : ; \textless\ =
\textgreater\ ? @ [ {$\backslash$} ] \^{} \_ {`} {\{}
{\textbar} \} \textasciitilde} \cite{pmm-2-10}. When all these
characters are exhausted we reuse them so as to prefix them by `+',
then by `++', and so on. An example is shown in the graphical
representation of a hypergraph of KS set 18-9 in the figure in
Sec.~\ref{sec:2424}, where ASCII characters printed next to
corresponding vertices from the hypergraph belong the MMP hypergraph
string {\tt 1234,4567,789A,ABCD,DEFG,GHI1,29BI,35CE,68FH.}
So encoded, MMP hypergraphs are generated by our algorithms and
programs or introduced into our programs to be processed. Each
edge is represented by a string of characters separated by
commas and all of them together form a hypergraph, i.e., a KS
set, as a single textual line/string which ends with a full stop.
When dealing with such ASCII line encoding of MMP hypergraphs
we call them MMP hypergraphs lines or strings  when needed.
The order of the strings and characters is irrelevant; gaps
in characters are allowed and its number is not limited; tens
of thousands of them are not a problem for our programs
{\tt shortd.c, mmpstrip.c, subgraph.c, vectorfind.c, states01.c}
and others
\cite{bdm-ndm-mp-1,pmmm05a,pmm-2-10,bdm-ndm-mp-fresl-jmp-10,mfwap-s-11,mp-nm-pka-mw-11}. 

To visualise the hypergraphs we represent them as figures showing
vertices as dots and edges as straight or curved lines each connecting
$n$-tuples of vertices. We often draw hypergraphs so as to start with
the biggest loop they contain. Usually we do not attach characters to
vertices in a figure because one can always arbitrary attach them and
then use program {\tt vectorfind} to ascribe vector components to each
vertex. In chosen figures in the following sections below we show
graphical representations of some of the KS sets that we found in
this study in the MMP hypergraph notation.

Our standard and compact definition of MMP hypergraphs enables
us to smoothly design algorithms for generation, handling, and
analysis of KS sets what together amounts to MMP hypergraph language.
In this work, we generate subgraphs of big chosen KS hypergraphs,
which we call {\em master sets\/}, by deleting a specified number of
edges from such master sets via our program {\tt mmpstrip}. Then we
filter them on the KS property via our program {\tt states01} which
just verifies whether they violate the conditions of the
Def.~\ref{th:ks}, i.e., whether they are KS sets. Program
{\tt states01} carries out an exhaustive search according to
a backtracking algorithm. This is a much less demanding task than a
constructive upward generation we used previously in
\cite{pmmm04a-arXiv} and although we have to deal with a huge volume
of hypergraphs, on computing clusters we can carry out generations
successfully, as we show in subsequent sections below.

A collection of all KS subsets of a particular master set $i-j$,
with $i$ vertices and $j$ edges, we call an $i-j$ {\em class\/} of
KS sets.

We can generate members of an $i-j$ class from a master set $i-j$,
on a computing grid, as follows. First, we strip edges from the
master set with {\tt mmpstrip} and then filter them with {\tt sed},
{\tt states01}, {\tt sed}, and {\tt shortd} to obtain, say, $m$
$k-l$ files, where $l=1,\dots,m$ and $k$ depends on $l$ in a
rather involved manner depending of how many vertices, if any, were
stripped together with stripped edges. Each file might contain
millions of KS sets (all with $l$ edges). Then, we use these files as
input files for the next round of distributed computing to randomly
generate a chosen number of non-isomorphic KS critical sets from
each line, i.e., each single KS set, of the obtained files by means of
{\tt states01} and {\tt shortd}. Thus we obtain arbitrary exhaustively
many $p-q$ KS criticals.

In elaborating on KS sets, we, as well as other authors in the
literature, make use of theory and algorithms from several disciplines:
quantum mechanics, lattice theory, graph theory, and geometry. Each
discipline has its own terminology which are often adopted by the
authors of papers on KS sets. In this context, the terms ``vertex,''
``atom,'' ``ray,'' ``1-dim subspace,'' and ``vector'' are synonymous,
as are the terms ``edge,'' ``block,'' ``context,'' and
``$n$-tuples (of mutually orthogonal vectors).''
Similarly, ``MMP hypergraph'' and ``MMP diagram'' mean the same
thing.

In the literature, KS sets are very often generated and tested
via the so-called {\em parity proof\/}. (The proof was generalised
by Lison\v ek, Raussendorf, and Singh \cite{lisonek-raus-14}.) 

\begin{definition}\label{def:parity} 
{\em Parity proof.} A parity proof of the KS theorem, in en
even-dimensional Hilbert space, via a KS set is a set of $k$ vertices
of the set that form $l$ edges ($l$ odd) such that each vertex shares
an even number of edges. Looking, e.g., at the 18-9 KS set in
Fig.~\ref{fig:24-24-class}, we see that, because each vertex shares
exactly two edges, there should be an even number of edges with 1s.
At the same time, each edge can contain only one 1 by definition, and
since there are an odd number of edges, there should also be an odd
number of edges with 1s, i.e., we have a contradiction. 
\end{definition}

Parity proofs face several problems, though. 

\begin{itemize}
\item KS sets with even number of edges cannot have parity proofs
per definition; 
\item Many KS sets with odd number of edges turn out 
not to have a parity proof, either; 
\item In some classes of KS sets we obtained, less than 0.1\permil\ 
have a parity proof, and in some others, none at all.

\end{itemize}

Parity proofs are just special and particular cases of our general
MMP hypergraph verification but sometimes they turn out to offer a
complementary method of generation of KS sets since-parity-proof
based programs are much faster than general MMP hypergraph based
ones, when applicable. 

\section{\label{sec:2424}Tentative 24-24 Class of 4-dim 
KS Sets; Generation of KS Sets via Stripping 
of Master Sets}

In this section we shall make use of the results about the 24-24 
class of 4-dim KS sets we obtained in 
\cite{mporl02,pmmm05a,pmmm04b,pmm-2-10,pavicic-book-05} to introduce
the main steps and strategy we shall undertake to obtain the
results in the subsequent sections. 

In 2002 Pavi\v ci\'c \cite{mporl02} realised that one can establish 
a correspondence between MMP hypergraphs and systems 
of nonlinear equations describing mutual orthogonalities of vectors 
as, for instance, in the following 3-dim example  
\begin{eqnarray}
{\mathbf x}\cdot{\mathbf y}&=x_1y_1+x_2y_2+x_3y_3\ &=0,\nonumber\\
{\mathbf x}\cdot{\mathbf z}&=x_1z_1+x_2z_2+x_3z_3\ &=0,\nonumber\\
{\mathbf y}\cdot{\mathbf z}&=y_1z_1+y_2z_2+y_3z_3\ &=0. 
\label{eq:mmp-eq}
\end{eqnarray}
The latter system is an unsolvable problem on any supercomputer, even
for the smallest KS sets while ascribing 0-1 valuations, required by
the definition of a KS set given in Theorem \ref{th:ks}, to vertices of
MMP hypergraphs is a problem of statistically polynomial complexity. 
In other words, solving for MMP hypergraphs is exponentially more
computationally efficient than solving for Hilbert space vectors
directly when searching for KS sets. Such a correspondence between
nonlinear systems and MMP hypergraphs enables us to generate KS sets
on a large scale (billions of them). This can be compared with less
than a dozen of KS sets discovered by several researchers between
1967 and the end  of the 20th century
\cite{koch-speck,peres,kern-peres,kern,bub,cabell-est-96a,aravind-600},
mostly exploring highly symmetrical geometrical structures defined 
by mutually orthogonal vectors. 

In 2004 Pavi\v ci\'c, McKay, Merlet, and Megill 
\cite{pmmm05a,pmmm04b} generated non-isomorphic MMP 
hypergraphs and filtered them by means of a program which 
was written for our algorithm of assigning 0s and 1s to their 
vertices and another algorithm for assigning vector components 
to vertices. The generation and assignments are exponentially 
complex tasks in general but applied to our KS MMP hypergraphs    
they turned out to be polynomially complex for the great majority 
of jobs. We say that they are {\em statistically polynomially 
complex\/}. Nevertheless, when we reached 24 vertices, 
the task became forbiddingly CPU-time consuming---we obtained 
over 300 KS sets with up to 23 vertices on a cluster with on 
average 100 CPUs running for several months. Among them there
were only 5 critical KS sets. So, we started to search for
another way of generating KS sets. We arrived at the idea of a 
faster generation as follows. 

Kernaghan \cite{kern} and Cabello, Estebaranz and 
Garc{\'\i}a-Alcaine \cite{cabell-est-96a} realised that their
18-9 and 20-11 KS sets were subsets of Peres' 24-24 set 
\cite{peres} but since they did not make use of graphical 
representation it took them a while to find their two sets and 
neither they nor Peres were able to find any more KS subsets 
the 24-24 set (Peres even wrote a computer program for the 
purpose \cite{peres-book}). 

After we generated the first few hundred KS sets in 
\cite{pmmm05a,pmmm04b} and started to draw their hypergraphs 
we visually recognised---see Fig.~\ref{fig:24-24-class}---that 
they were all subgraphs of the hypergraph we drew for Peres' 
24-24 set. Then Pavi\v ci\'c, Megill, and Merlet designed an 
algorithm for stripping ({\em peeling\/}) edges off the latter 
hypergraph and obtained 1232 KS subsets \cite{pmm-2-10} (including
all 6 criticals from Fig.~\ref{fig:24-24-class}) within less than 2
minutes on a PC. These 1233 KS sets form a 24-24 {\em class\/} of
KS sets and Peres' 24-24 set is their master set. (We would just
like to mention here that we generated and scanned, during 3
CPU-months, {\em all\/} non-isomorphic hypergraphs with 24 vertices
and 24 edges and that among all millions of them there is only one
KS set---Peres' 24-24 one.)

\begin{figure}[hbt]
\includegraphics[width=0.99\textwidth]{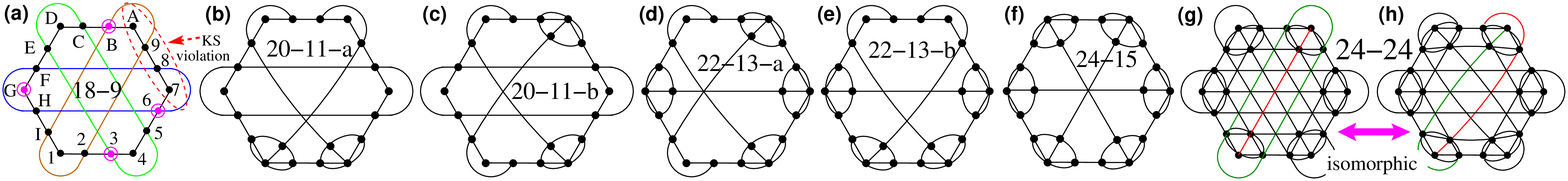}
\caption[4-dim KS Sets]{4-dim KS sets from the 24-24 KS class;
(a)-(f) are all critical KS sets from the class; (a) and (c) were
found in \cite{kern} and \cite{cabell-est-96a}, respectively; (b),
(d), and (e) were found in \cite{pmmm05a}; (f) was found in
\cite{pmm-2-10}; (g) and (h) are two isomorphic representation of
a non-critical Peres' 24-24 KS set found in \cite{peres}---it
contains all (a)-(f) as well as all the other 1226 KS sets from
the 24-24 class; (a) shows that its vertices cannot satisfy the
conditions of Theorem \ref{th:ks}; encircled vertices represent
possible 1-assignments.} 
\label{fig:24-24-class}
\end{figure}

This led us to another aspect of generating KS sets.
All vectors forming KS sets in the 24-24 class have components
from the set \{-1,0,1\} since Peres' 24-24 set has components 
from this class. However, we also found KS sets that were not 
subsets of the 24-24 set (e.g., the 22-11 one, shown in Fig.~3(a) 
of \cite{pmm-2-10}) and those subsets have the components from 
a wider set of values (see Table 2 of \cite{pmm-2-10} for the 
aforementioned 22-11 set). That indicated that there is  
another class or other classes which contain those sets 
or both kinds of sets and we started stripping master sets 
meanwhile discovered 
\cite{aravind-600,aravind10,waeg-aravind-jpa-11,waeg-aravind-megill-pavicic-11,mp-nm-pka-mw-11,waeg-aravind-megill-pavicic-11}. 

We designed algorithms and programs which exhaustively 
generate all KS sets from all stripped subsets of chosen master KS sets 
that we introduced and described in Sec.~\ref{sec:for}.
They are computationally rather demanding and require many CPU months 
of running on clusters and supercomputers but that is feasible 
with today's resources. In the rest of the paper we present various 
outcomes of such calculations with our algorithms and the features 
of the critical KS sets we obtained on our clusters.

\section{\label{sec:6074}Tentative 60-74 Class of 
4-dim KS Sets}

Waegell and Aravind have derived a 60-75 KS set from a 4-dim regular 
polytope (600-cell) with 60 pairs of vertices \cite{aravind10}. 
The vertices correspond to vectors whose 
components have values from the set 
${\cal V}=\{0,\pm(\sqrt{5}-1)/2,\pm 1,\pm(\sqrt{5}+1)/2,2\}$ and one 
can use them to write down the 60-75 set \cite[Table 2]{aravind10}.
MMP hypergraph of the 60-75 generated in \cite{mp-nm-pka-mw-11} is
given in Appendix \ref{app:0-1}.

Generation of smaller KS sets from the master sets will be carried
out by relying on the MMP hypergraph structure only and the vertices
of the obtained set can be ascribed values from $\cal{V}$ 
later on, if needed, via (a) our program {\tt vectorfind} randomly,
or (b) via our program {\tt subgraph} so as to trace down vertices
which survived stripping of edges. We need to ascribe values from
$\cal V$ to the vertices, e.g., for an experiment
(cf~\cite[Fig.~1]{waeg-aravind-megill-pavicic-11}). 

In 2011, Megill, Fresl, Waegell, Aravind, and Pavi\v ci\'c 
\cite{mfwap-s-11} presented preliminary and partial results of 
generating subsets of the 60-75 set by stripping it of its edges 
and obtaining their features. Here we present in many respects an 
almost exhaustive analysis of these subsets. 

We start by stripping just one edge at a time of the 60-75 set in
75 different ways so as to obtain seventy five 60-74 sets. To be more
explicit, we remove one edge from the 60-75 set to get the 1st 60-74,
then we put it back and remove another edge to the 2nd 60-74 and so
forth. It turns out that all 75 of the so obtained 60-74 sets are isomorphic
to each other and that they all reduce to a single MMP hypergraph
string 60-74 given in Appendix \ref{app:0-1}.

\parindent=20pt
We shall therefore consider this 60-74 KS set to be a master set for 
all smaller KS sets we obtain from it. Therefore, we shall call the 
collection not a 60-75 but a 60-74 class of 4-dim KS sets. The 
number of sets from the class we generated and analysed by running 
our programs over a century of CPU time on our clusters are given 
in Fig.~\ref{table:60-74}. The stripping technique applied to  
the sets means a removal of one edge at the time and filtering
out the KS sets with the help of several additional algorithms and
programs.

\begin{figure}[hbt]
\includegraphics[width=0.99\textwidth]{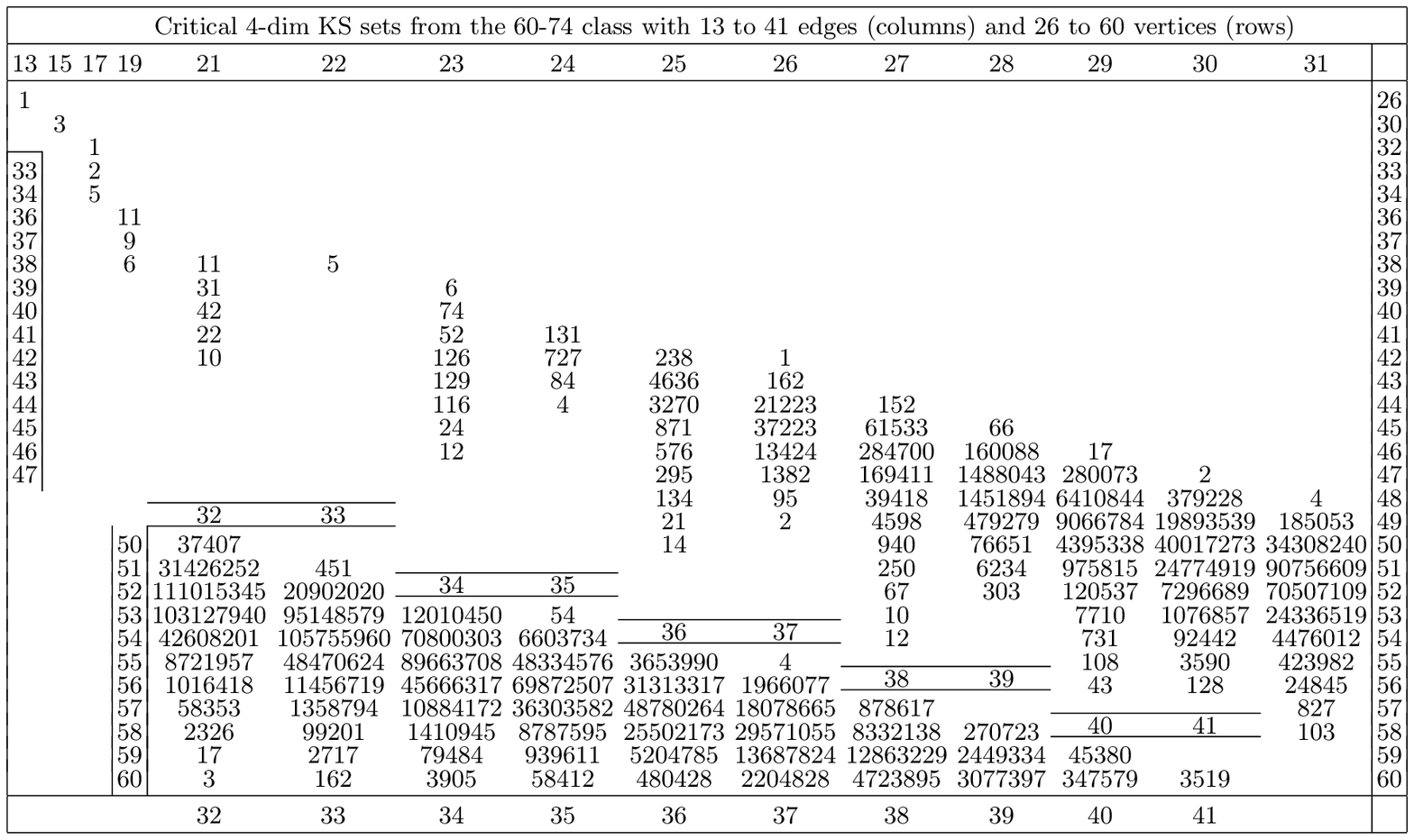}
\caption{List of 1,540,184,852 non-isomorphic KS critical sets from the 
60-74 class we obtained on our cluster. We conjecture that all possible 
types of vertex-edge sets are given here, i.e., that an exhaustive 
generation would not provide us with any new type. We also conjecture 
that an exhaustive generation might give up to about an order of 
magnitude more samples of these sets. We obtained no critical sets with
27, 28, 29, 31, or 35 vertices.}
\label{table:60-74}
\end{figure}

In \cite{mfwap-s-11} we obtained only about 8,000 KS sets and 
many were missing. Here we have $1.54\times 10^9$ sets and among them 
all types of sets that were obtained by means of much faster parity 
proofs and which were missing in \cite[Table 1]{mfwap-s-11} (denoted
there by $\otimes$). We also obtained new types of KS sets with both 
even (mostly) and odd number (23) of edges that we did not obtain in 
\cite{mfwap-s-11}, in particular:
38{\bf -22}, 39{\bf -23}, 41,43,44{\bf -24}, 
42,\dots,44,46,\dots,49{\bf -26}, 45,50,\dots,52{\bf -28}, 
47,48,54,\dots,56{\bf -30}, 50,56,\dots,60{\bf -32}, 60{\bf -34}.

Our aforementioned conjecture that the table in Fig.~\ref{table:60-74}
shows all the types of KS criticals from the 60-74 class is based
on the following statistics. The table now shows $1.54\times 10^9$
KS criticals. The last new type, 47-30, started to appear after we
reached $1.07\times 10^9$ sets; before that, 59-32 after
$5.5\times 10^8$, 55-37 after $3.37\times 10^8$, and all the other
150 types were already appearing within $2.15\times 10^8$ generated
sets. Here we stress that our method of generating sets is as
random as a program can possibly be and that therefore the ``late''
appearance of the aforementioned three types is due only to their
very low occurrence among the sets, i.e., to a minuscule
probability to appear at all. 

This can be well illustrated by looking at the KS criticals with 
parity proofs. Among all $1.5\times 10^9$ criticals only
$1.2\times 10^5$ have parity proofs and among them some are still
missing. In particular, we have $3\times 10^6$ 60-39 criticals and
$3.5\times 10^3$ 60-41 criticals and none of them has a parity proof
although there are at least two (60-39 and 60-41 whose MMP hypergraph
strings are given in Appendix \ref{app:0-1}) that do have such a
proof which we obtained by means of a parity-proof program in
\cite{waeg-aravind-megill-pavicic-11}. The strings are presented with
their maximal loops, hexadecagon and heptadecagon (first 16 and 17
edges up to ``,,,''), respectively, to facilitate graphical
representation. In Fig.~\ref{fig:ks-60-74-26-38}, 60-41 is drawn
(vertex ``2'' is indicated and and other vertices from the loop
follow anti-clockwise) and we can see that there are 22 encircled
(in red online) vertices that share four edges and, of course
(otherwise we would not have a parity proof), not a single one of
which would share three edges. The probability that a randomly
generated hypergraph has such a structure is extremely low and this
explains why we did not get them even after more than $10^{15}$ runs.

We used a procedure that strips one edge at a time of smaller and 
smaller sets and simultaneously checks them on KS property, KS 
criticality, maximal loops, number of iterations, level of 
classical non-contextuality of each set, etc. A choice of 
them is represented graphically by means of MMP hypergraphs in 
Fig.~\ref{fig:ks-60-74-26-38}.  

\begin{figure}[hbt]
\begin{center}
\includegraphics[width=0.99\textwidth]{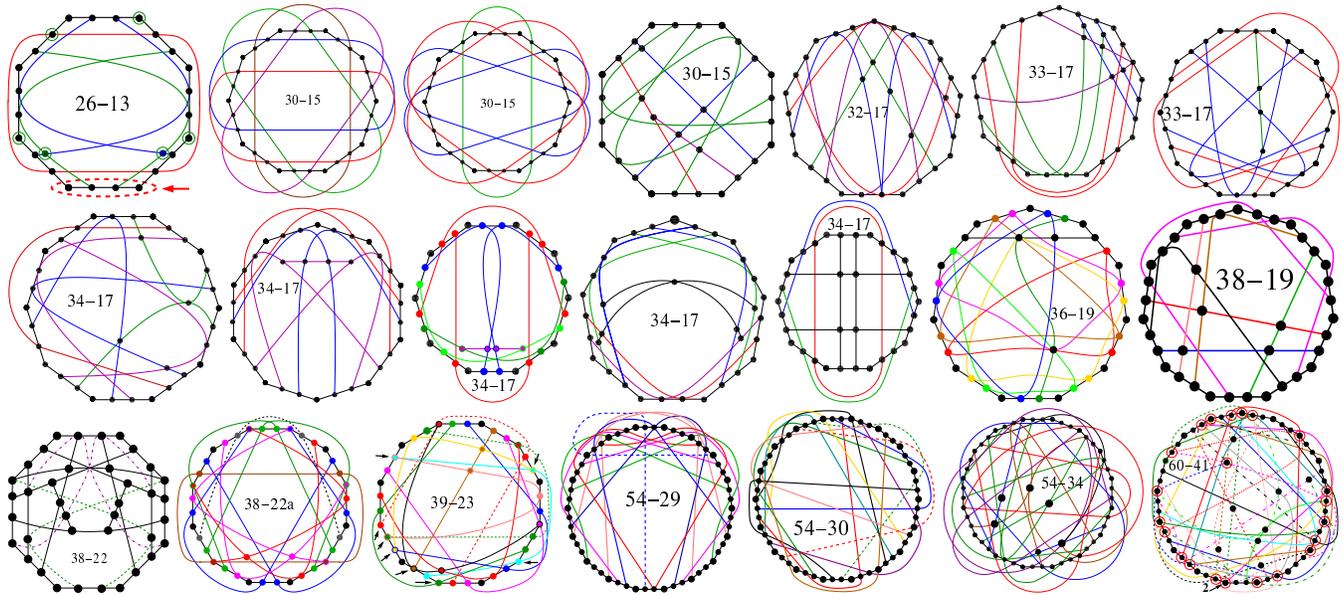}
\end{center}
\caption[MMP Hypergraphs Kochen-Specker Sets from the 60-74 Class.]{
MMP hypergraphs from the 60-74 class shown with the help of their 
maximal loops; 26-13 is the smallest set from the class---the 
arrow points at a ``graphical proof'' of contextuality (all zeros, 
while rings (green online) denote ``1''); 26-13 through 36-19 all
have parity proofs; the first two 30-15 and the last 34-17 have two
axes of symmetry; three middle 34-17, one axis; 38-22 are the 
smallest sets that have even number of edges; 39-23 is the smallest 
set with an odd number of edges which does not have a parity proof;
54-29 and 54-30 are the two smallest sets with the biggest loops 
(18-gon); 54-34 is a typical large set; 60-41 belongs to the largest 
sets of criticals; it does have a parity proof, while other  
60-41 criticals do not have it (see text).}
\label{fig:ks-60-74-26-38}
\end{figure}

The KS criticals 26-13 to 36-19 all have parity proofs and among 
the sets with up to 38 vertices and odd number of edges there is no 
one which fails the parity proof. The first sets with odd number of 
edges without parity proofs are 39-23 sets. One of them is shown in
Fig.~\ref{fig:ks-60-74-26-38} in which arrows point to vertices that 
share an odd number of edges and therefore violate the parity proof 
condition from Def.~\ref{def:parity}. Actually, none of the 39-23 
sets satisfy the parity proofs and this is the reason why this type of
sets is missing in Table 1 of \cite{waeg-aravind-megill-pavicic-11}.

None of the sets with even number of edges can have a parity proof
per definition. Two of the smallest such sets are 38-22 and 38-22a 
shown in Fig.~\ref{fig:ks-60-74-26-38}. 

Two of the smallest sets with the biggest maximal loops in the class, 
octadecagon, are 54-29 and 54-30. They show an interesting property 
of having all vertices contained in the maximal loop like the smallest
 sets 26-13 and 30-15. Set 54-29 does not have a parity proof because 
it contains vertices that share three edges. It also has a property 
that some of its vertices share only one edge which most smaller 
set do not posses.

As we can see from Fig.~\ref{fig:ks-60-74-26-38} the maximal loops
range from octagon (26-13) to octadecagon (54-29) in contrast to the
sets from the sets from the 24-24 class in Fig.~\ref{fig:24-24-class}.
On the other hand, the majority of sets from the 24-24 class have
edges which intersect each other at more than one vertex, while in the
vast 60-74 class there is not a single such set. It follows that not
only the two classes are disjoint but that is also unlikely that they
would belong to a wider class which would contain them both. 

However, there is a class which contains the 24-24
class---the 60-105 one, which we present in the next section. 

\section{\label{sec:60105}60-105 Class of 4-dim 
  KS Sets Defined by Hilbert Space Operators and
Properly Containing 24-24 Class}

When we envisage an application of KS sets in the field of 
quantum computation and communication, a qubit implementation 
comes forward as most interesting. And while the real vectors 
of the KS sets from the 24-24 class do enable a qubit 
representation, as recent experiments have shown, it is not 
clear whether the vector components of the real vectors 
defining the 60-74 class offer us a qubit representation.
Recall that the dimension of the Hilbert space of a quantum 
system and the spin of this system satisfy $\dim{\mathcal{H}}_s=2s+1$.
So, a 4-dim KS set can be realised either via an $s=3/2$ 
particle, say by means of a Stern-Gerlach device, or via
two qubits: $\dim({\mathcal{H}}^2\otimes{\mathcal{H}}^2)=2^2=4$

In order to achieve a qubit representation in the complex 4-dim 
Hilbert space, by means of complex vectors, Aravind and Waegell 
\cite{waeg-aravind-jpa-11} made use of Pauli operators (e.g., 
$\sigma_x^{(1)},\sigma_y^{(2)}$), where the superscripts refer to 
one of two qubits. In a 4-dim Hilbert space they form 
9 mutual tensor products and 6 tensor products with the unit 
vectors. Altogether, these 4-dim operators form 15 commuting 
triplets each of which has four eigenvectors (tetrads) in common. 
There are 60 different eigenvectors that form the resulting 105
tetrads as given in Tables 1 and 2 of \cite{waeg-aravind-jpa-11}.
Their components take values from the set $\{0,\pm 1, \pm i\}$.
A few lines of the former Table are given in Table 
\ref{T:pauli-op-60-105}, below. 

\begin{table}[ht]
\begin{center}
\setlength{\tabcolsep}{4.1pt}
\begin{tabular}{|c||c|c|c|c|}
\hline
Pauli product triples&\multicolumn{4}{|c|}{4 eigenvectors of each 
product from the triple}\\
\hline
$\sigma_x^{(1)}\otimes I^{(2)},\ I^{(1)}\otimes \sigma_z^{(2)},\ \sigma_z^{(1)}\otimes \sigma_z^{(2)}$
&\ \ \ $|1000\rangle$\ \ \ &$|0100\rangle$&
$|0010\rangle$&$|0001\rangle$\\
\hline
$\sigma_x^{(1)}\otimes I^{(2)},\ I^{(1)}\otimes \sigma_x^{(2)},\ \sigma_x^{(1)}\otimes \sigma_x^{(2)}$
&$|1111\rangle$ &$|1-\!11-\!1\rangle$ &
$|11-\!1-\!1\rangle$&$|1-\!1-\!11\rangle$\\
\hline
\dots, \dots, \dots 
&\dots&\dots &\dots&\dots\\
\hline
$\sigma_y^{(1)}\otimes I^{(2)},\ I^{(1)}\otimes \sigma_z^{(2)},\ \sigma_y^{(1)}\otimes \sigma_z^{(2)}$
&\ \ \ $|10i0\rangle$\ \ \ &$|010i\rangle$&
$|10i0\rangle$&$|010i\rangle$\\
\hline
\dots, \dots, \dots 
&\dots&\dots &\dots&\dots\\
\hline
$\sigma_x^{(1)}\otimes \sigma_y^{(2)},\ \sigma_y^{(1)}\otimes \sigma_x^{(2)},\ \sigma_z^{(1)}\otimes \sigma_z^{(2)}$
&$|100i\rangle$&$|01\!-\! i0\rangle$&
$|01i0\rangle$&$|100\!-\! i\rangle$\\
\hline
\dots, \dots, \dots 
&\dots&\dots &\dots&\dots\\
\hline
\end{tabular}
\end{center}
\caption[Products of Pauli operators and their eigenvectors]{A 
sample from a complete list of 15 Pauli operator products and their 
eigenvectors given in Ref.~\cite{waeg-aravind-jpa-11}.}
\label{T:pauli-op-60-105}
\end{table}

The latter table represents a 60-105 master set. Its MMP hypergraph
string is given in Appendix \ref{app:0-2}. By removing one of 105
edges from the string at a time, each time a different one, we
obtain 105 sets. They all turn out to belong to two non-isomorphic
non-critical KS sets in contrast to the 60-75 set which reduces to
the unique 60-74 one. By applying the same technique as in
Sec.~\ref{sec:6074} we generate critical KS sets listed in the table
in Fig.~\ref{table:105}. They make the 60-105 class of KS sets. 

\begin{figure}[hbt]
\includegraphics[width=0.99\textwidth]{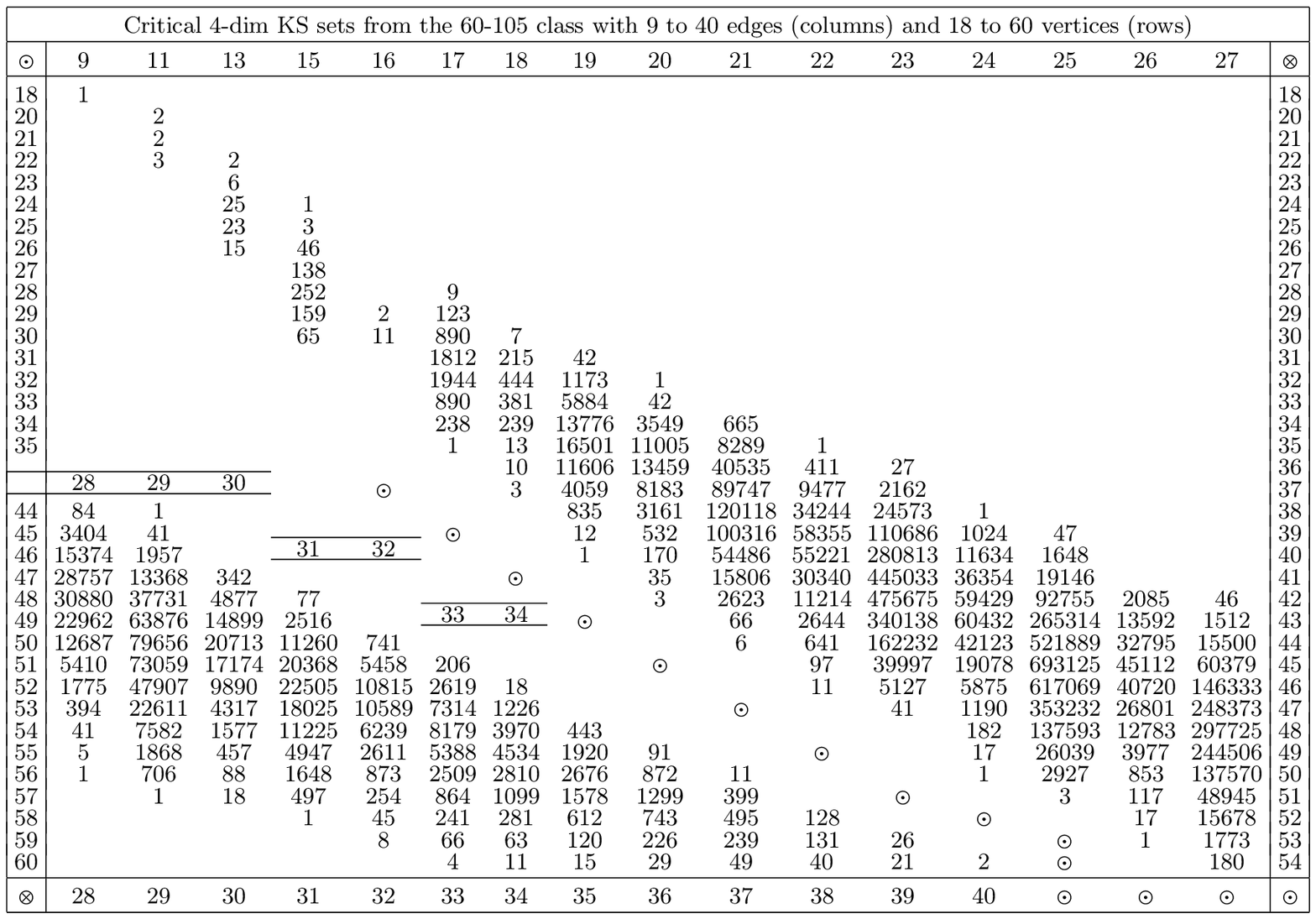}
\caption{List of 7,720,539 non-isomorphic KS critical sets from the 
60-105 class we obtained on our cluster. We conjecture that all possible
types of vertex-edge sets are given here. We obtained no critical
sets with 10, 12, or 14 vertices}
\label{table:105}
\end{figure}

Although the generated critical KS sets from the 60-105 class are 
more than two orders of magnitude less numerous than the ones from 
the 60-74 class, the statistics indicate that the majority of types 
has been generated. 

MMP hypergraph of the master set 60-105 properly contains all MMP
hypergraphs from the 24-24 class \cite{waeg-aravind-jpa-11} and also
the ones we obtained by means of our down-up generation in
\cite{pmmm05a,pmm-2-10} but which did not belong to the 24-24 class
as well as new ones, which do not belong to either of those two kinds,
shown in Fig.~\ref{fig:ks-60-105-21-24}. That is why we called
24-24 class {\em tentative\/} in the title of Sec.~\ref{sec:2424}. 

\begin{figure}[hbt]
\begin{center}
\includegraphics[width=0.99\textwidth]{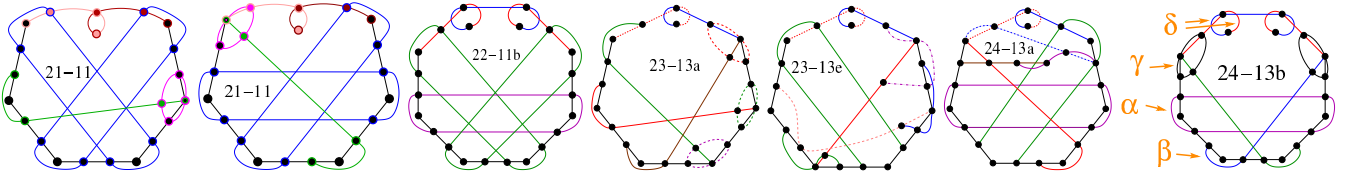}
\end{center}
\caption[MMP Hypergraphs Kochen-Specker Sets from the 60-105 Class.]{
MMP hypergraphs of KS critical sets from the 60-105 class with up to 
24 vertices that are not isomorphic to the ones shown in 
Fig.~\ref{fig:24-24-class}. They share edges of the form 
$\alpha$, $\beta$, and $\gamma$ which characterise 24-24 sets 
but not, e.g., the one of the form $\delta$, which is specific to 
the 60-105 sets. Maximal loops of the criticals shown here range
from hexagons to octagons.}
\label{fig:ks-60-105-21-24}
\end{figure}

Still, with respect to vector representation, the 24-24 class is
not uniquely determined by  the coordinatization of the 60-105
master set. The vector components of the 60-105 set are complex
(taking values from the set \{0, $\pm$1, $\pm i$\}) and Peres'
24-24 master set can take over them directly as shown in Appendix
\ref{app:0-2}. 

But, as we mentioned above, for the master set 24-24 and therefore
all of its subsets there exist real coordinatizations, e.g., the one
originally found by Peres, and that is what Waegell and Aravind meant
when they said that ``60-105 system contain[ed] (in ten different
ways) 24-24 systems of rays and bases used by Peres and others''
\cite{waeg-aravind-jpa-11}.

The fact that the 24-24 class can have both real and
complex coordinatization depend on particular structure of its
sets. In contrast, the systems 21-11 shown in
Fig.~\ref{fig:ks-60-105-21-24} do not possess real
coordinatizations, apparently due to the
$\delta$-feature of their structure---see below. 

Another example of different coordinatizations within the
60-105 class is Pavi{\v c}i{\'c}, Merlet,  Mc{K}ay, and Megill's
20-11a \cite{pmmm05a}, shown in Fig.~\ref{fig:24-24-class}. Its
60-105 coordinatizations might be complex as given in Appendix
\ref{app:0-2} as well as real. If we compared the components with
those of the 24-24 set, we would see that 20-11a might be generated
(stripped) directly from the 24-24. The 20-11a also possesses real
coordinatizations, though, one of which is given in \cite{pmmm05a}.

On the other hand, Cabello, Estebaranz, and Garc{\'\i}a-Alcaine's
18-9 \cite{cabell-est-96a} and Kernaghan's 20-11b \cite{kern} 
(both shown in Fig.~\ref{fig:24-24-class}) have real
coordinatizations with components from $\{0,\pm 1\}$ in 60-105 as
given in Appendix \ref{app:0-2}. By comparing their components we
can see that they are not generated directly from the presented
24-24 set with the given complex coordinatization (it does not have
enough real components) but from some other subsets of 60-105. 

Of course, here we should pose a question on whether there is an even 
wider class which properly contains the criticals from the 60-105 
class and this is an open question. Since the biggest such criticals 
contain only 60 vertices such a bigger class might exist (the master
set from Sec.~\ref{sec:300675} has 300 vertices). However, a wider
class which would properly contain both 60-74 and 60-105 might not 
exist, since these two classes have too disparate properties. First, 
not a single critical KS set from the 60-105 class is isomorphic to 
any of $1.5\times 10^9$ critical KS sets from the 60-74 class. 
Second, there is an important structural difference together with 
all similarities. 

The similarities are of the  $\alpha$, $\beta$, and $\gamma$ kind 
shown at 24-13b and 24-13c in Fig.~\ref{fig:ks-60-105-21-24}.
$\alpha$ is an edge whose vertices  each share a single edge from
the maximal loop; $\beta$ consists of 2 such vertices, 1 which
shares two loop edges and a third edge and 1 which shares only
that third age; $\gamma$ is the third edge from the previous
$\beta$ definition. 

A definite difference with and a dominant feature of 65-105 sets
is the $\delta$-feature (see Fig.~\ref{fig:ks-60-105-21-24}).
It refers to two neighbouring edges from the maximal loop exclusively
sharing two vertices, i.e., intersecting each other at two vertices
which do not share any third edge. The $\delta$-feature characterises
most of the criticals shown in Figs.~\ref{fig:ks-60-105-21-24},
\ref{fig:ks-60-105-29-30}, and \ref{fig:ks-60-105-big}. It might
correspond to a rank-2 projector and be related to the fact that
in a KS test one need not distinguish which of the
two vertices that share two edges was assigned a 1. The role of
projectors of a higher rank in a description of KS sets has been
explored by Waegell and Aravind in details in
\cite{waeg-aravind-jpa-11}.

The portion of sets from the 60-105 class with an odd number of edges 
which possess the parity proofs and the overall number of sets from
the class with the parity proofs is much higher than in the 60-74 
class. Of $7.5\times 10^6$ 60-105 criticals, we obtained,   
$5.72\times 10^6$ have parity proofs, i.e. 76.3\%. The latter number 
includes 6 criticals from the former 24-24 class which all have 
parity proofs. There are 132 types of KS criticals with an odd number
of edges of which 45 were previously reported by Waegell and Aravind
\cite{waeg-aravind-jpa-11} and additional 12 by Pavi{\v c}i{\'c}
\cite{pavicic-book-13} and 111 with an even number of edges of which
22 were previously found by Pavi{\v c}i{\'c} \cite{pavicic-book-13}.

A general feature of all classes is that smaller sets have only odd
number of edges and that they all have parity 
proofs. On the other hand, among large sets with odd number of 
edges there are only a very few ones with the parity proofs. 
As we saw in Sec.~\ref{sec:6074} we did not obtain a single 
such 60-39 or 60-41 set in the 60-74 class although they exist 
(and are given above) and of 21 60-39 sets in the 60-105 class
no one has a parity proof and, to our knowledge, it is not known 
whether such a set exists. One of 11 smallest sets without a
parity proof is the 26-15 shown in Fig.~\ref{fig:ks-60-105-29-30}.
Encircled vertices (in red online) do not satisfy the parity proof
condition; they do not share an even number of edges. 

\begin{figure}[hbt]
\begin{center} 
\includegraphics[width=0.99\textwidth]{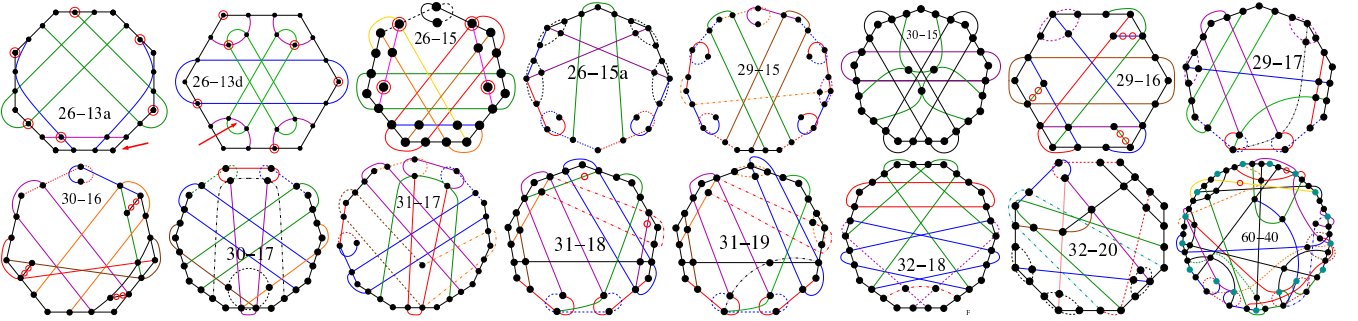}
\end{center}
\caption[MMP Hypergraphs Kochen-Specker Sets from the 60-105 Class.]{
Bigger 60-105-class criticals; sets with even number of edges (no parity) 
compared with sets with odd number of edges (with parity) of the same 
vertex size; arrows indicate edges at which conditions (1) and (2)
of the KS theorem are violated and the theorem proved; rings
(red online) in 29-16, 30-16, 31-18, and 60-40 denote vertices that
share just one edge; 26-15 is one of the smallest sets without parity
proofs; 60-40 is one of the biggest criticals; its maximal loop forms
a heptadecagon (17-gon).}
\label{fig:ks-60-105-29-30}
\end{figure}

The smallest sets with even number of edges are 29-16. In 
Fig.~\ref{fig:ks-60-105-29-30} a sample of them is shown with 
vertices which share only one edge drawn as rings (red online). As
the number of vertices and edges increase there are fewer and fewer
such vertices which are dominant among 3-dim KS criticals (see
Sec.~\ref{sec:3d}). Yet, there is one of them in the 60-40 set.  

Our generation of KS sets via stripping of master sets was so far 
completely random. As we already stressed this does require a 
considerable amount of CPU time. What slows down the generation 
is not the stripping itself, which is extremely fast, but 
filtering on the KS property and criticality. Algorithms which 
would be focused on particular arrangement of vertices and edges 
might prove more efficient and even serve us to obtain KS sets 
without previous stripping from any master set. Possible
arrangements of such a kind are the ones which would have all 
vertices contained within a single loop as 26-13 to 46-23 or 
nearly so as 50-25 and 54-27 in Fig.~\ref{fig:ks-60-105-big}.

\begin{figure}[hbt]
\begin{center}
\includegraphics[width=0.99\textwidth]{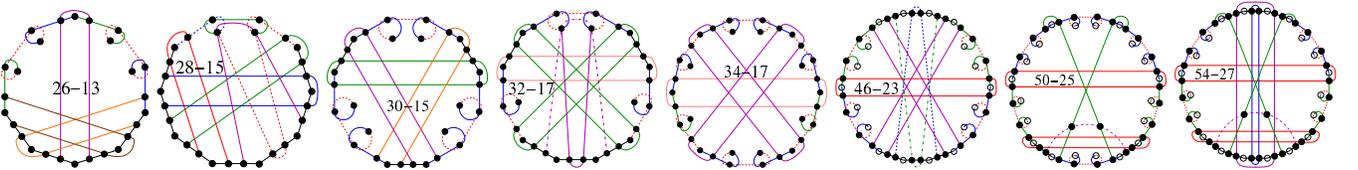}
\end{center} 
\caption[MMP of Hypergraphs Kochen-Specker Sets from the 60-105 Class.]{
26-13 to 46-23 samples of 60-105 criticals with all the 
vertices contained in the maximal loop; there are no such sets 
with 50 or more vertices; 50-25 and 54-27 are the closest 
structures; rings in 46-23 to 54-27 denote the end vertices of 
edges.}
\label{fig:ks-60-105-big}
\end{figure}

\section{\label{sec:300675}300-675 Class of 
4-dim KS Sets Containing 60-74 Class}

In 2011 Waegell, Aravind, Megill, and Pavi{\v c}i{\'c} made use of
600-cell convex regular 4-polytope to obtain a 60-75 master KS set and
a huge number of KS criticals which we call the 60-74 KS class
\cite{waeg-aravind-megill-pavicic-11}. Three years later, in 2014,
Waegell and Aravind considered its dual 120-cell and obtained a
300-675 master set and from it a number of different KS sets via
parity proofs \cite{waeg-aravind-fp-14}. In particular, using parity
proof algorithms and programs, they found the following 102 types of
critical KS sets from their 300-675 master sets: 38-{\bf 19},
42-{\bf 21}, 44$\cdots$46-{\bf 23}, 48$\cdots$50-{\bf 25},
50$\cdots$54-{\bf 27}, 52$\cdots$58-{\bf 29}, 54$\cdots$62-{\bf 31},
56$\cdots$66-{\bf 33}, 58$\cdots$70-{\bf 35}, 60$\cdots$74-{\bf 37},
53$\cdots$78-{\bf 39}, and 65$\cdots$82-{\bf 41}. We show MMP
hypergraphs for some of them (38-19, 42-21, 48-25) in
Fig.~\ref{fig:ks-300-675}.

\begin{figure}[hbt]
\begin{center}
  \includegraphics[width=0.99\textwidth]{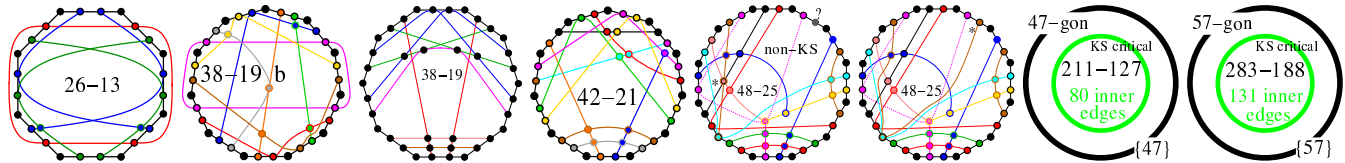}
\end{center} 
\caption[MMP Hypergraphs of Kochen-Specker Sets from the 300-675 Class.]{
  KS criticals from the 300-675 class together with one non-KS set (see
  text); KS criticals from the higher vertex-edge group (211 to 283
  vertices and 127 to 188 edges) are represented by circles since
  the vertices and edges are too numerous to be discernible in a
  figure---they are all listed in Fig.~\ref{fig:table-4dim-300-675}, though;
  black circles represent maximal loops with 47 
  (Schl{\"a}fli symbol \{47\}) and 57 edges (\{57\}).}
\label{fig:ks-300-675}
\end{figure}

Among the smallest KS criticals we generated from the 300-675 master
set are one 26-13 (shown in Fig.~\ref{fig:ks-300-675}), two
(non-isomorphic) 30-15, one 32-17, one 33-17, four 34-17, two 38-19
(one of them, 38-19b, is shown in Fig.~\ref{fig:ks-300-675}), 
one 43-24, and one 44-26. Apart from the last two, all of them have
parity proofs. These KS criticals with parity proofs are subgraphs
of the master set 60-74 and therefore belong to the 60-74 class. 

Waegell and Aravind actually show in \cite{waeg-aravind-fp-14}, by the
very construction of the 300-675 master set, that the 60-75 master set
is properly contained in it, i.e., that the hypergraph of
the 60-75 master sets is contained in the hypergraph of the
300-675 master set.

Next, in Fig.~\ref{fig:ks-300-675}, we present three hypergraph
MMP representations of the KS criticals obtained in
\cite{waeg-aravind-fp-14}: 38-19 (max loop: 10-gon),
42-21 (11-gon), and 48-25 (12-gon). Their MMP hypergraphs are
given in Appendix \ref{app:0-3}. Our program {\tt subgraph} shows
that these KS criticals are not subgraphs of the master set 60-74
(or 60-75) and that, therefore, the class 60-74 does not contain
them. 

Here we use the opportunity to show yet another advantage of
the hypergraph approach to KS sets. Waegell and Aravind made a
misprint somewhere in their Table 7 \cite{waeg-aravind-fp-14} which
should have defined their 48-25 set but an automated translation
gives a hypergraph denoted ``non-KS 42-21'' in
Fig.~\ref{fig:ks-300-675}. To find the misprint in their list of
vertices and edges one should invest a considerable amount of
time and most likely they themselves as well. However, in our
representation it is immediately visually apparent that in a
parity proof the vertex can neither share three
edges, denoted by ``*'', nor just one, denoted by ``?''.
Therefore we can easily amend the misprint by disconnecting the
brown edge from the *-vertex and extending it to the ?-vertex,
so as to obtain the 48-25 KS critical set shown as the next set
in the figure; its ASCII MMP representation is given above. It
provably does not belong to the 60-74 class.

Further advantages are obvious from the generation of a cluster 
of unprecedentedly big KS criticals indicated in the last two
figures in Fig.~\ref{fig:ks-300-675} and listed in all details
in Fig.~\ref{fig:table-4dim-300-675}. The generation of KS criticals
in the 300-675 class is an extremely demanding task due to the
intricacy of the master set 300-675 itself which stems from the
high number of vertices. If we strip too many loops in the
first step with {\tt mmpstrip} we shall find ourselves in the
non-KS desert, i.e., the probability of finding a KS set will
be too small. If we strip only, say, 500 loops, from the master
set, verification of whether a single obtained MMP hypergraph is
a KS set and if it is to reduce it to a critical KS set will
take between one and three CPU-months (3 GHz). At the first glance
it might look as if Waegell and Aravind also stumbled upon this
problem of intricately interwoven edges and orthogonalities:
``we have not found any [set] with more than 41 bases, but we
cannot be sure about the upper limit because our searches have
been limited to only the reduced sets in Table 4''
\cite[p.~1093, bottom]{waeg-aravind-fp-14}.

However, they actually could not have found them because with
their parity-proof-based programs they could not have seen
them at all. More precisely, 14\%\ of all criticals in the 300-675
class have a parity proof but the probability KS sets having them is
not uniformly distributed throughout the class. All KS sets with
parity proofs are in the bottom part of the class. KS sets from the
top part of the class do not have parity proofs---none of the 221-127
to 283-188 generated KS criticals has a parity proof. Hence, they
are invisible for parity-proof-based algorithms and programs and since
search algorithms in the literature rely almost exclusively on parity
proofs we give an MMP representation of the 221-127 critical KS
set (47-gon) in Appendix \ref{app:0-3}.

\begin{figure}[hbt]
\begin{center}
  \includegraphics[width=0.95\textwidth]{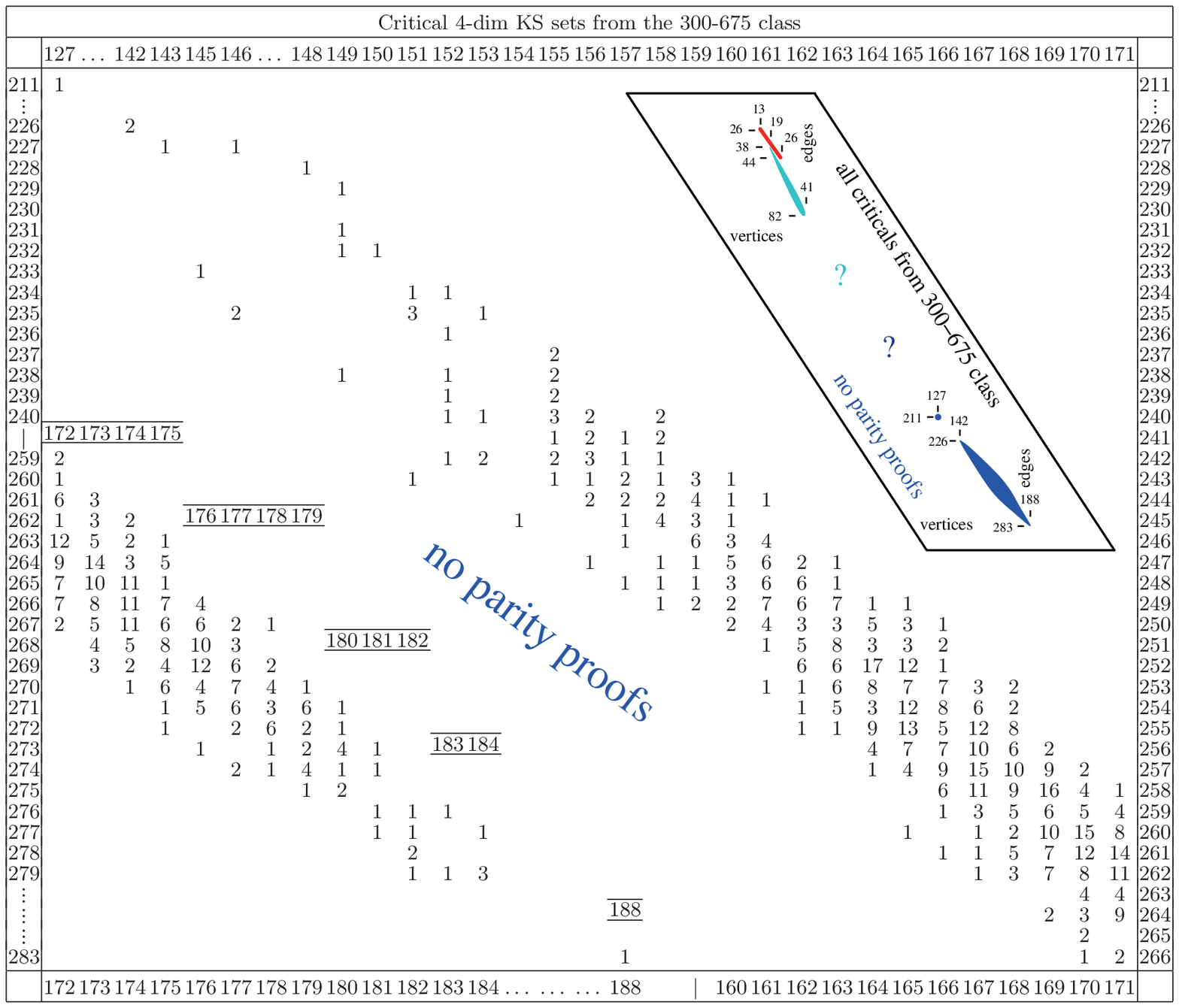}
\end{center} 
\caption{Critical 4-dim KS sets from the 300-675 class;
  the very table shows only big KS criticals with 127 to 188 edges
  (columns) and 211 to 283 vertices (rows); the inset shows all
  critical KS sets from the 300-675 class: at the bottom right
  (blue online) are all criticals from the outer table, at the top
  left, upper line (red online), are 26-13,\dots,44-26 we obtained
  (belonging to the 60-74 class as well; not shown in the table),
  and also at the top left, lower line (cyan online), are
  28-19,\dots,82-41 obtained by Waegell and Aravind
  \cite{waeg-aravind-fp-14} (equally not shown in the table itself).}
\label{fig:table-4dim-300-675}
\end{figure}

Higher criticals from the class 300-675 we obtained and presented
in Fig.~\ref{fig:table-4dim-300-675} are far less numerous than KS
criticals from any other class we presented in this paper. This is,
however, not due to a small number of sets in the class---their
overall number is according to our tests staggeringly huge; this is
due to the fact that their generation is computationally extremely
demanding and time consuming. Therefore we generated the sets in
stages and subjected them to several levels of filtering. We first
randomly stripped 400 to 550 edges from the master set 300-675
by {\tt mmpstrip} thus obtaining 150 groups of sets with 275 down
to 125 edges. Then we filtered these sets for the KS property
and randomly reduced them to criticals by means of {\tt states01}.
This procedure takes up to three CPU months for each single 
critical. We generated higher criticals from the KS noncritical
sets in the range from 190 to 275 edges. E.g., the 211-127 KS
critical we obtained from a set with 190 edges. For sets with less
than 190 edges we observed a sudden drop to criticals with up to
40 edges. 

\section{\label{sec:4dim-witt}148-265 Class of 4-dim KS Sets}

In January 2017 Waegell and Aravind \cite{waeg-aravind-17-arXiv}
showed that that the Penrose dodecahedron, Zimba and Penrose used to
construct their 40-40 non-critical KS set \cite{zimba-penrose}, can
be extracted from the {\em Witting polytope} in
$\mathbb{C}^{4}$. Actually Waegell and Aravind consider a
148-265 KS master sets and their subsets, 40-40 being one of
them. Since this is a work in progress, we shall not go into
details but will only list the types of KS criticals the master
set 148-265 can be reduced to and give two of their hypergraphs,
in Fig.~\ref{fig:ks-148-265}, so as to round up our presentation
of generation of KS criticals from all known 4-dim KS master sets
and their classes. 

\begin{figure}[hbt]
\begin{center}
  \includegraphics[width=0.6\textwidth]{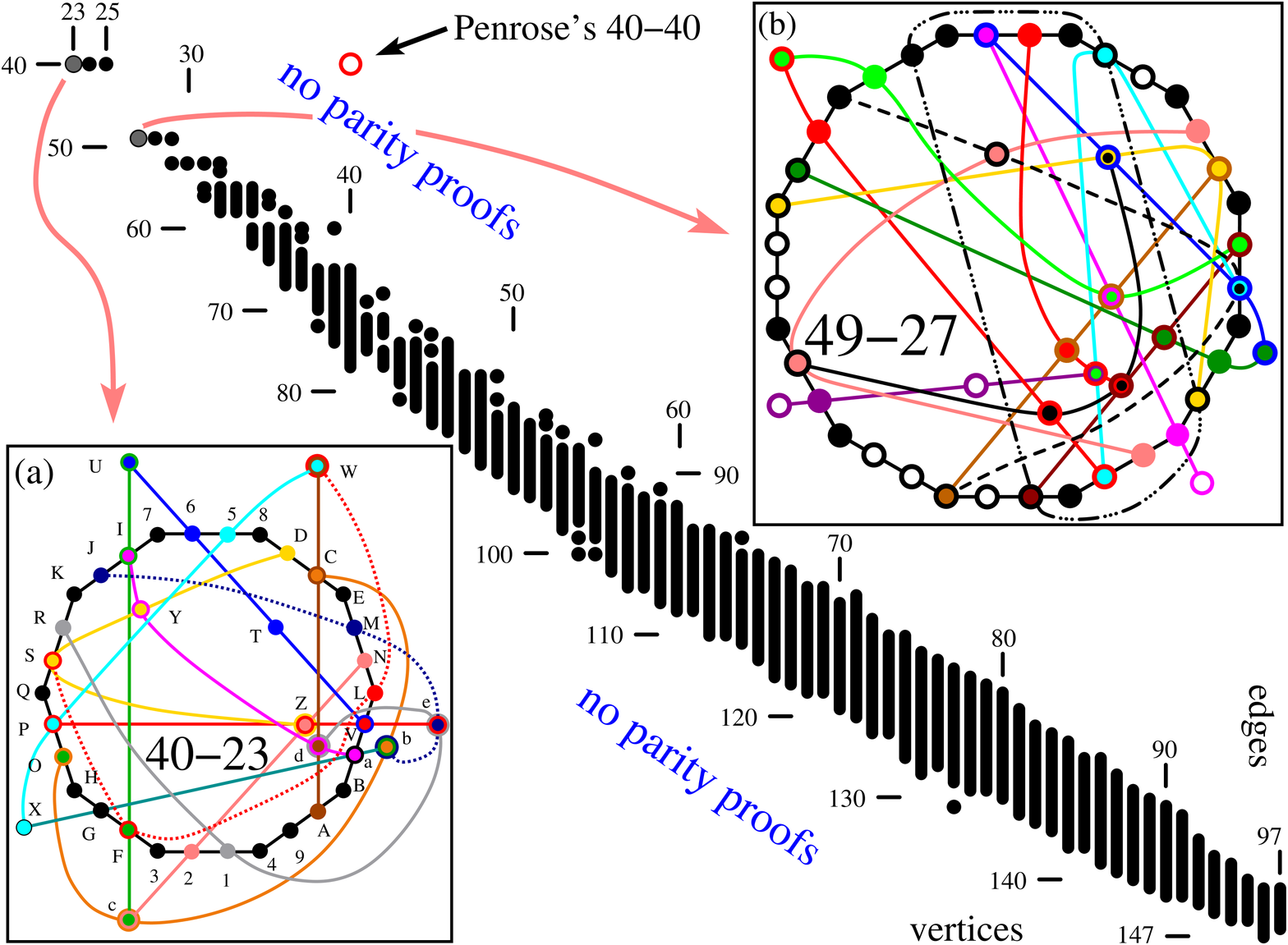}
\end{center} 
\caption[MMP Hypergraphs or Kochen-Specker Sets from the 148-265 Class.]{
  KS criticals from the 4-dim 148-265 class; none of them has a parity
  proof; 40-40 circle (red online) indicates the Penrose 40-40
  non-critical KS set; inset (a) shows one of smallest KS criticals
  generated from Penrose's 40-40 set; inset (b) shows one of the
  smallest KS criticals not contained in the 40-40 set. }
  \label{fig:ks-148-265}
\end{figure}

Waegell and Aravind in \cite{waeg-aravind-17-arXiv} make use of a
rather involved coordinatization but they also indicate that a
simpler one, in which vector components take the values from the
set $\{0,\pm 1,\pm\omega,\pm\omega^2\}$, where
$\omega=e^{2\pi i/3}=(-1+i\sqrt{3})/2$, can be used
\cite[Eq.~(6)]{waeg-aravind-17-arXiv}. We explicitly verified
that, in the master set 148-265, vectors can indeed be ascribed
a valuation from this set which means that all sets from the
148-265 class can easily be given a random valuation with the
help of our program {\tt vectorfind} by simply introducing the
7 values given above as its options. Two examples of such a
valuation are 40-23 and 49-27 KS criticals. 40-23
MMP hypergraph, shown in Fig.~\ref{fig:ks-148-265}(a) is one of
56 40-23 subgraphs of Penrose's 40-40 KS hypergraph. 49-27,
shown in Fig.~\ref{fig:ks-148-265}(b), is the smallest critical
from the 148-265 class which is not contained in its 40-40 set.
MMP hypergraph strings of 40-23 and 49-27 are given in
Appendix \ref{app:0-4}.

In contrast to the smallest sets from the other 4dim KS classes
the above smallest hypergraphs do not show geometrical symmetries
and that is caused by the geometrical features of the Witting
polytope which in turn cause that the vertices share both even and
odd number of edges, i.e., that they do not have parity proofs. 
Actually, none of 250140 KS criticals in the 148-265 class we
obtained has a parity proof, so, they are completely invisible
for the parity-proof-based algorithms and programs.

The maximal loops of the criticals are up to 36-gons big and
therefore smaller than the ones of all the higher KS criticals
from the 300-675 class but bigger than ones of all the other
KS criticals from any other class.

Similarly to 60-74 and 300-675 and unlike 24-24 and 60-105 classes,
no two edges share more than one vertex.

Program {\tt subgraph\/} verified that the master set 148-265 is not
a subgraph of the master set 300-675. Program {\tt shortd\/} verifies
that the classes 148-265 and 300-675 are completely disjoint. 

\section{\label{sec:6dim}$\medstar/\bigtriangleup$ 21-7 6-dim 
KS Set and 236-1216 Class of 6-dim KS Sets}

Lison\v ek, Badzi\c ag, Portillo, and Cabello \cite{lisonek-14} 
recently found a 6-dim 21-7 KS set which they drew in the form of a 
seven pointed star, a regular heptagram with Schl\"afli symbol 
\{7/3\}, as shown in Fig.~\ref{fig:6-dim}. (It was experimentally
implemented in \cite{canas-cabello-14}.) They chose vector
component values from the set $\{0,1,\omega,\omega^2\}$, as we did
in Sec.~\ref{sec:4dim-witt}, however, since $\omega$ is
a cube root of 1 and therefore $\omega\times\omega^2=1$, and 1
is already present as a component, for their set $\omega^2$ is not
needed. To see this, we start with the MMP hypergraph representation
of the seven star set:
{\tt 123456,6789AB,BCD3EF,F5G8HI,IAJD2K,KE4G7L,LH9JC1.} 
We assign {\tt 1} to any point and then proceed along the edges. 
Our program {\tt vectorfind\/} can then ascribe the components to
vertices: {\tt 1={\rm (0,0,0,0,0,1)}, 2={\rm (0,0,0,0,1,0)}, ...,6={\rm (1,0,0,0,0,0)}, 7={\rm (0,1,0,$\omega$,1,$\omega$)}, 8={\rm (0,0,1,1,$\omega$,$\omega$)}, 9={\rm (0,$\omega$,$\omega$,1,1,0)}, A={\rm (0,$\omega$,1,$\omega$,0,1)}, B={\rm (0,1,$\omega$,0,$\omega$,1)},\break C={\rm (1,$\omega$,1,0,$\omega$,0)}, D={\rm ($\omega$,1,1,0,0,$\omega$)}, E={\rm ($\omega$,$\omega$,0,0,1,1)}, F={\rm (1,0,$\omega$,0,1,$\omega$)}, G={\rm (1,0,0,$\omega$,$\omega$,1)}, H={\rm ($\omega$,0,1,$\omega$,1,0)}, \break I={\rm ($\omega$,0,$\omega$,1,0,1)}, J={\rm (1,1,$\omega$,$\omega$,0,0)}, K={\rm (1,$\omega$,0,1,0,$\omega$)}, L={\rm ($\omega$,1,0,1,$\omega$,0)}}.
Recall that dot products (orthogonality of vectors) involve
the complex conjugates, e.g.,
${\tt K}\cdot{\tt L}^\dagger=\omega^*+\omega+0+0+1+0+0=(-1-i\sqrt{3})/2+(-1+i\sqrt{3})/2+1=0$.

\begin{figure}[hbt]
\includegraphics[width=0.99\textwidth]{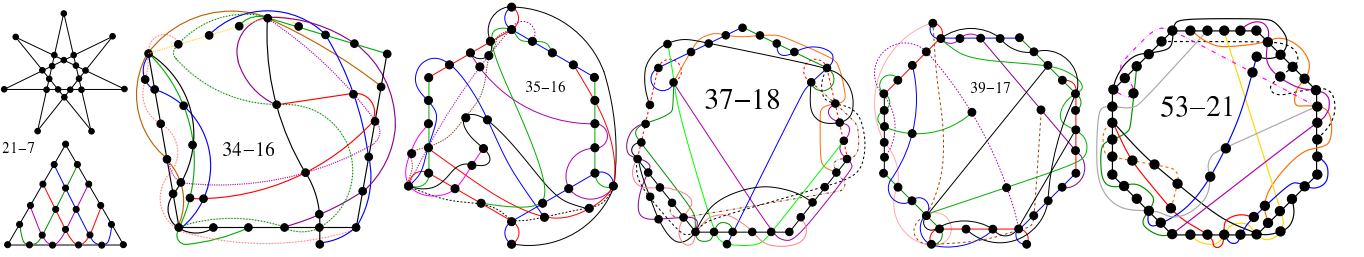}
\caption[6-dim Kochen-Specker Sets]{6-dim KS critical sets:
$\medstar/\bigtriangleup$ 21-7 and KS from the 236-1216 KS class;
$\medstar$ 21-7 is from \cite{lisonek-14}; $\bigtriangleup$
21-7 is isomorphic to $\medstar$; others are critical KS sets
from the 236-1216 class; 53-21 has a parity proof.}
\label{fig:6-dim}
\end{figure}

The heptagram is isomorphic to a triangle ($\bigtriangleup$)
hypergraph shown in Fig.~\ref{fig:6-dim} below the star
($\medstar$). The advantage of the triangular representation is that
it can describe both odd and even dimensional sets while the star
like representation is limited to the even dimensional sets.
4-dim triangle (5 pointed star), which does not 
admit a 0-1 state, would be a 10-5 KS set if it had a vectorial 
representation in the complex Hilbert space, but apparently it does 
not have it. Here we stress that all the results we obtained for
the KS sets of the 60-105 class depend on vector components from
$\{0,\pm 1, \pm i\}$. There are 4-dim KS sets with vector 
components from other sets, e.g., the ones from the 60-74 class. 
There are also, hypergraphs which do not admit non-contextual 0-1 
states, e.g., those smaller than the 18-9 \cite{pmmm05a}, or the 
above 10-5 one, for which we actually do not know whether they 
have a vectorial representation with vector component values from 
some other sets. A direct solving of nonlinear equations which would 
answer this question is rather demanding.  

Neither the 5-dim triangle (15-6 set) nor the 7-dim triangle (28-8) 
are KS hypergraphs (they do admit 0-1 states), but the 8-dim one (36-9) 
is. The latter KS set is also not a subgraph of the 120-2024 class 
(see Sec.~\ref{sec:8dim}). 

The authors of \cite{lisonek-14} have made an attempt to find a 
bigger 6-dim set but did not find any. 

Waegell and Aravind appreciated the approach as the first one 
``in a dimension that is not of the form $2^N$'' 
\cite{waeg-aravind-fp-14}, meaning that the 6-dim space cannot 
``host'' qubits (recall that two qubits reside in the $2^2$dim, 
i.e., 4-dim Hilbert space, three qubits in the $2^3$dim, i.e., 
8-dim space, etc.). Subsequently Aravind and Waegell 
\cite{aravind-waegell-6dim-private} designed a 6-dim 
236-1216 master KS set but since it did not allow parity proofs
they could not generate smaller sets with their parity proof
programs. So, they sent the master set to us and we generated 
$3.7\times 10^6$ KS criticals in this paper. We say that they 
make the 236-1216 class. Their statistics is shown in
Fig.~\ref{fig:table-6-dim}.
The vector components take values from the 
set $\{0,\pm 1/2,\pm 1/\sqrt{3},\pm 1/\sqrt{2},1\}$. 
The class does not contain the 21-7 KS set, though (verified
with {\tt subgraph}). 

\begin{figure}[hbt]
\includegraphics[width=0.99\textwidth]{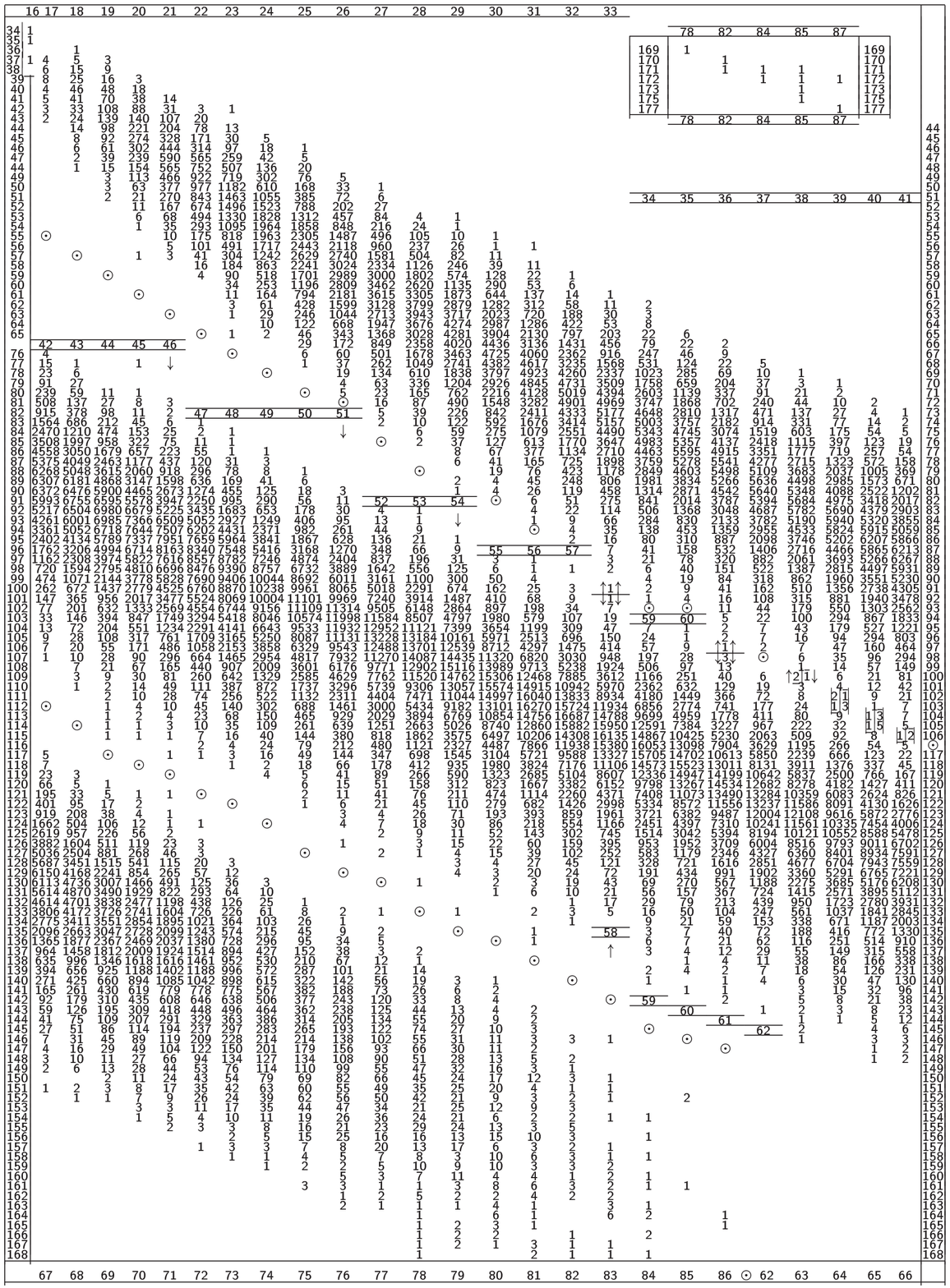}
\caption{List of 3,714,503 non-isomorphic
  6-dim KS critical sets from the 236-1216 class; 16 to 87 edges
  (columns); 34 to 177 vertices (rows); 169-78 to 87-177 sets
  are show in the inset.}
\label{fig:table-6-dim}
\end{figure}

Aravind \cite{aravind-waegell-6dim-private} has arrived at the 
236-1216 master KS set by considering hypercubes which led him 
to a hexeract (6-cube, 6-dim cube) with Schl\"afli symbol 
\{4,3,3,3,3\} or \{4,$3^4$\}. The master set written in the MMP 
notation occupies more than 3 pages so we do not print it here.

The approach of Aravind and Waegell is very geometrical and 
unorthodox and by no means straightforward, so, it is outside of the 
scope of the present paper. It will be presented in detail in a 
separate publication. The master set in the MMP notation is given
in our repository.

The features of the 236-1216 class are: 

--- Its KS sets cannot be implemented via qubits but can via 
spin-$\frac{5}{2}$ quantum systems; 

--- Its smallest KS sets have an even number of 
edges and small sets with odd and even number of edges are evenly 
distributed, unlike in any other class;

--- Although the number of vertices of the master set is comparable
with the 4-dim 300-675 and the number of edges is twice as high,
the criticals are computationally much easier to generate; a generation
of a single KS critical is up to 1,000 times faster; 

--- Types of sets with a definite number of edges and 
different number of vertices are more numerous than in other classes 
(columns in Fig.~\ref{fig:table-6-dim} are higher than in other tables);
a dynamic algorithm compensated for a lower occurrence of smaller
KS sets;

--- Edges connect vertices in much more irregular way than in 
other classes as the figures in Fig.~\ref{fig:6-dim} show. We were 
not able to find a single symmetric hypergraph; 

--- Statistics from Fig.~\ref{fig:table-6-dim} shows gaps in 
the KS sets with high number of edges indicating that a 
more extensive generation would generate many more sets possibly 
with higher number of edges and vertices. 

There is another peculiarity we should mention. 

As already stressed above, in the literature, most of the KS proof 
have been found via parity proofs. However, in the 236-1216 class
among $3.7\times 10^6$ KS critical sets we generated we found only
8 KS criticals with a parity proof. Their edges are in the interval
from 21 to 39. We shall present and discuss them in detail in a
subsequent publication and here we only show one of them (53-21) in
Fig.~\ref{fig:6-dim}.

\section{\label{sec:8dim}120-2024 Class of  8-dim
KS Sets and $\medstar/\bigtriangleup$ 36-9 8-dim KS Set}

We start with a brief history of generation of 8-dim KS sets
which can be realised with either 3 qubits ($2^3=8$) or
spin-$\frac{7}{2}$ systems. 
In 1995 Kernaghan and Peres produced a 36-11 KS critical set
and a 40-25 non-critical one (experimentally implemented in
\cite{canas-cabello-8d-14}) from which several smaller ones
including 36-11 can be obtained \cite{kern-peres}; in 2006 Ruuge
and van Oystaeyen gave a scheme for constructing 8-dim KS proofs
but did not themselves construct any \cite{ruuge05}; in 2012 Ruuge
claimed to have given an example of a 36-vertex 8-dim KS set
\cite{ruuge12} but we were not able to identify its octads of
orthogonal vertices in \cite{ruuge12} (nor to contact him), so,
we could not verify whether it is isomorphic to 36-11 from
\cite{kern-peres} as claimed in \cite{ruuge12}); and finally, also
in 2012, Planat discussed 8-dim KS sets that can be obtained from
the Kernaghan-Peres' 40-25 KS set \cite{planat-12}. In 2015 Waegell
and Aravind obtained a KS master set with 120 vertices and 2025
edges and, from it, many smaller 8-dim KS sets, including non-critical
Kernaghan-Peres' 40-25 one \cite{waeg-aravind-jpa-15} (see also
\cite{waegel-aravind-12}). In the present paper we generate
$6.9\times 10^6$ non-isomorphic KS criticals, listed in the table in
Fig.~\ref{table:8-dim}, from that Waegell-Aravind's 120-2025 master
set. We also produce a new star/triangle ($\medstar/\bigtriangleup$)
36-9 KS set which is not a subgraph of the 120-2025 master set. 

To obtain KS sets, in Refs.~\cite{ruuge05,waeg-aravind-jpa-15}, the
authors made use of the Lie algebra E8. Waegell and Aravind reduced
it to a collection of 120 vectors (rays, vertices) and 2025 bases
(octads, edges) \cite{waeg-aravind-jpa-15} to obtain their 120-2025
KS master set. We verified that by peeling off one edge at the time
we obtain 2025 varieties of the 120-2024 KS sets which are all
isomorphic to each other and therefore reduce to a single 120-2024
KS master set from which we generate the 120-2024 KS class, i.e.,
smaller KS critical sets. Critical KS sets from the 120-2024 class
are given in the table in Fig.~\ref{table:8-dim}. 

\begin{figure}[hbt]
\includegraphics[width=0.99\textwidth]{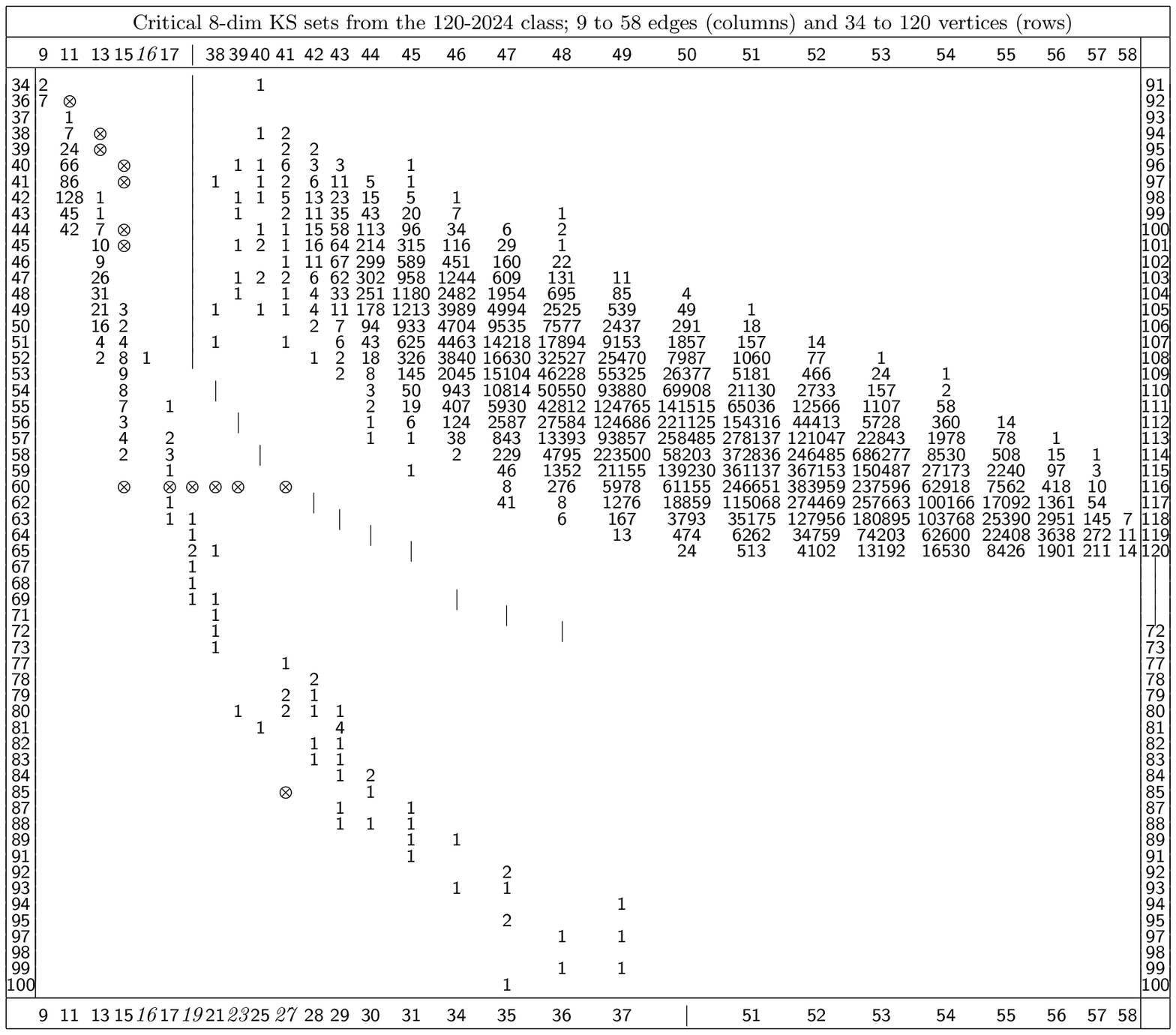}
\caption{List of 6,925,540 non-isomorphic 8-dim KS critical sets from
  the 120-2024 class we obtained on our cluster and of those, denoted
  as $\otimes$, obtained by Waegell and Aravind
  \cite{waeg-aravind-jpa-15} and still not by us.}
\label{table:8-dim}
\end{figure}

The coordinatization (vector components) in \cite{waeg-aravind-jpa-15}
is taken over from D.~Richter and is based on tetrads formed by
expressions $r_me^{in\pi/30}$ (values of constants $r_m$ and $n$ are
given in \cite{waeg-aravind-jpa-15}) so that their real and imaginary
parts form octads. Using this coordinatization Waegell and Aravind
generate sets of bases (edges)  which define their KS sets. We,
however, do not need the coordinatization to obtain KS sets. We start
with the master set 120-2024 and simply strip off edges. Then we
filter the smaller sets via {\tt states01\/} to obtain critical KS
sets. We can always add vector components later on, if needed.  

The distribution of sets from the 120-2024 KS class is different from
the above 6-dim class as well as from the three of the 4-dim ones
and somewhat similar to the 300-675 4-dim class with respect to the
following feature. The critical sets are split so as to be clustered
in two groups of subsets with respect to the number of vertices and
edges: first one, sparsely spread, over 9 to about 40 edges and 34
to about 100 vertices and the second one, densely populated, over
about 41 to 58 edges and about 100 to 120 vertices as shown in the
table in Fig.~\ref{table:8-dim}. The split structure of the 120-2024
class resembles the similarly split structure of the 4-dim 300-675
class. We conjecture that there is only one or at most a few KS
noncritical sets with about 100 vertices and 40 edges which most
the smaller critical sets are subsets of. 

Similarly to the 4-dim classes (with the exception of the 60-105
one) and the 6-dim class, the number of critical sets which exhibit
a parity proof is very small with respect to the total number
of critical sets, but on the other hand, parity proof algorithms
used by Waegell and Aravind \cite{waeg-aravind-jpa-15} are very
efficient in generating the sets so that the two approaches (via
the MMP algorithms for bare hypergraphs and the parity-proof-based
ones for vectors corresponding to vertices of hypergraphs) turn out
to be complementary. In particular, Waegell and Aravind
\cite{waeg-aravind-jpa-15} obtained the following sets which still
did not appear in the course of our computer generation so far:
36-11 (Kernaghan-Peres), 38,39-{\bf 13}, 40,41,44,45-{\bf 15},
48-{\bf 17}, 60-{\bf 15,17,19,21,23,27}, and 85-{\bf 25} (we do show
these sets in the table in Fig.~\ref{table:8-dim} as $\otimes$);
both Waegell and Aravind \cite{waeg-aravind-jpa-15} and we in the
present paper obtained 34-{\bf 9}, 36-{\bf 9}, 37-{\bf 11}, and
95-{\bf 35}; Waegell and Aravind \cite{waeg-aravind-jpa-15} have
not obtained all the other sets we obtained in the table in
Fig.~\ref{table:8-dim} and most of them they actually cannot obtain
due to the features of the parity based algorithm they make use of
but, still, the parity proof based programs confirm themselves as a
powerful complimentary method of providing us with KS criticals since
our general MMP hypergraph algorithms are CPU-time demanding. 

In Fig.~\ref{fig:8-dim-smallest}, we show five chosen KS criticals
from the 120-2024 class. KS criticals 34-9 are the smallest in the
class. KS 36-9 is particularly interesting because it can be viewed
as an 8-dim version of 18-9 from Fig.~\ref{fig:24-24-class}(a) with
graphically analogous edges where each vertex from the 18-9 is
represented by a pair of vertices in the 36-9. KS 44-11 is one of  
the critical KS sets with the biggest maximal loop (heptagon) among 
the sets with 11 edges (second smallest number of edges). KS 52-16 has
the smallest even number of edges. One of 14 KS 120-58 has the
biggest maximal loop---tetradecagon (14-gon); it is not shown in
the figure. 

\begin{figure}[hbt]
\includegraphics[width=0.99\textwidth]{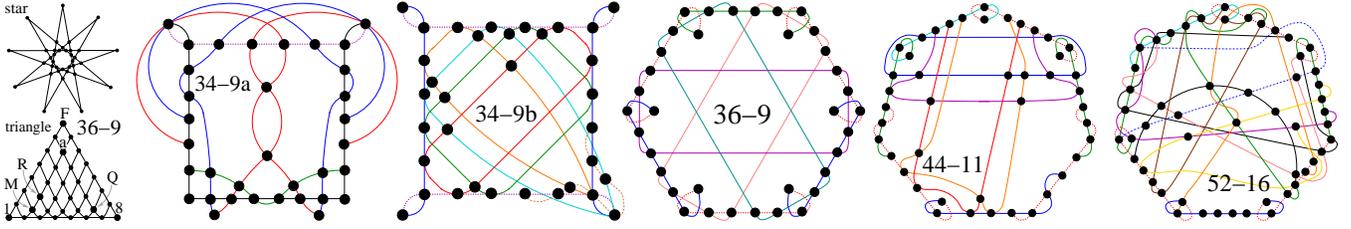}
\caption[8-dim Kochen-Specker Sets]{8-dim Kochen-Specker sets:
$\medstar/\bigtriangleup$ 36-9 KS set and five chosen critical
KS sets from the 120-2024 KS class; see the text for a
description of their features.} 
\label{fig:8-dim-smallest}
\end{figure}

In Sec.~\ref{sec:60105} we have seen that the 24-24 class
is contained in the 60-105 class and in Sec.~\ref{sec:300675}
that the 60-74 class is contained in the 300-675 class. On the
other hand in Sec.~\ref{sec:6dim} we have shown that the
$\medstar/\bigtriangleup$ KS set is not contained in the much bigger
236-1216 class of the 6-dim KS sets. Here we verified that
8-dim $\medstar/\bigtriangleup$ 36-9 KS critical set, shown in
Fig.~\ref{fig:8-dim-smallest}, is not contained in also much bigger
120-2024 class of 8-dim KS sets.

The MMP representations of the star and triangle forms (they are
mutually isomorphic) of the 36-9 critical are given in Appendix
\ref{app:0-6}. In Fig.~\ref{fig:8-dim-smallest} the first three edges
correspond to the edges of the triangle as indicated by its vertices
{\tt 1,8,F} and then the inner vertices are denoted in alphabetical
order from left to right from the bottom horizontal ones (indicated
by {\tt M} to {\tt Q}) to the single {\tt a} at the top.

In contrast to 6-dim 21-7 set from Fig.~\ref{fig:6-dim}, this 8-dim
36-9 can have real vector components from \{-1,0,1\}.
Coordinatizations for the triangle and for the star are given in
Appendix \ref{app:0-6}.

Interestingly, our program {\tt vectorfind} finds the triangle
coordinatization sooner than the one for the star.

The 8-dim star/triangle set is not smaller than the smallest sets from
the 120-2024 KS class as the 6-dim one is with respect to the smallest
sets from the 6-dim 236-1216 class; the 34-9 sets shown in
Fig.~\ref{fig:8-dim-smallest} are smaller. The 120-2024 class contains
at least seven 36-9 criticals but their structure is very different
from the star/triangle 36-9 (Cf.~36-9 in the middle of
Fig.~\ref{fig:8-dim-smallest}).

Via our program {\tt subgraph} we prove that the star/triangle 36-9
or any other 36-9 isomorphic to it cannot be obtained by stripping
edges and vertices from the master set 120-2024 down to sets with
36 vertices and 9 edges, i.e., that it cannot be a subgraph of
the master set and that it therefore does not belong to the
120-2024 class.

Of all sets from the 120-2024 class we generated so far, only
ca.~0.1\permil\ have parity proofs, notably 609 of 6,925,540. The
star/triangle 36-9 does have a parity proof, though.

\section{\label{sec:16dim}80-265 Class of 16-dim KS Sets}

In 2012 Harvey and Chryssanthacopoulos constructed an 80-265 KS
master set in the 8-dim real Hilbert space with vector components
from the set $\{-1,0,1\}$ \cite{harv-cryss-aravind-12a}.
They considered it for four qubits ($2^4=16$)
although---theoretically---it can also serve as a KS set for
spin-$\frac{15}{2}$ systems. The set has far too many redundant
edges, so, Planat promptly designed a procedure to obtain smaller
KS sets and he claimed to have obtained three sets with the
initial number of vertices: 80-21, 80-22, and 80-23
\cite{planat-12}, however, as we show below, his 80-21 and 80-22
are not KS sets and 80-23 is not critical. In this paper we
generate $4.1\times 10^6$ non-isomorphic critical KS sets from
the 80-265 master set. We say that KS critical sets that can be
generated by stripping the 80-265 master set form the 80-265
class of KS critical sets. The ones we obtained so far are shown
in the table in Fig.~\ref{table:16dim}.

\begin{figure}[hbt]
\includegraphics[width=0.99\textwidth]{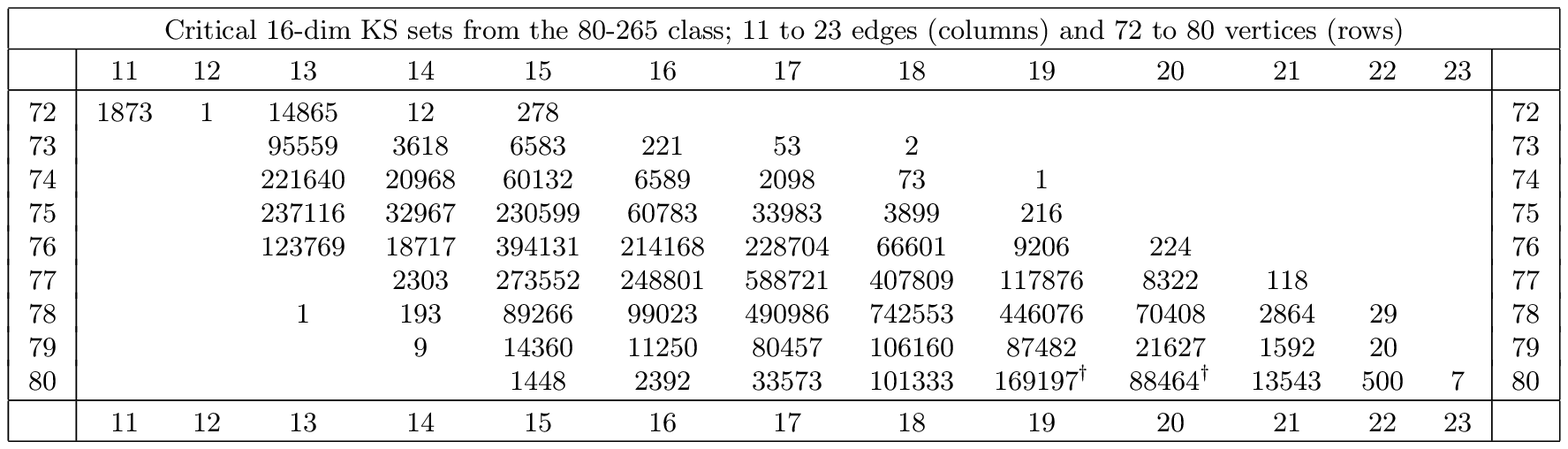}
\caption{List of 4,069,963 non-isomorphic 16-dim KS critical sets from
  the 80-265 class. Two 80-19s and one 80-20 KS criticals
  we generated from Planat's \cite{planat-12} non-critical 80-23 are
  among our 169,197 \ \ 80-19s \ and \ 88,464 \ \ 80-20s, respectively,
  indicated by $^\dagger$ in the table.}
\label{table:16dim}
\end{figure}

The original 80-265 master file printed in \cite{harv-cryss-aravind-12a}
took over 11 pages. Its MMP representation is much shorter. Still, it
takes over one page. So, we shall consider some smaller examples,
but, first, we shall check the sets Planat obtained in \cite{planat-12}.

His set 80-21 given by 21 lines of Eq.~(17) in \cite{planat-12} has an
MMP rendering with 21 edges as given in Appendix \ref{app:0-7}.
But this set is not a KS set. For instance, according to our program
{\tt states01} we can assign ``1'' to {\tt G, H, Y, o, r} and {\tt u},
so as to exhaust all 21 edges, i.e., when we delete the edges that
contain them, then none is left, meaning that each contains one ``1''
and therefore the set is non-contextual.
Cf.~Fig.~\ref{fig:24-24-class}(a) where, e.g., we can assign ``1'' to
none of vertices {\tt 789A} and for which there is always an edge to
which one cannot assign ``1'' at all vertices contained in it.

Then it is claimed that this 80-21 set together with the 1st line of
Eq.~(18) from \cite{planat-12}, in MMP notation:
${\tt 2ACEZbhj\$(*:<>?@}$, form an 80-22 KS set. However, this 80-22
is not a KS set, either.

All lines from Eqs.~(17) and (18), the last line reading   
{\small\tt notuwy!\#')-:\textless =\textgreater ?} in MMP notation,
form an 80-23 non-critical KS set. By deleting the first line,
{\small\tt Zbhjprsv\$(*-:\textless=@}, we get a non-critical 80-22 KS
set. If we also deleted the eighth line ({\small\tt HIKLMPQRTUVWbhlm}),
we would get a non-critical 80-21 KS set. These sets contain
one 80-20 critical set and two non-isomorphic 80-19 criticals---all
shown in Appendix \ref{app:0-7}. Their maximal loops are pentagons.
We obtained them from the aforementioned 80-23,22,21 via our
program {\tt states01}. 

The goal of \cite{planat-12} was to find small KS sets, but the 
table in Fig.~\ref{table:16dim} shows that its non-critical KS set
80-23 is bigger than all 2.5 million KS criticals we generated from the
master 80-265 set and listed in the table in Fig.~\ref{table:16dim}.
This shows that algorithms for automated exhaustive generation of
MMP hypergraphs, although probabilistic until full exhaustion is
reached, are indispensable sources for obtaining new KS sets. 
Still, the 80-20 and 80-19s we obtained with the help of our program
{\tt states01} are not isomorphic to any of the 80-20s and 80-19s we
listed in the table in Fig.~\ref{table:16dim}. This is because the
probability of generating any specific KS set via our programs
{\tt mmpstrip} and {\tt states01} is very low due to the their
probabilistic algorithms. Within established probabilities for
obtaining MMP hypergraphs with wanted number of edges and vertices
they are generated completely at random.

The 16-dim KS criticals listed in the table in Fig.~\ref{table:16dim}
have maximal loops in the range from a square to a heptagon as
illustrated in Fig.~\ref{fig:16-dim-smallest}. The vector components
corresponding to vertices from the set \{-1,0,1\} for the master set
listed in \cite{harv-cryss-aravind-12a} can be traced down to any
chosen MMP hypergraph from the table in Fig.~\ref{table:16dim}
via any of our programs {\tt mmpstrip}, {\tt states01},
{\tt mmpshuffle}, etc., or, equivalently, program {\tt vectorfind}
can generate the components directly for a given hypergraph. 

\begin{figure}[hbt]
\includegraphics[width=0.99\textwidth]{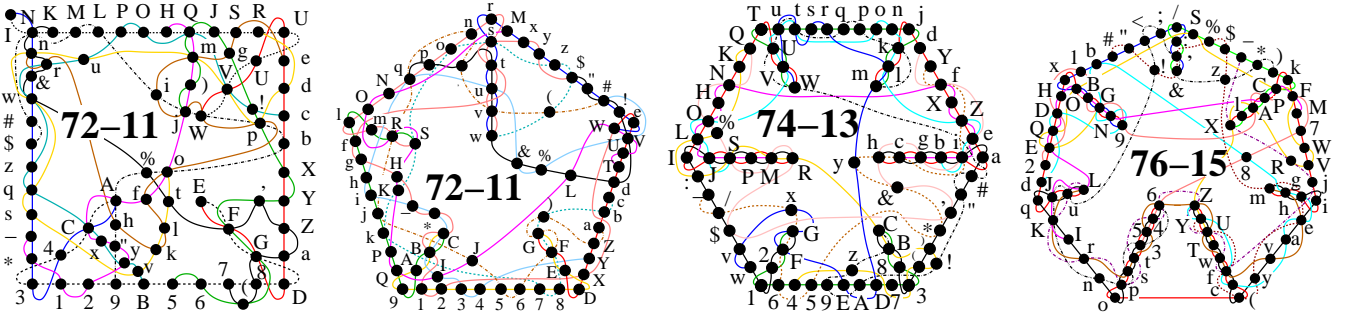}
\caption[16-dim Kochen-Specker Sets]{16-dim critical KS sets. The
smallest sets with square, pentagon, hexagon, and heptagon
maximal loops are shown. MMP hypergraph strings are given in
Appendix \ref{app:0-7}.} 
\label{fig:16-dim-smallest}
\end{figure}

In contrast to all previous classes of KS sets apart from 4-dim
60-105, 16-dim 80-625 class has a significant number of parity proofs,
notably, 28\%. There are ca.~64\%\ criticals with an odd number of
edges but only 44\%\ of them have parity proofs.

Also in contrast to all previous classes, except the 6-dim 236-1216
class, the KS criticals of the 16-dim 80-265 class do not exhibit
symmetries. They have rather intricate and dense structure. In
particular, all vertices share at least two edges and some pairs of
edges share eight vertices. Also, in contrast to KS sets from the
classes in smaller dimensions (not counting the tentative 24-24 
class for which we in Sec.~\ref{sec:60105} proved to be contained
in the 60-105 class), there are no maximal loops bigger than
heptagons. There are ca.~10\%\ of squares, 86\%\ of pentagons,
4.6\%\ of hexagons, and 1\textpertenthousand\ of heptagons. 

Why is 77-13 missing, while 76-13 has 123,769 non-isomorphic instances
and 78-13 is present, is an open question. 

16-dim star/triangle set does not admit 0-1 states and is critical and
therefore would be a critical KS set if one found a coordinatization
for it. We have not found any so far. It has 16+1=17 edges and
(16+1)16/2=136 vertices, which are 1.7 times the highest number of
vertices of the critical sets from the 80-265 class. Its structure is
dissimilar to any obtained set from the 80-265 class so it is very
unlikely that it might belong to it, however, for the time being, the
program {\tt subgraph} which would give us a definitive answer to this
question it is still running. 

\section{\label{sec:32dim}160-661 Class of 32-dim KS Sets}

Recently Planat and Saniga, extending Aravind's and DiVincenzo-Peres'
generalisations of the Bell-Kochen-Specker theorem
\cite{aravind-02,divinc-peres}, constructed a 32-dim KS master set
with 160 vertices/vectors and 661 edges with a real coordinatization
from the set $\{-1,0,1\}$ \cite{planat-saniga-12}. This is a
very big set which corresponds to states of five qubits, so, they did
not present it in their paper. But, M.~Planat kindly sent us the
set in their notation and we translated it to an MMP encoded hypergraph. 
Planat and Saniga only published a smaller 160-21 KS set they obtained
from the master set. MMP hypergraph string of that set is given in
Appendix \ref{app:0-8}. (Edge 14 in \cite{planat-saniga-12} which reads
08 should read 108.) 

However, this KS set is not critical and it contains at least two
smaller critical KS sets, 160-19 and 152-19 ones. There is no point
in giving their MMP representations here because we obtained thousands
of smaller KS 32 criticals from the 160-661 master set as shown in the
table in Fig.~\ref{table:32dim}. The non-isomorphic KS criticals with
the smallest number of edges (11) all have 144 vertices and we show
one of them in Fig.~\ref{fig:32-dim-smallest}.

\begin{figure}[hbt]
\includegraphics[width=0.99\textwidth]{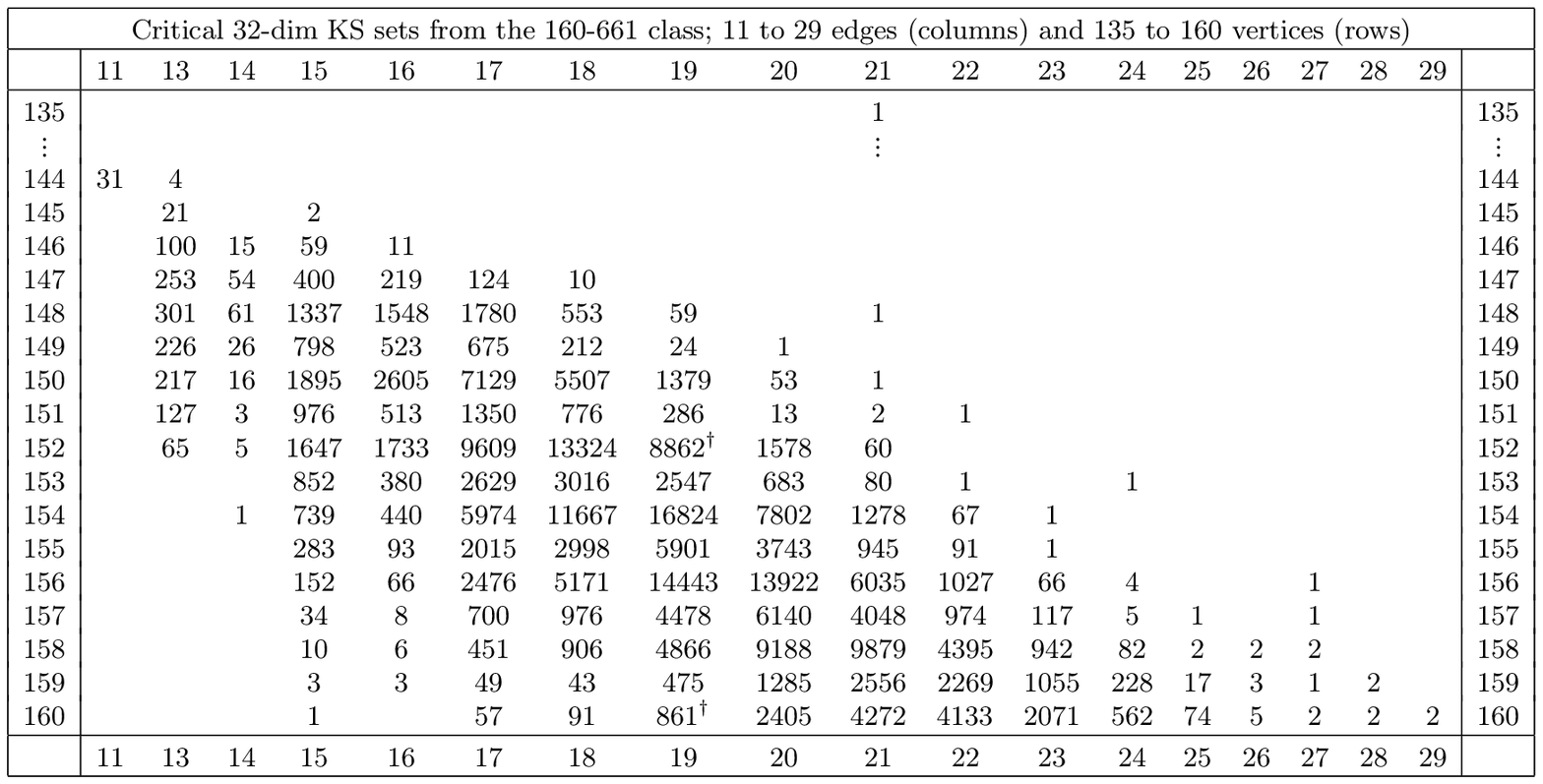}
\caption{List of 254,318 non-isomorphic 32-dim KS critical sets from the 
  160-661 class and 52-19 and 60-19 criticals we derived
  from Planat and Saniga's \cite{planat-saniga-12} non-critical 160-21;
  the latter criticals are included in 8862 \ \ 52-19s \ and \ 861 \
  \ 60-19s, respectively, and indicated by $^\dagger$ in the table.}
\label{table:32dim}
\end{figure}

\begin{figure}[hbt]
  \includegraphics[width=0.42\textwidth]{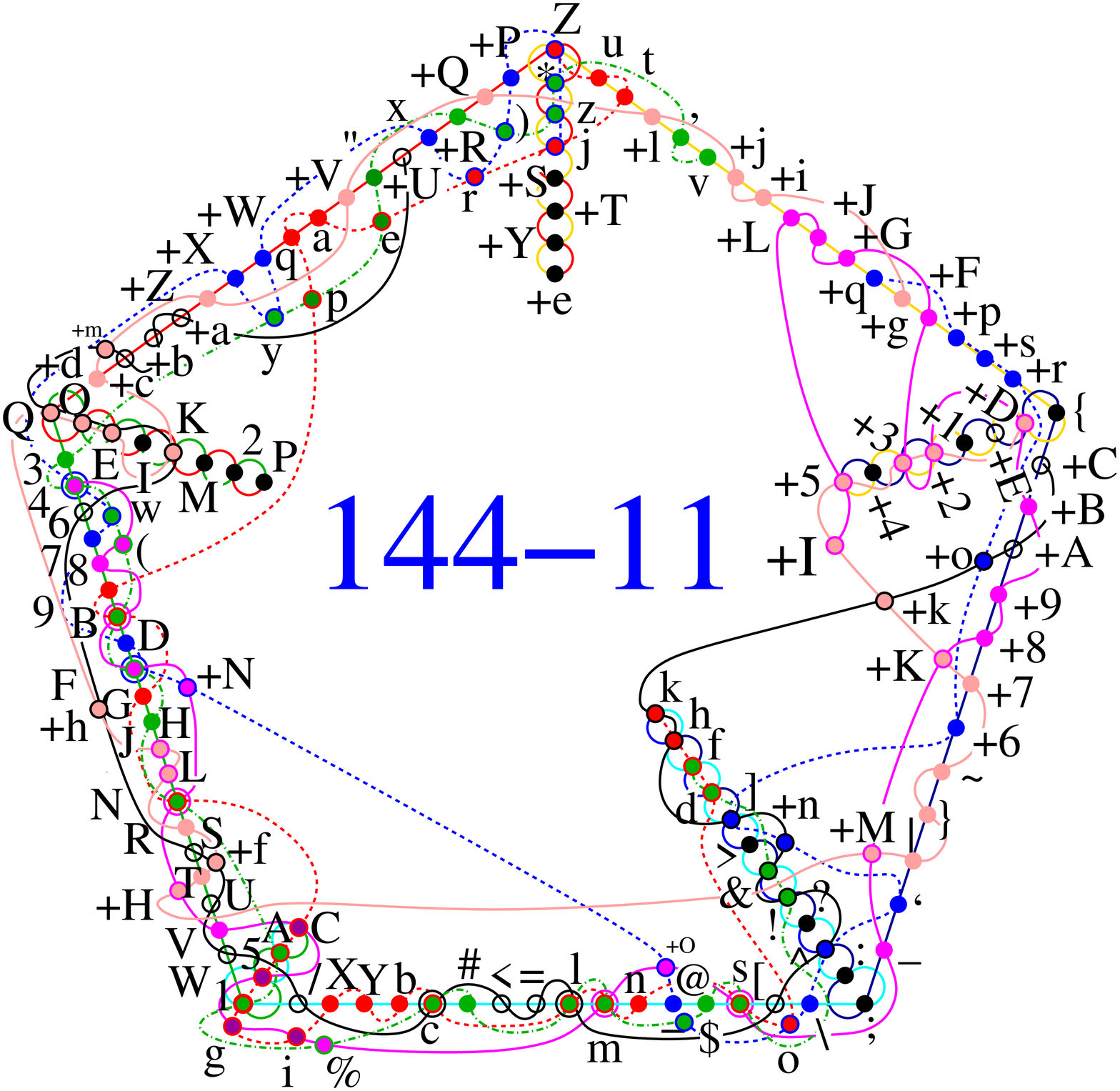}
\caption[32-dim KS sets]{144-11; one of 31 smallest 32-dim critical KS
sets.}
\label{fig:32-dim-smallest}
\end{figure}

It is interesting that a single KS set with 21 edges (the number of
edges of the KS set from \cite{planat-saniga-12}) has 9 vertices less
than any other set we found (135); indicated by $\vdots$ in the table
in Fig.~\ref{table:32dim}. This might stem from some geometrical
structure of the set or its smaller subset. We do not show the
135-21 hypergraph because it has almost twice as many edges as
144-11 and its figure would be much more difficult to read. Also
135-21 maximal loop is a square and its hypergraph has 47 vertices
outside of the loop as opposed to 24 such vertices of the 144-11
hypergraph shown in Fig.~\ref{fig:32-dim-smallest}. 

We find the complexity of KS criticals from the 32-dim 160-661 
class similar to the one of the 16-dim 80-265 class. Not a single
vertex shares only one edge, and some pairs of edges share 
16 vertices. The distribution of distances between maximal
bases was considered in \cite{planat-saniga-12} for a single
non-critical 160-21 set. In our approach such a distribution
does not play any role for either obtaining thousands of KS
criticals or proving that they really are KS sets.

The vector components of vertices from the set \{-1,0,1\}
for the master set they are listed in \cite{planat-saniga-12}.
As for all the sets from the previous classes given above we can
either trace them or generate them for any given hypergraph. 

Of all $2.5\times 10^5$ KS criticals we obtained, only 10.7\%\
have a parity proof. 

In contrast to all KS sets from the classes in smaller dimensions,
there are no maximal loops bigger than hexagons. There are 11.9\%\
of squares, 87.4\%\ of pentagons, and 0.7\%\ of hexagons.

The star/triangle KS set, although not admitting 0-1 states and
being critical, is far too complicated to be considered here. It
has (32+1)32/2=528 vertices and 32+1=33 edges which makes it far
bigger than any of the critical sets from the present class.
However, if one found a coordinatization for it would be the
biggest critical KS set of all known ones.

\section{3-dim KS Sets}
\label{sec:3d}

The successful generations of all the above presented KS sets in up
to 32 dimensions were enabled by newly found big master sets and
they were in turn derived from various polytopes (like, e.g.,
120-cell and 600-cell), or Lie algebras, or some involved
individual constructions which made use geometric symmetries of
even-dimensional spaces. Even without the big master sets a direct
generation of smaller KS sets is possible via our MMP algorithms
\cite{pmmm05a-corr} because in four and higher dimensional space
those KS sets are pretty small. Disparately, for the 3-dim space,
so far, no one has come forward with a master set and of the few
known 3-dim KS sets no one is small and all of them are critical
and cannot be lessened.

Since it would be very important to find more 3-dim KS
sets to gain a better insight into the structure of
contextual KS sets and enable new breakthroughs in their
generation and application algorithms and programs, in this
section we give MMP representations and KS hypergraphs of
the four known (the only known ones) 3-dim KS sets and
one that was claimed to be of such kind (the Yu-Oh 13-set),
but is not, as we show below.

The full specification of all vertices (their vector
components) is, as shown by Larsson \cite{larsson} and
Pavi\v ci\'c, Merlet, McKay, and Megill \cite{pmmm04c},
indispensable ``for an experimental realisation, which
involves procedures equivalent to basis rotations''
\cite[p.~332, end of the 1st par.]{held-09}.
E.g., spin-1 particle flying through a sequence of generalized 
Stern-Gerlach devices whose filters/paths correspond to 3 orthogonal 
{\em eigenprojections\/} of the spin observable 
\cite{anti-shimony} and we would not have a correct measurement
statistics if we ignored some of the vertices present in
particular edges.    

As shown in Fig.~\ref{fig:3-dim} 
Bub's \cite{bub}, Conway and Kochen's \cite{peres-book},
Peres' \cite{peres}, original Kochen and Specker's
\cite{koch-speck} KS sets and Yu and Oh's non-KS set
\cite{yu-oh-12}, have 49, 51, 57, 192, and 25 vertices,
respectively (and 36, 37, 40, 118, and 16 edges, respectively).
In Fig.~\ref{fig:3-dim}, the vertices that share only one edge
are denoted by fully greyed dots and grey ASCII characters. If we
ignored them in an implementation, we would be left with 33, 31,
33, 117, and 13 vertices, but then the measurements would give us
incorrect data as we explained above. Surprisingly, in all
presentations of their KS sets the aforementioned authors simply 
dropped the (grey) vertices that shared only one edge in an attempt
to present their KS sets as being smaller and therefore more
attractive for possible implementations.

Yet, all that vertices/vectors have definite vector components
in the coordinatization they made use of. Thus, it is just the
visual presentation of these KS sets in the original papers and
subsequent reviews in numerous articles and books of these sets
that are misleading, not the actual structures of them (which are
perfectly correct). 

\begin{figure}[hbt]
\includegraphics[width=0.99\textwidth]{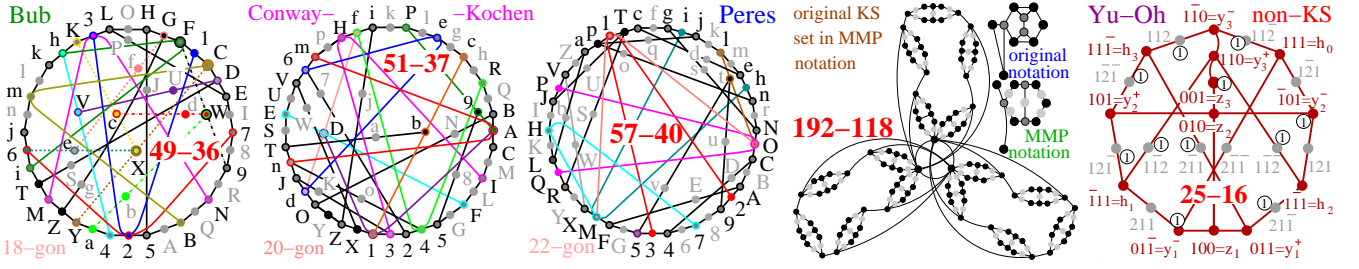}
\caption[3-dim KS Sets]{Four 3-dim KS sets: Bub's 49-36,
  Conway-Kochen's 51-37, Peres' 57-40, and original
  Kochen-Specker's 192-118---and Yu-Oh's non-KS set named
  ``13 vertices (rays) set'' according to 13 red (online;
  black in print) vertices; the components of each vector
  (vertex, ray) are from the set $\{0,\pm 1,\pm 2\}$; \=1 and
  \=2 stand for -1 and -2, respectively.}
\label{fig:3-dim}
\end{figure}

In 2012, Yu and Oh published a paper \cite{yu-oh-12} in which they
introduced a set with 13 vertices which they call a 13-ray
set---13-vertex set in our notation. The set is displayed in
Fig.~\ref{fig:3-dim} where the 13 vertices are shown as red dots
(in online version; black dots in printed version). Yu and Oh dropped
the vertices that share only one edge, shown as grey dots in the
figure (12 of them), following the aforesaid manner. In our figure
we see that after restoring the dropped vertices it is possible to
assign 0s and 1s to vertices from all edges. So, the Kochen-Specker
theorem \ref{th:ks} tells us that Yu-Oh's 13-vertex set is not a KS
set. Actually, Yu and Oh themselves cite the Bell-Kochen-Specker
theorem in the same wording as in Theorem \ref{th:ks} and admit that
their set does not satisfy the conditions of the theorem
\cite[p.~3, top]{yu-oh-12}. That can be formulated as the following
lemma.

\begin{lemma}\label{lemma:13} Yu-Oh's 13-vertex set is not a KS set.
\end{lemma}

\begin{proof}
It is possible to assign 0s and 1s to vertices in such a way
that no two orthogonal directions are both assigned 1 and no three
mutually orthogonal directions are all assigned 0 as shown by
encircled 1s in Fig.~\ref{fig:3-dim}.
\end{proof}

Yu and Oh admit the validity of Lemma \ref{lemma:13} as follows:
``The KS value assignments to the 13-ray set are possible; i.e., no
logical contradiction can be extracted by considering conditions 1
and 2 [of Theorem \ref{th:ks}] only.'' Yet, they entitle
their paper {\em``State-Independent Proof of Kochen-Specker Theorem
with 13 Rays''} and on p.~3 they claim to have nevertheless ``proved
the original KS theorem'' \cite{yu-oh-12}. How come, when the Lemma
\ref{lemma:13} proves the contrary? Well, it seems to be a question
of misapplied terminology. In the paper they proceed to define a
new kind of contextuality through their inequalities (2), (3),
and (4) applied to their 13-non-KS-set, and then they mistakenly claim
that their proof of such a newly defined contextuality amounts to the
proof of the Kochen-Specker theorem. Only a KS set can be a proof of
the KS theorem since a violation of conditions 1 and 2 of the theorem
is tantamount to a definition of a KS set and therefore no non-KS can
prove the theorem since it does not violate them \cite{kochen-15}.
This does not mean that the contextuality Yu and Oh proved for their
13-set is wrong. This only means that via such a contextuality for
their 13-set one cannot prove the Kochen-Specker theorem simply
because the set is not a KS set.

Hence, we are left with the four big KS sets as the only known 3-dim
KS sets. (Gould and Aravind have proven that the so-called Penrose's
3-dim KS set is isomorphic to Peres' one
\cite{peres-penrose-gould-aravind09}.) It is well-known that all of
them are critical and our program {\tt states01} confirms that. So,
we cannot use them to generate smaller KS sets. (By the way, Yu-Oh's
13 vertex set is a subset of Peres' critical set and its criticality
is yet another avenue of proving that Yu-Oh's set cannot be a KS set
and therefore that it cannot prove the KS theorem \cite{kochen-15}.)
But when we look at their MMP hypergraphs we notice that the number of
grey dots, i.e., the number of vertices that share only one edge, 
increases with the total number of vertices within a KS set, 
in contrast to the opposite trend of KS sets from the 4-dim 60-105
class; Cf.~red circles in 29-16, 30-16, 31-18, and 60-40 in
Fig.~\ref{fig:ks-60-105-29-30}. In particular, there are 
16, 20, 24, and 75 such vertices (grey dots) in Bub, Conway-Kochen,
Peres, and Kochen-Specker's hypergraphs in Fig.~\ref{fig:3-dim},
respectively.  

Therefore, we conjecture that a more complex non-critical KS sets,
interwoven similarly to higher dimensional KS sets, with
comparatively low number of vertices, c.a.~50, might be found on
clusters and supercomputers and used to generate smaller 3-dim KS
criticals. This is a work in progress.

\section{\label{sec:discussion}Discussion}

In the past ten years the exploration and generation of
contextual sets, in particular, Kochen-Specker (KS) sets received
a lot of attention (see Sec.~\ref{sec:intro})
both for their possible applications and implementations and for
their further theoretical usage and development in quantum
mechanics and quantum information. The approaches to generation
of KS sets diversified and many partial results were achieved,
recently. Therefore, we have focused our efforts on the
unification of results, features, structure, and mutual relations
of different KS sets as well as on the development of technique and
method of their arbitrary exhaustive generation and handling. 

In pursuing this goal, we have concentrated neither on immediate
experimental implementation (small sets), nor on the standard 
parity proofs based algorithms. Instead, we have made use of the
general MMP hypergraph language by means of which we generated a
large number of new types of contextual KS critical sets and
numerous non-isomorphic instances within each of them, which mostly
cannot be generated by other known algorithms. The approach also
gave us the results we were not originally concerned with, such as
an abundance of novel small KS sets, and in addition provided us
with an explanation why the other approaches, like parity proof
ones, failed to spot them---it turns out that only a very few
KS sets have a parity proof (in some classes under 1\permil\ of
sets and in some none at all) what makes them completely invisible
for the parity proof based algorithms and programs, predominant in
the literature. 

Instead of parity proofs of only some KS sets with a particular
structure, the MMP hypergraph algorithms and methods enable direct
numerical proofs of the KS theorem for any chosen KS set via
literal verifying of KS theorem conditions: program {\tt states01\/}
gives a maximal number of 1s for a chosen KS set and after deleting
all edges that contain these 1s at least one edge should remain.
This is all automated but the user can easily check the output
MMP hypergraph strings by hand. When the MMP hypergraphs are drawn
as figures, the proofs also become ``visual'' as indicated in
Fig.~\ref{fig:24-24-class}(a) and
Fig.~\ref{fig:ks-60-74-26-38}$\,$(26-13) by dashed red ellipses.

We developed the hypergraph approach to KS sets introduced in
Sec.~\ref{sec:for} from the lattice theory of Hilbert spaces and the
way of assigning of 0-1 states to vertices of hypergraphs we took
over from the methods of dealing with discrete states defined on
those lattices. In particular, we redesigned our programs for
analysing the Hilbert lattice features and turned them into the
programs we used to build up an MMP-hypergraph-based language
(specified in Sec.~\ref{sec:for}) for generating, analysing,
filtering, and modifying KS sets. We made use of this language to
obtain numerous novel KS sets and classes and their features in
4-dim (in Secs.~\ref{sec:6074}, \ref{sec:60105}, \ref{sec:300675},
and \ref{sec:4dim-witt}), 6-dim (in Sec.~\ref{sec:6dim}), 8-dim
(in Sec.~\ref{sec:8dim}), 16-dim (in Sec.~\ref{sec:16dim}),
and 32-dim (in Sec.~\ref{sec:32dim}) Hilbert spaces. We also
reviewed the 3-dim KS sets Sec.~\ref{sec:3d}. In the table
in Fig.~\ref{table:overview} we list the most important properties
of critical KS sets we obtained and compare them with ones obtained
previously. 

\begin{figure}[hbt]
\includegraphics[width=0.99\textwidth]{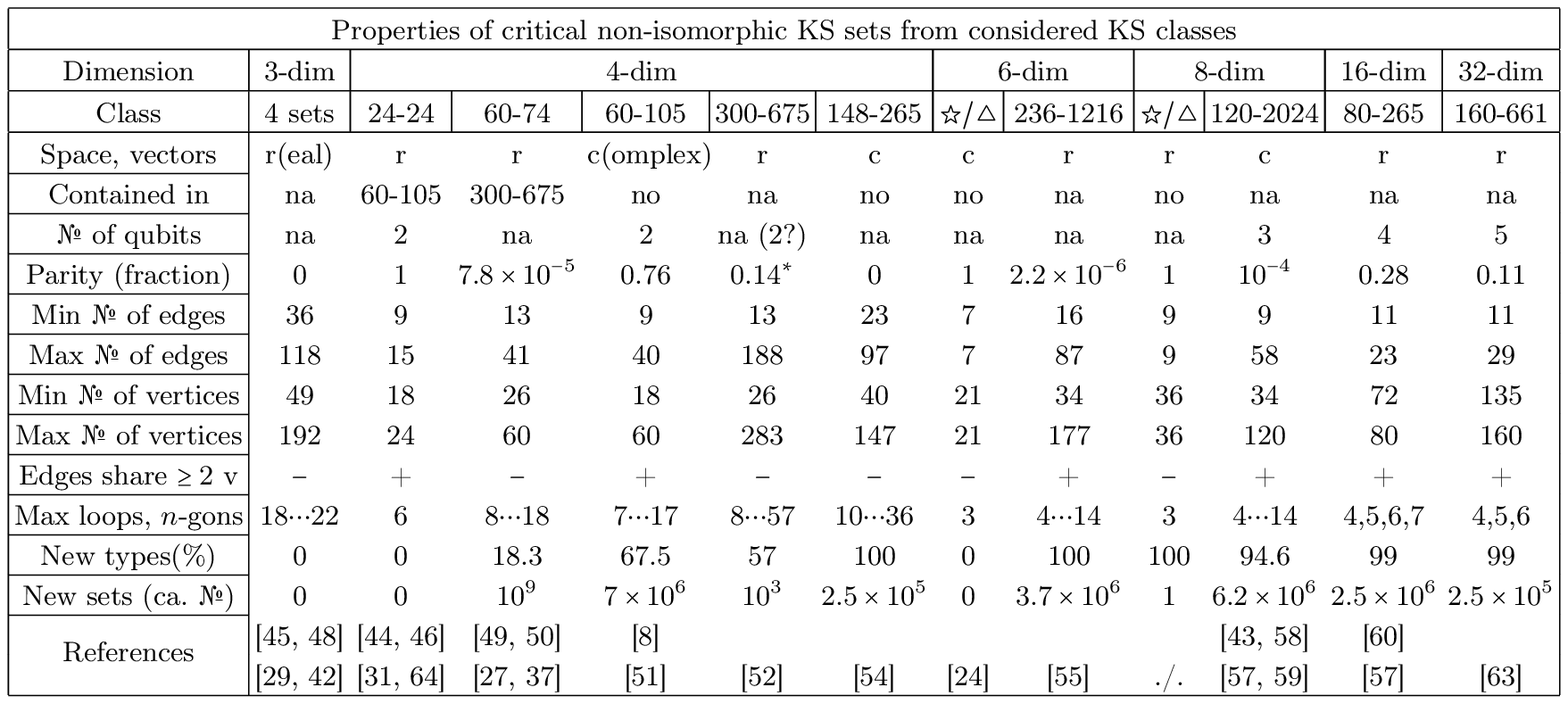}
\caption{List of properties of generated critical KS sets  and their
  comparison with the previously obtained sets; ``na'' stands for
  ``not applicable,'' e.g. for the biggest known classes; ``2?''
  refers to the claim in \cite{waeg-aravind-fp-14} that the operators
  defining the 300-675 might be redefined to allow a representation
  via 2 qubits; $^*$ in 0.14$^*$ for 300-675 indicates unevenly
  distributed parity proofs: all lower criticals have it while all
  higher ones (211-127 to 283-188) lack it;
  ``Edges share $\ge 2$ v'' = edges share 2 or more vertices, i.e.,
  intersect each other two or more times; 
  ``New types'' ``0'' for, e.g., 24-24 means that there are no new types
  of 24-24 KS criticals in this paper---they are only elaborated on
  and discussed here; ``18.3'' (for 60-74) means: 28 (new types
  obtained in this paper for the first time; others were obtained
  in \cite{mfwap-s-11,mp-nm-pka-mw-11,aravind10,waeg-aravind-megill-pavicic-11})
  / 153 (total number of types) = 18.3\%; ``100''(for 236-1216) means
  that all types are from this paper; ``New sets'' give \textnumero\ of
  sets obtained in this paper for the first time.} 
\label{table:overview}
\end{figure}

The large variety of KS sets we obtain provides a much greater 
choice for KS experiments and insight into their properties 
and the properties of gates that handle them as well as the 
properties of quantum sets in general. While smaller KS sets 
are currently preferred for a feasibility of experimental 
implementations, in the future other sets that are intrinsically
different (i.e. non-isomorphic) may become desirable for more
sophisticated  experiments, verification of the KS theorem with
different setups, etc., especially because bigger sets do not require
higher efficiency of measurements but only a higher number of
measurements. Since the sets we found are critical, there will
be no redundancy in any experimental setup making use of them.

Finally, we want to stress that our generation of KS sets in 16-dim
and 32-dim spaces allowing for their implementation by means of 4 and
5 qubits, respectively, is---as, actually, all generations in this
paper---vector/vertex-based and therefore complementary
to a recent operator-based generation of KS sets for 4, 5, and 6
qubits (explicitly, and more of them, in principle) by Waegell and
Aravind \cite{waeg-aravind-pra-13}. Building a correspondence between
the two approaches (via eigenvectors of their operators) is a work
in progress so that we, as of yet, cannot say to which extent our
results overlap. We can only say that we have not obtained KS
criticals with 9 edges (bases, in their terminology) for 4 and
5 qubits, as they have. Our minimal number of edges for the
corresponding two classes is 11, as shown in the tables in 
Figs.~\ref{table:16dim}, \ref{table:32dim}, and \ref{table:overview}. 
We did not include 16-dim and 32-dim KS criticals with 9 edges from 
\cite{waeg-aravind-pra-13} into our tables because we were not able
to find out whether they, respectively, belong to the 80-265 and
160-661 classes or not.

The generation of KS critical sets we presented in this work is,
with our algorithms and programs, straightforward but demanding and
CPU-time consuming. The jobs require cluster and grid calculations and
even with them it might take months to obtain a required or desired
particular set which might be needed in elaboration, confirmation,
or checking particular assumptions about construction, unification,
geometry, computation, or implementation of contextual sets.
We provide samples of the KS criticals in the MMP hypergraph notation
in the Supplemental Material. They are just a tiny fraction of TBs of
data we generated but the reader can obtain any KS sets from us upon
``appropriate'' request (e.g., the 60-74 file has 231 GB).
Also, the reader her/himself may generate any KS set by making
use of our programs that are freely available at our repository
{\tt http://goo.gl/xbx8U2}.

See Appendix \ref{app:B} for MMP hypergraph strings of chosen KS 
critical sets from most of the types from all classes we considered 
in the paper. 

\begin{acknowledgments}
Supported by the Alexander or Humboldt Foundation, the Croatian
Science Foundation through project IP-2014-09-7515, and the Ministry
of Science, Education, and Sport of Croatia through the CEMS funding.
Computational support was provided by the cluster Isabella of the
Zagreb University Computing Centre and by the Croatian National
Grid Infrastructure. The programs that were used in this work
were mostly written by Norman D.~Megill and his upgrading of
the programs in the course of obtaining the reported results is
gratefully acknowledged. P.K.~Aravind and Mordecai Waegell's
sending of their excel formatted versions of the master sets from
\cite{aravind-waegell-6dim-private,waeg-aravind-fp-14,waeg-aravind-jpa-15,waeg-aravind-17-arXiv}
and their communication with us on the details concerning these sets
and papers are gratefully acknowledged. Michel Planat's
sending of his master set from \cite{planat-saniga-12} is also
gratefully acknowledged.
\end{acknowledgments}

\appendix

\parindent=0pt

\section{\label{app:0}MMP Hypergraph Strings Referred
  to in the Article}

\subsection{\label{app:0-1}Section \ref{sec:6074}}

\begin{sloppypar}

{\bf Master set 60-75} {\1{1234, 5678, 9ABC, DEFG, HIJK, LMNO, PQRS, TUVW, XYZa, bcde, fghi, jklm, maSK, lZRJ, kYQI, jXPH, ieWO, hdVN, gcUM, fbTL, nopq, rstu, vwxy, yuqG, xtpF, wsoE, vrnD, pdYC, qeZB, ocXA, nba9, uROC, tQNB, sPM9, rSLA, wUHC, xVIA, yWJ9, vTKB, jgEB, liGA, khF9, mfDC, pLJ8, qMI7, nNH6, oOK5, vhX8, ygY5, xfZ6, wia7, ukc6, rjd7, tlb5, sme8, UQG8, TRF7, VPD5, WSE6, tXW4, rZU3, sYT2, uaV1, ofQ1, piP3, ngR4, qhS2, xjO2, wkL4, vlM1, ymN3, eHF1, cJD2, bIE3, dKG4.}}

\smallskip  
{\bf Master set 60-74} {\1{1234, 5678, 9ABC, DEFG, HIJK, LMNO, PQRS, TUVW, XYZa, WOGC, VNFB, UMEA, TLD9, aSK8, ZRJ7, YQI6, XPH5, bcde, fghi, jklm, mie4, lhd3, kgc2, jfb1, ndIB, oeJA,  pcK9, qbHC, reF8, sdE5, tbD6, ucG7, rfY9, uhXA, tiaB, sgZC, nkT8, plV6, qmU7, ojW5, piRE, ogSD, qhQF, nfPG, smNK, ujLI, rlMH, tkOJ, sTQ1, uVS4, tUP3, rWR2, oYN3,  nZM4, qaL2, pXO1, qpon, utsr, vePL, wcQM, xbRN, ydSO, vkYE, yjZF, wmXD, xlaG, wfVJ, xgUI, yiTH, vhWK, x954, vB71, yA62, wC83.}}

\smallskip  
 {\bf 60-39} {\14123, 3hsS, SQRP, PcwO, OLMN, NiX5, 5786, 6teF, FDEG, GVgp, pnqo, oabA, ABC9, 9fxK, KHIJ, JUd4,,, bcde, ujHZ, 8vTl, 2mMn, 7IRq, EyY1, WLrB, QCkD, IlXt, Uxj7, buE5, n8QH, 1KcT, arKk, hnUE, m5Ww, t1gQ, yHiA, 7Dsc, Of8Y, VbI2, Z2Dv, AtOU.}

\smallskip  
{\bf 60-41} {\12341, 1gQt, t6eF, FjBS, SydM, MLON, NiX5, 5mWw, whCI, I7Rq, qnpo, oAba, arkK, Kfx9, 9EPl, l8Tv, vZD2,,, 5678, 9ABC, DEFG, HIJK, PQRS, bcde, VgpG, hs3S, 2mMn, EyY1, WLrB, JUd4, QCkD, Uxj7, buE5, n8QH, 1KcT, 9Zdq, 7Dsc, CTuM, Of8Y, VbI2, iTFq, BcnX.}
\end{sloppypar}

\subsection{\label{app:0-2}Section \ref{sec:60105}}

\begin{sloppypar}
{\bf Master set 60-105}
{\11234, 5678, 9ABC, DEFG, HIJK, LMNO, PQRS, TUVW, XYZa, bcde, fghi, jklm, nopq, rstu, vwxy, 12FG, 12RS, 13MO, 13UW, 14gh, 14pq, 23fi, 23no, 24LN, 24TV, 34DE, 34PQ, 56JK, 56VW, 57EG, 57Ya, 58kl, 58oq, 67jm, 67np, 68DF, 68XZ, 78HI, 78TU, 9ANO, 9AZa, 9BIK, 9BQS, 9Ccd, 9Cnq, ABbe, ABop, ACHJ, ACPR, BCLM, BCXY, DERS, DFYa, DGde, DGtu, EFbc, EFrs, EGXZ, FGPQ, HIVW, HJQS, HKhi, HKsu, IJfg, IJrt, IKPR, JKTU, LMZa, LNUW, LOlm, LOru, MNjk, MNst, MOTV, NOXY, PSkm, PSxy, QRjl, QRvw, TWce, TWwy, UVbd, UVvx, Xagi, Xavy, YZfh, YZwx, bctu, bdwy, benq, cdop, cevx, ders, fgsu, fhvy, fipq, ghno, giwx, hirt, jkru, jlxy, jmoq, klnp, kmvw, lmst.}

\smallskip
{\bf Master set 24-24 as a subgraph of the master set 60-15 with original complex vector components:} {\11234, 5678, 9ABC, DEAC, FG9B, 9CHI, ABJK, HJKI, DFGE, 58JI, 67HK, 56GE, 78DF, LMNO, N4BC, O39A, 1MKI, 2LHJ, 1LFE, 2MDG, 57O4, 68N3, NO34, 12LM.$\{$1=(1,0,0,0), 2=(0,1,0,0), L=(0,0,1,0), M=(0,0,0,1), 5=(1,1,$i$,$i$), 6=(1,-1,$i$,-$i$), 7=(1,1,-$i$,-$i$), 8=(1,-1,-$i$,$i$), N=(1,1,0,0), O=(1,-1,0,0), 3=(0,0,1,1), 4=(0,0,1,-1), D=(1,0,$i$,0), F=(0,1,0,$i$), G=(1,0,-$i$,0), E=(0,1,0,-$i$), 9=(1,1,$i$,-$i$), A=(1,1,-$i$,$i$), B=(1,-1,$i$,$i$), C=(1,-1,-$i$,-$i$), H=(1,0,0,$i$), J=(1,0,0,-$i$), K=(0,1,$i$,0), $I$=(0,1,-$i$,0)$\}$}

\smallskip
{{\bf 20-11a with coordinatization from 60-105:} \11234, 5678, 19A8, 5BC4, DEFG, HIJK, 6BDF, A7HJ, C3HI, 29DE, HKDG.$\{$1=(1,0,0,0), 2=(0,1,0,0), 9=(0,0,1,0), 5=(1,1,$i$,$i$), 6=(1,-1,$i$,-$i$), B=(1,1,-$i$,-$i$), C=(1,-1,0,0), 3=(0,0,1,1), 4=(0,0,1,-1), A=(0,1,0,$i$), 7=(1,0,-$i$,0), 8=(0,1,0,-$i$), H=(1,1,$i$,-$i$), $I$=(1,1,-$i$,$i$), J=(1,-1,$i$,$i$), K=(1,-1,-$i$,-$i$), D=(1,0,0,$i$), E=(1,0,0,-$i$), F=(0,1,$i$,0), G=(0,1,-$i$,0)$\}$}

\smallskip
{{\bf 18-9 with coordinatization from 60-105:} \11234, 1567, 869A, BACD, ECFG, H3IF, H4B9, 25ED, 87IG.$\{$1=(1,0,0,0), 2=(0,1,0,0), 5=(0,0,1,0), H=(1,-1,0,0), 3=(0,0,1,1), 4=(0,0,1,-1), 8=(1,0,-1,0), 6=(0,1,0,-1), 7=(0,1,0,1), B=(1,1,-1,-1), 9=(1,1,1,1), A=(1,-1,1,-1), E=(1,0,0,-1), C=(0,1,1,0), D=(1,0,0,1), I=(1,1,1,-1), F=(1,1,-1,1), G=(1,-1,1,1)$\}$}

\smallskip
{{\bf 20-11b with coordinatization from 60-105:} \11234, 1567, 2589, ABCD, EFGH, 8IFH, ADI9, J7EH, J6CD, K4GH, K3BD.$\{$1=(1,0,0,0), 2=(0,1,0,0), 5=(0,0,1,0), K=(1,1,0,0), 3=(0,0,1,1), 4=(0,0,1,-1), J=(1,0,-1,0), 6=(0,1,0,-1), 7=(0,1,0,1), A=(1,1,-1,-1), B=(1,-1,-1,1), C=(1,1,1,1), D=(1,-1,1,-1), 8=(1,0,0,-1), I=(0,1,1,0), 9=(1,0,0,1), E=(1,1,1,-1), F=(1,1,-1,1), G=(1,-1,-1,-1), H=(1,-1,1,1)$\}$}
\end{sloppypar}

\subsection{\label{app:0-3}Section \ref{sec:300675}}

\begin{sloppypar}

{\bf 38-19} \
{\14123,3C6L,LVYE,EDGF,FOQP,PZaS,SNRJ,JIKH,HTXM,M794,,,5678,9ABC,B5N1,8OA2,TUVW,baXG,ZYcI,QKbU,cRDW}

\smallskip
{\bf 42-21} \
{\12143,36CP,PIOe,eGcW,WVXU,ULdg,gFHa,aRZY,YbfT,TJNS,S8A2,,,5678,9ABC,DEFG,HIJK,LMNO,479Q,15BR,Qbcd,EVMZ,DKfX.}

\smallskip
{\bf 48-25}
{\11243,36CP,Pcde,eVUX,XQYZ,ZmiE,EFGD,DabO,OMNL,LIhk,kjgR,R5B1,,,5678,9ABC,HIJK,479Q,28AS,SJTU,RVNW,GHfg,Sahi,Fjdl,\break cTYW,SfMl,RKmb.}

\smallskip
{\bf 221-127} {\small\tt ++S++T++U++V, ++V++5++6++7, ++7+j+k+", +"+z+!+`, +`+\^{}+\_++R, ++R++O++P++Q, ++Q+g+i+h, +h+V+e+>, +>+<+=+?, +?|\}+T, +T+R+m+S, +SQ\#++E, ++E++C++D++F, ++Fy+fz, zRt++8, ++8++9++B++A, ++A+b+c+a, +a\{+/J, JN>+M, +M+L++3+|, +|+\{+\textasciitilde +\}, +\}+\&+'+(, +(D+d++L, ++L++K++M++N, ++NU;S, S!+s@, @f+Ed, dcu++H, ++HBI+\$, +\$+\%++4+\#, +\#AL+B, +B+9+A+[, +[+@+\textbackslash +], +]+w+y+x, +x6\^{}E, ECs+o, +o+p+q+r, +r+I+J+K, +K+G++J+H, +H-+:+8, +8VW\&, \&\$+W\%, \%O+P++G, ++GXYZ, Zn+3h, hgxi, iP+t++S {\bf\,,\,,\,,\,} ++G++H++I++J, ++1++2++3++4, +-+/+:+;, +)+*+;+\textasciitilde, +s+t+u+v, +l+m+n++U, +d+e+f++B, +X+Y+Z++A, +W+n+!++A, +V+Z+z++6, +U+i+\textbackslash ++6, +N+O+P+Q, +F+c++6++I, +C+D+E+J, +6+7+8+[, +5+W++2++6, +3+4+(++A, \textasciitilde +1+2++I, \_`\{+v, ]\^{}+*++1, @[\textbackslash +Y, =>?+k, <\textbackslash +2++5, /:;++P, '()*, \#-+u++O, "*+K+>, !+J+Q++R, wx+7++F, vx+U+>, stu<, ru+4+q, pq?++5, o)+`++J, mn++9++T, l*+y+\_, jk[+t, ef(+b, ab@+\{, TU+g++H, Pv+M+\^{}, NR+I++M, Mq++M++O, KL-++7, w+P+=++U, \&:+U+\%, HI'+x, GNry, EFQ(, M/+4+E, QT+a+', Tf+l++D, Su+D++Q, CD+J++8, F+f+[+\_, t+3+B++N, s+K++F++K, u+\#+>++C, 9+3+J+f, 8k++F++V, 7+p++8++K, 6=+@++B, \^{}+O+d+j, 5H++P++V, IKh++R, 79O++C, H\{+A++B, 5n=+X, JZ+<++S, MS]++D, 9L++9++L, \$\^{}+[++E, G+o+(++C, k+3+T+\^{}, +I+V+(++S, Or+K+d, 4=+w++E, 3G+e++K, 2:+X++5, 1Cw+\^{}.}
\end{sloppypar}

\subsection{\label{app:0-4}Section \ref{sec:4dim-witt}}

\begin{sloppypar}

{\bf 40-23} \
{\small\tt 3124, 49AB, BaVL, LMNE, ECD8, 8567, 7IJK, KRSQ, QOPH, HFG3 ,,, TUV6, WXP5, YZSD, abXG, cZN2, deR1, ebMJ, dWCA, cUIF, cbOC, daYI, WSLF, eZVP.} \ \ {\small  {\tt 1}=(1,0,$\omega^2$,$\omega^2$), {\tt 2}=(1,$\omega$,0,-$\omega^2$), {\tt 3}=(1,-$\omega$,-$\omega^2$,0), {\tt 4}=(0,1,-$\omega$,$\omega$), {\tt 5}=(1,$\omega$,0,-$\omega$), {\tt 6}=(1,0,$\omega$,$\omega$), {\tt 7}=(0,1,-1,1), {\tt 8}=(1,-$\omega$,-$\omega$,0), {\tt A}=(1,0,0,0), {\tt B}=(0,1,-$\omega^2$,1), {\tt C}=(0,0,0,1), {\tt D}=(1,-1,-$\omega^2$,0), {\tt E}=(1,-$\omega^2$,-1,0), {\tt F}=(1,$\omega$,0,-1), {\tt G}=(0,1,-$\omega$,$\omega^2$), {\tt H}=(1,0,$\omega^2$,1), {\tt I}=(1,0,1,1), {\tt J}=(1,1,0,-1), {\tt K}=(1,-1,-1,0), {\tt L}=(1,$\omega^2$,0,-$\omega^2$), {\tt M}=(0,1,-$\omega$,1), {\tt N}=(1,0,1,$\omega^2$), {\tt O}=(1,-$\omega^2$,-$\omega^2$,0), {\tt P}=(1,$\omega^2$,0,-1), {\tt Q}=(0,1,-1,$\omega$), {\tt R}=(1,0,1,$\omega$), {\tt S}=(1,1,0,-$\omega$), {\tt U}=(0,1,-$\omega^2$,$\omega^2$), {\tt V}=(1,-$\omega^2$,-$\omega$,0), {\tt W}=(0,0,1,0), {\tt X}=(1,1,0,-$\omega^2$), {\tt Y}=(1,0,$\omega^2$,$\omega$), {\tt Z}=(0,1,-$\omega^2$,$\omega$), {\tt a}=(1,0,$\omega$,$\omega^2$), {\tt b}=(1,-1,-$\omega$,0), {\tt c}=(1,-$\omega$,-1,0), {\tt d}=(0,1,0,0), {\tt e}=(1,0,$\omega$,1)}

\smallskip
{\bf 49-27} \
{\small\tt 3241, 1675, 5XHW, WCLn, nEmj, jiIc, caZb, blJD, DghO, OPQN, NRS9, 98A3 ,,, 3BCD, 3EFG, 4HIJ, 4KLM, TUVS, 5YFZ, aUGM, dOEH, deCZ,  dRfM, geXK, 6kfh, 6UCI, lkFL, 7RBm.} \ \ {\small{\tt 1}=(0,0,0,1), {\tt 3}=(1,0,0,0), {\tt 4}=(0,0,1,0), {\tt 5}=(1,1,-$\omega$,0), {\tt 6}=(1,$\omega^2$,-$\omega^2$,0), {\tt 7}=(1,$\omega$,-1,0), {\tt 9}=(0,1,-1,-$\omega^2$), {\tt B}=(0,1,$\omega^2$,$\omega$), {\tt C}=(0,1,1,1), {\tt D}=(0,1,$\omega$,$\omega^2$), {\tt E}=(0,1,1,$\omega$), {\tt F}=(0,1,$\omega$,1), {\tt G}=(0,1,$\omega^2$,$\omega^2$), {\tt H}=(1,-1,0,$\omega$), {\tt I}=(1,-$\omega^2$,0,$\omega^2$), {\tt J}=(1,-$\omega$,0,1), {\tt K}=(1,-$\omega^2$,0,1), {\tt L}=(1,-$\omega$,0,$\omega$), {\tt M}=(1,-1,0,$\omega^2$), {\tt N}=(1,-1,-1,0), {\tt O}=(1,0,1,-$\omega$), {\tt R}=(1,0,1,-$\omega^2$), {\tt S}=(1,1,0,$\omega^2$), {\tt U}=(1,0,$\omega^2$,-$\omega^2$), {\tt W}=(1,0,$\omega$,-$\omega$), {\tt X}=(0,1,$\omega$,$\omega$), {\tt Z}=(1,-1,0,1), {\tt a}=(1,1,-$\omega^2$,0), {\tt b}=(1,0,$\omega^2$,-1), {\tt c}=(0,1,$\omega^2$,1), {\tt d}=(1,1,-1,0), {\tt e}=(1,0,1,-1), {\tt f}=(0,1,1,$\omega^2$), {\tt g}=(1,$\omega^2$,-1,0), {\tt h}=(1,-$\omega^2$,0,$\omega$), {\tt j}=(1,0,$\omega$,-$\omega^2$), {\tt k}=(1,0,$\omega^2$,-$\omega$), {\tt l}=(1,$\omega$,-$\omega^2$,0), {\tt m}=(1,-$\omega$,0,$\omega^2$), {\tt n}=(1,$\omega$,-$\omega$,0)}

\end{sloppypar}

\subsection{\label{app:0-6}Section \ref{sec:8dim}}

\begin{sloppypar}

{\bf 8dim $\mathbf\medstar$} \
{\small\tt 12345678,89ABCDEF,FGHI4JKL,L7MNBOPQ,QERSI3TU,UK6VNAWX,XPDYSH2Z,ZTJ5VM9a,aWOCYRG1.}

\smallskip
{\bf 8dim $\mathbf\bigtriangleup$} \
{\small\tt 12345678, 89ABCDEF,FGHIJKL1,L2MRVYaE,KM3NSWZD,JRN4OTXC,IVSO5PUB,HYWTP6QA,GaZXUQ79.},

\smallskip
{\bf Coordinatization for $\mathbf\bigtriangleup$} \
 {\small {\tt 1}=(0,0,0,0,0,0,0,1), {\tt 2}=(0,0,0,0,0,0,1,0), {\tt 3}=(0,0,0,0,0,1,0,0), {\tt 4}=(0,0,0,0,1,0,0,0), {\tt 5}=(0,0,1,1,0,0,0,0), {\tt 6}=(0,0,1,-1,0,0,0,0), {\tt 7}=(1,1,0,0,0,0,0,0), {\tt 8}=(1,-1,0,0,0,0,0,0), {\tt 9}=(0,0,0,0,0,0,1,1), {\tt A}=(0,0,1,1,1,-1,0,0), {\tt B}=(1,1,0,0,0,0,-1,1), {\tt C}=(1,1,0,0,0,0,1,-1), {\tt D}=(0,0,1,0,-1,0,0,0), {\tt E}=(0,0,0,1,0,1,0,0), {\tt F}=(0,0,1,-1,1,1,0,0), {\tt G}=(0,0,0,1,1, 0,0,0), {\tt H}=(0,0,1,1,-1,1,0,0), {\tt I}=(1,0,0,0,0,0,1,0), {\tt J}=(0,0,1,0,0,-1,0,0), {\tt K}=(1,0,0,0,0,0,-1,0), {\tt L}=(0,1,0,0,0,0,0,0), {\tt M}=(0,0,1,0,1,0,0,0), {\tt N}=(0,0,0,1,0,0,0,0), {\tt O}=(1,-1,0,0,0,0,-1,-1), {\tt P}=(0,0,0,0,1,1,0,0), {\tt Q}=(1,-1,0,0,0,0,-1,1), {\tt R}=(1,0,0,0,0,0,0,1), {\tt S}=(0,1,0,0,0,0,0,-1), {\tt T}=(0,1,0,0,0,0,-1,0), {\tt U}=(0,0,1,-1,1,-1,0,0), {\tt V}=(0,0,1,-1,-1,1,0,0), {\tt W}=(1,1,0,0,0,0,1,1), {\tt X}=(0,0,1,0,0,1,0,0), {\tt Y}=(1,0,0,0,0,0,0,-1), {\tt Z}=(1,-1,0,0,0,0,1,-1), {\tt a}=(0,0,1,1,-1,-1,0,0)}

\smallskip
{\bf Coordinatization for $\mathbf\medstar$} \
{\tt 1-8}---as for $\mathbf\bigtriangleup$, {\small {\tt 9}=(1,1,0,0,0,0,-1,1), {\tt A}=(0,0,1,1,1,-1,0,0), {\tt B}=(0,0,0,0,0,0,1,1), {\tt C}=(0,0,1,-1,1,1,0,0), {\tt D}=(0,0,0,1,0,1,0,0), {\tt E}=(0,0,1,0,-1,0,0,0), {\tt F}=(1,1,0,0,0,0,1,-1), {\tt G}=(0,0,1,0,0,-1,0,0), {\tt H}=(1,0,0,0,0,0,0,1), {\tt I}=(0,0,0,1,0,0,0,0), {\tt J}=(1,-1,0,0,0,0,-1,-1), {\tt K}=(0,1,0,0,0,0,-1,0), {\tt L}=(0,0,1,0,0,1,0,0), {\tt M}=(0,0,1,-1,1,-1,0,0), {\tt N}=(1,-1,0,0,0,0,-1,1), {\tt O}=(0,0,0,1,1,0,0,0), {\tt P}=(0,0,1,1,-1,-1,0,0), {\tt Q}=(1,-1,0,0,0,0,1,-1), {\tt R}=(1,0,0,0,0,0,-1,0), {\tt S}=(0,0,1,0,1,0,0,0), {\tt T}=(0,1,0,0,0,0,0,-1), {\tt U}=(1,1,0,0,0,0,1,1), {\tt V}=(0,0,0,0,1,1,0,0), {\tt W}=(0,0,1,1,-1,1,0,0), {\tt X}=(1,0,0,0,0,0,0,-1), {\tt Y}=(0,1,0,0,0,0,0,0), {\tt Z}=(0,0,1,-1,-1,1,0,0), {\tt a}=(1,0,0,0,0,0,1,0).}

\end{sloppypar}

\subsection{\label{app:0-7}Section \ref{sec:16dim}}

\begin{sloppypar}

{\bf 16dim 80-21}
{\small{\tt Zbhjprsv\$(*-:<=@, HIPQZbdehjlmpqvx, Zehmoqwxz!"\#')>?, 12457BCGLMPQZbdm, 378GIJKLOPUVXgik, 12346789ACEGJKOV, 16ACNOVWXYakny!\#, HIKLMPQRTUVWbhlm, KRVWXYabcfghiklm, 129AKRVWbhlmqrv", 37BFXYfgstvwxy"\#, 123456789ABCDEFG, 23ABCDEG\$\%:;<>?@, Xafiuwy!\$()/:;=?, 46CEpsuwxyz!\$(/;, ILQTnot\#\%\&'*<=?@, 5BDFLMNRSTUWrsz", 5BDFHINPQRSWpqvx, XYacfgik\$\%\&(*/;@, HIKLMOPQSTUW\&(;@, 16ACXYadeklm\%\&/;.}}

\smallskip
{\bf 16dim 80-20}
{\small\tt 123456789ABCDEFG, HIJKLMNOPQRSTUFG, VWXYZabcdeRSTUDE, fghijklmndeQABCG, opqrstmn456789BC, uvwxyz!"jklbce3E, \#\$\%\&()"ilYZacde, *-z!stghjkmn89BC, /:;\textless qrstfghXPQ2D, =;\textless xyoprtWOU5679, \textgreater ?-\$\%\&'()w!WZaOT, ()vwz!"plnceMN7C, @?;\textless *-vwxyz!ghjk, \textgreater =:fVWIJKLMNRS1A, \textless 'y"rshikYaeLS49, \textgreater =\textless \&)xgkVWHKNOTU, ?*'v!ofnZJPS17BG, /uXbHOPQTU23DEFG, :\#\&(u"XYcINP12EF, @?*-\&(uvwz!XIN2E.}
    
\smallskip
{\bf 16dim 80-19a}
{\small\tt 123456789ABCDEFG, HIJKLMNOPQRSDEFG, TUVWXYZabcdefgRS, hijklmnofgNOPQFG, pqrsdeLM9ABCDEFG, tuvwpqrsdeJKLMPQ, xyz!"\#vwbc5678BC, \$\%\&'()xyz!"\#Zabc, *-/\!:()"\#smnoXYMO, ;\textless \!/\!:rklmnoMN3478, =\textless :\&')z!\#loWacIS, \textgreater ?*-hijnUVXYfgNO, \textless /\%)!"uwqjkoegAC, ?=-')!\#uvVWYIS9C, =*-/:moTXYHS2468, \textgreater *twrnUXMN3478AB, @UVWXYZabcIR1357, @\%\&(xy"TZcHKQREG, ?*pjUYefJQ2358EF.}

\smallskip
{\bf 16dim 80-19b}
{\small\tt 123456789ABCDEFG, HIJKLMNO9ABCDEFG, PQRSTUVWXYZaNO78, bcdefghijklmLM56, nopqfghijklmXYZa, rstuvwxyVWKM46FG, z!"\#vwxy12345678, "\#xypqjklmZaJO38, \$\%\&'(uimYaJN28EG, )*-\!/\!:'\!(thmUWYZCD, ;\textless /:!\#stwyIJ23BD, :z\#swxoqdeglRSTW, =\textgreater *-\%\&'(bcdePQST, :\$rsoqceglQRTWAG, ?@\textgreater \textless )-/\$\&(ruHM45, \textgreater *\%(opflPQSTXYZa, ?@;\textless )/:\$rstuRUVW, *-:'(!"tvybcPQUW, \textless )beQSHKLM14569D.}

\smallskip
{\bf 16d 72-11 4-gon}
{\small\tt 34AC1256789BEFGD, DEFGXYZabcdeUVWT, TUVWHJKLMOPQRSNI, IN*-\&\#\$zswnq4AC3 ,,, fghijklmnopqreSW, stuvwklmnopqrdRV, xyz!uvwqrcOPQVBC, "\#\$\%hijopbLMNW9A, \&'(\%twnrZaKN78FG, )'(!gjmpXYJQ56FG, *-)"xyfjmoHQ12AC.} 

\smallskip
{\bf 16d 72-11 5-gon}
{\small\tt 9ABC12345678EFGD, DEFGXYZabcdTUVWe, eWxyz!"\#\$Mstuvwr, rstuvwnopqNOmRSl, lmRSfghijkPQABC9 ,,, HIJKLMNOPQRSTUVW, \%\&'!"\#\$pqtuvwdLV, ()z\$oswkbcKS78FG, ()'\#nqvjZaJR56FG, *-\%\&\#quhiYIV34BC, *-xy\$stfgXHW12BC.} 

\smallskip
{\bf 16d 74-13 6-gon}
{\small\tt 12FG34569ADE8BC7, 78BC()*'"\#bcghia, abcghiXYZdefklmj, jklmnopqrstuUVWT, TUVWHKLNOQJMPRSI, IJMPRS:-/\$vw2FG1 ,,, vwxyrstulmVWDEFG, z!"\#pqtukmUW9ABC, \$\%\&'defghiNOPQRS, -/*\&xyYZcfhi56FG, :()\%z!KLMQRS34BC, XZbcefgiHJLMOPQS, notuZbeiklJLOSUV.} 

\smallskip
{\bf 16d 76-15 7-gon}
{\small\tt imghjVWM7FlXPACk, klXPAC)*-\$\%S:\&'/, /:\&';\textless "\#b1NO9BGx, xNO9BGdHQ2DEuJLq, quJLnrIKpst3456o, opst3456wcfTUYZ(, (wcfTUYZyvaemghi ,,, 123456789ABCDEFG, HIJKLMNOPQBCDEFG, RSTUVWXYZOPQAEFG, abcdefghYZNQ89DG, vwrstumefhUZKL56, yz!"\#\$\%\&'(mefgUY, -wprtujfRUXZKL46, ;\textless )*z!\&'npsuJK36.} 

\end{sloppypar}

\subsection{\label{app:0-8}Section \ref{sec:32dim}}

\begin{sloppypar}

{\bf 160-21}
{\small\tt 123456789ABCDEFGHIJKLMNOPQRSTUVW, 3XY45Za8b9CDEcdFIJKefghiMRjklmTW, nopqrstuvwxyz!"\#9IJK\$\%\&jklm'()*-, /:;\textless =\textgreater ?@[\textbackslash ]\^{}\_`\{|\}\textasciitilde +1+2+3+4+5+6+7+8+9+A+B+C+D+E, +F+G+H+I+J+K+L+M23+N4+O679AB+PDFGJ+QLNOPQSUV, 2+NXY45+O679A+PDEcdIJ+QKL\$Q\&jklm'SV-, 123XYa8bAB+PDEJ+QK\textasciitilde +R+S+T+3+4+5+U+8+V+W+X+B+C+D+Y, noqruwy"38CFfiM+Z\%O+ajkl+bmS(T)U*+cV, +G+H+d+J+e+K+f+g+NXY+OZab+Pcd+Qefghi+Z\$\%+ajm'),  /+h;+i+j\textgreater +k@+l\textbackslash +m\^{}\_+n\{+o+F+Gop+d+pqr+J+qv+ewx!+M, nqrvxyz\#13+NY457Z8b9ABCDEdFGHIJ+QK, 2679AIJKLMN\$\%OPQR\&jklm'S(T)U*V-W, +r\textbackslash \^{}+s+n`|+o2367Za8bACFeN\%+ajl(T)U*+cW,\break 
+t/+h:+j=+u\textgreater +l[+r\textbackslash +v\_+n`+w+R+S+2+T+3+4+x+U+6+7+8+9+D+y+z, +t/;+!+l[]+m13XY46ZaABCEcdGHM+Z+aR+bT+cW, \textasciitilde +R+2+3+U+6+"+\#+8+V+A+B+Y+E+y+zfg+ZN\$OQRjkS(T)+c-, +u\textgreater @+\$+l[]+m+v\_+n`\{+\%|+o1+NY+PcGIJgLi'S)*-, +F+\&no+H+'+dstv+e+Kwxy+M+w+R+2+4+5+x+U+6+(+V+A+C+D+)+Y+E, 3XY4589CDEcdFIJKfghLiMQRjklmSTVW, +F+\&no+'p+pq+Irs+qtuvwx+Lyz!+M"\#\&kl+b(*+c-, +t+h+!+i+j+u+k+\$+l+r+m+s+v+n+\%+o+w+R+S+T+x+U+"+\#+(+V+W+X+)+Y+y+z.}

\end{sloppypar}

\section{\label{app:B}Samples of KS Criticals in MMP
  Hypergraph Notation}

\small{Whenever possible, an attempt has been made to show the sets not
already drawn above. Smaller sets are presented in the MMP hypergraph
figure-to-draw format: the edges before ``{\tt ,,,}'' build the
maximal loop; ``{\tt *}'' following a character/vertex means that
the vertex belongs to an edge from the loop. Thus, the reader can
easily turn each of the following smaller MMP hypergraph strings into an
MMP hypergraph figure reversing the procedure we explained in the body
of the paper on how to turn any MMP hypergraph figure into an MMP
hypergraph string. Bigger sets are presented just as they were generated
(without maximal loop denotations)---only a space is inserted after each
comma for easier formatting; these spaces should be taken out before
processing.}

\subsection{\label{app:A-1}3-dim}

\linespread{0.5}
\smallskip
\small{Bub {\bf 49-36} (18-gon) {\tt 425,5AB,BQN,NR9,987,7IE,EDC,C1F,FGH,HOL,L3K,Khk,klm,mnj,j6i,iTM,MZY,Ya4,,,1*2*3*,2*6*7*,\break 5*JH*,3*M*N*,L*PE*,D*ST*,D*UV,F*9*W,XY*C*,a*bW,XK*c,cdW,X6*e,efG*,G*gZ*,h*i*F*,h*4*V,C*B*m*.}

Conway-Kochen {\bf 51-37} (20-gon) {\tt 132,245,5GF,FLI,IMC,CAB,BQR,Rhc,cge,elP,Pki,ifH,Hpm,m6V,VUE,EST,TnJ,JdO,OYZ,\break ZX1,,,2*6*7,5*89,DE*F*,3*H*I*,3*J*K,B*NO*,A*9P*,V*WX*,T*ab,1*bc*,d*6*e*,Dd*f*,f*4*R*,Z*ji*,m*n*A*,o4*n*,\break 1*A*D.} 

Peres {\bf 57-40} (22-gon) {\tt 534,467,789,92A,ABC,CON,Nrn,nhe,eml,lkj,jig,gfc,cT1,1pa,aZV,VPJ,JIH,HKL,LQR,RYX,XMF,\break FG5,,,1*2*3*,C*D4*,A*EF*,H*7*M*,O*P*Q*,R*ST*,T*UJ*,V*WX*,L*ba*,c*de*,M*h*i*,n*op*,p*qj*,g*sN*,tu9*,tl*O*,\break tv5*,1*M*O*.}

\smallskip
\begin{sloppypar}
Kochen-Specker {\bf 192-118} (18-gon) {\tt 123, 345, 567, 789, 9AB, BC1, DEF, FGH, HIJ, JKL, LMN, NOD, PQR, RST, TUV, VWX, XYZ, ZaP, bcd, def, fgh, hij, jkl, lmb, 1no, opq, qrs, stu, uvw, wx1, yz!, !"\#, \#\$\%, \%\&', '(), )*y, -/:, :;\textless , \textless =\textgreater , \textgreater ?@, @[\textbackslash , \textbackslash ]-, \textasciicircum \_`, `\{\textbar , \textbar \}\textasciitilde , \textasciitilde +1+2, +2+3+4, +4+5\textasciicircum , +6+7+8, +8+9+A, +A+B+C, +C+D+E, +E+F+G, +G+H+6, y+I+J, +J+K+L, +L+M+N, +N+O+P, +P+Q+R, +R+Sy, +T+U+V, +V+W+X, +X+Y+Z, +Z+a+b, +b+c+d, +d+e+T, +f+g+h, +h+i+j, +j+k+l, +l+m+n, +n+o+p, +p+q+f, +r+s+t, +t+u+v, +v+w+x, +x+y+z, +z+!+", +"+\#+r, +\$+\%+\&, +\&+'+(, +(+)+*, +*+-+/, +/+:+;, +;+\textless +\$, +T+=+\textgreater , +\textgreater +?+@, +@+[+\textbackslash , +\textbackslash +]+\textasciicircum , +\textasciicircum +\_+`, +`+\{+T, 4A+\textbar , GM+\}, SY+\textasciitilde , ek++1, pv++2, "(++3, ;[++4, \{+3++5, +9+F++6, +K+Q++7, +W+c++8, +i+o++9, +u+!++A, +'+:++B, +?+\_++C, D7y, PJy, bVy, shy, -\%+T, \textasciicircum \textgreater +T, \textasciitilde +C+T, +N+6+T, +f+Z1, +r+l1, +\$+x1, +\textbackslash +*1, 1y+T.}
\end{sloppypar}
}
  
\subsection{\label{app:A-2}4-dim 24-24}

\small{Maximal loops of all KS sets from the 24-24 class are hexagons.}

\smallskip
\small{Master set {\bf 24-24} {\tt 3241,1576,69LI,IHJD,DEGF,FBN3,,,2*89*A,5*8B*C,E*KL*M,H*KN*O,7*AJ*M,4*CG*O,2*5*E*H*,4*B*J*L*,\break 7*9*G*N*,1*8D*K,3*CI*M,6*AF*O,1*2*5*8,I*J*L*M,F*G*N*O,6*7*9*A,3*4*B*C,D*E*H*K.}}

\smallskip
\small{{\bf 18-9} {\tt 1234,4567,789A,ABCD,DEFG,GHI1,,,I*2*9*B*,3*5*C*E*,6*8*F*H*.}
  
{\bf 20-11}  {\tt IHJK,KBCG,GDFE,E7A8,8154,429I,,,9*A*F*J*,5*67*8*,1*2*34*,34*C*D*,68*B*H*.}

{\bf 20-11} {\tt HIJK,KEGF,F257,76D3,3ACB,B18H,,,C*D*G*J*,8*9A*B*,45*6*7*,2*9B*E*,1*47*I*.}

{\bf 22-13} {\tt JKLM,M29A,A486,63CE,EBFI,I1HJ,,,F*GH*I*,B*C*DE*,78*9*A*,3*4*56*,2*DE*L*,1*56*K*,7A*GI*.}

{\bf 22-13} {\tt KJLM,MABE,ECDI,IFGH,H264,415K,,,9B*H*L*,8D*G*K*,7A*C*E*,6*9F*H*,5*8J*K*,1*2*34*,34*7E*.}

{\bf 24-15} {\tt 43FG,G7E8,86IK,KAHC,C9LO,O1N4,,,L*MN*O*,H*I*JK*,DE*F*G*,9*A*BC*,56*7*8*,1*23*4*,BC*DG*,24*JK*,58*MO*.}}

\subsection{\label{app:A-3}4-dim 60-74}

\small{Master set is given in Appendix \ref{app:0} of the paper.}

\smallskip
\small{{\bf 26-13} {\tt ONQP,P2LB,BCDE,E8A9,91M6,657I,IGHF,F3KO,,,JK*L*M*,3*47*D*,2*4A*H*,1*C*G*Q*,5*8*JN*.}

{\bf 30-15} {\tt LJKM,MIQH,H9T6,6EN5,51R3,3AD4,42S7,7BO8,8CUF,FGPL,,,R*S*T*U*,N*O*P*Q*,C*D*E*K*,9*A*B*J*,1*2*G*I*.}

{\bf 36-19} {\tt ZXYa,aNOP,P67A,A89S,SQRW,WTUV,VKLM,MBCD,D45G,GEFJ,JHIZ,,,35*U*Y*,2C*I*O*,17*F*L*,26*E*R*,19*H*N*,\break 13B*Q*,24*8*K*,12T*X*.}

{\bf 38-19} {\tt aZcb,b34O,OPRQ,QIJK,KCDE,E56W,WVXY,YTUS,S2HA,ABN9,917a,,,LMN*R*,FGH*P*,7*8MX*,4*8B*U*,3*6*J*T*,\break 1*D*GV*,C*FLc*,2*5*I*Z*.}

{\bf 38-21} {\tt aZbc,cNOR,RPQU,USTY,YVWX,XDEI,IGHF,F5C3,328J,JKML,L9Ba,,,C*E*M*b*,AB*H*W*,8*B*O*T*,67G*Q*,5*7K*V*,\break 47AZ*,13*6S*,3*49*D*,1B*E*P*,2*7E*N*.}

{\bf 38-22} {\tt aZbc,cHIR,RPXQ,Q6G8,84CJ,JKVL,L57U,USYT,T2FA,A3BM,MNWO,ODEa,,,V*W*X*Y*,F*G*L*b*,B*C*U*Z*,9A*I*K*,\break 7*8*H*N*,5*A*E*P*,18*D*S*,13*L*R*,2*4*O*R*,6*9O*U*.}

{\bf 39-23} {\tt ZWXY,YOPN,N2Tb,b8I7,713B,BCGD,DJUc,c9AL,LKMS,SRVQ,Q4H6,65dZ,,,ab*c*d*,T*U*V*Z*,H*I*J*P*,EFG*d*,\break C*M*X*b*,B*O*R*a,4*8*FK*,3*A*EH*,1*Q*W*c*,2*4*9*B*,3*5*K*N*.}

{\bf 40-23} {\tt dbce,eYaZ,ZTUS,SMWI,IHKJ,JEOD,D6B5,548N,NGXA,A273,3FL1,1CVP,PQRd,,,V*W*X*a*,L*M*N*O*,F*G*K*R*,\break B*C*U*c*,9A*E*T*,7*8*C*H*,6*9H*b*,3*5*Q*Y*,4*9F*W*,2*D*P*S*.}

{\bf 49-30} {\tt lknm,mMTX,X56F,FAIb,bBNj,jgih,h2SU,U4ZJ,JDOG,G1YV,V3H8,879Q,QPRe,ecfd,dCaW,WKLl,,,Y*Z*a*b*,U*V*W*X*,\break S*T*b*f*,N*O*R*n*,H*I*J*M*,EF*G*L*,D*Q*b*k*,A*P*V*i*,9*B*X*c*,6*C*H*S*,4*9*A*K*,7*EN*U*,1*5*K*N*,1*C*Q*g*.}

{\bf 51-29} {\tt omnp,pjlk,kGW9,96D7,73AU,USXT,T8BO,OLNM,MHVY,YZba,aIJK,KCER,RPQe,ecfd,d125,54Fg,ghio,,,V*W*X*f*,\break b*d*i*l*,F*G*H*U*,D*E*H*n*,A*B*C*c*,8*K*h*j*,5*N*Q*S*,A*M*P*g*,3*E*L*k*,9*B*Q*Z*,3*4*8*Y*,5*7*V*m*.}

{\bf 54-30} {\tt qprs,smno,o8fD,DABC,C7c6,659O,OMNL,L4bF,FEGR,RQgl,lijk,kPda,aYZX,X3eI,IJKH,H12S,SThW,WUVq,,,e*f*g*h*,\break b*c*d*r*,9*K*P*T*,7*8*V*Z*,4*5*8*J*,2*3*N*p*,4*Q*U*Y*,6*S*j*n*,3*B*G*P*,1*7*M*Q*,I*L*i*m*,2*9*D*F*.}

{\bf 55-31} {\tt sqrt,tnpo,oYZX,X4Kb,bacd,dUWV,VLhE,E56e,efjg,gFMG,GAC8,823I,IHJN,NBi9,91DP,POQT,TRSm,mkls,,,h*i*j*p*,\break K*L*M*N*,D*E*G*c*,C*S*Z*r*,C*J*Q*b*,A*B*Y*a*,I*W*X*f*,78*9*R*,4*P*V*q*,3*4*B*F*,1*M*S*n*,1*3*U*e*,1*A*H*h*.}

{\bf 57-31} {\tt tsuv,vpqr,rHea,aYZb,bInT,TRSU,UJok,kjli,i7m1,148K,KLQM,MFGE,E265,5Nd9,9ABD,DCPX,XVWh,hfgt,,,m*n*o*u*,\break cd*e*l*,N*OP*Q*,8*B*G*O,7*B*H*J*,6*A*I*L*,A*F*S*W*,4*9*R*g*,39*V*Z*,2*3f*j*,1*A*cq*,2*B*Y*p*,1*2*C*s*.}

{\bf 58-32} {\tt utvw,wqrs,shim,mjlk,kQRY,YXcg,gdef,fSUT,TAVM,MLbN,NFGW,W129,97Z8,86nC,CBaK,KHIJ,J3oD,DEPu,,,n*o*pv*,\break Z*a*b*c*,V*W*l*r*,OP*U*p,E*R*S*i*,9*Q*T*t*,456*c*,8*D*e*h*,3*C*X*j*,2*B*D*L*,3*7*d*i*,2*A*Od*,1*6*J*M*,\break C*OR*q*.}

{\bf 60-40} {\tt xvwy,ysut,tgih,hABZ,ZYfa,akpD,D4nS,SUlT,TLWC,CMrQ,QPRX,Xeom,mbcq,q8N7,79KE,E2F3,3JV6,6Ojx,,,p*q*r*x*,\break n*o*r*u*,j*k*l*m*,de*f*i*,V*W*X*p*,M*N*O*e*,J*K*L*e*,HIR*s*,F*Go*w*,E*IU*x*,K*S*c*w*,9*B*GO*,A*J*n*v*,\break 56*T*o*,4*9*m*s*,6*D*b*h*,6*8*f*s*,3*N*P*u*,7*A*C*d,5B*Q*y*,3*C*g*m*,1D*e*y*.}}

\subsection{\label{app:A-4}4-dim 60-105}

\small{Master set is given in Appendix \ref{app:0} of the paper.}

\smallskip

\small{{\bf 22-11} {\tt JKLM,M135,56B9,97C8,824F,FGIH,HIED,DAKJ,,,A*B*C*E*,3*4*6*G*,1*2*7*L*.}

{\bf 25-13} {\tt MNPO,OPKL,L89D,DCIG,GJFE,EFBA,A374,4165,56HM,,,H*I*J*N*,7*9*B*K*,24*8*C*,1*23*4*.}

{\bf 27-15} {\tt OPRQ,QRMN,NCDG,GBA9,9A68,8154,42EH,HIJL,LKPO,,,E*FG*J*,B*D*FG*,5*6*78*,1*2*34*,78*I*K*,34*C*M*.}

{\bf 28-15} {\tt PQSR,RSAF,FDEG,GCON,NOML,LMKB,B3I4,4158,86H9,9JQP,,,H*I*J*K*,5*6*78*,1*23*4*,78*C*E*,24*A*D*.}

{\bf 30-15} {\tt RSTU,U132,23LK,KLMQ,QPON,NOED,D5A6,64FG,GJIH,HIC9,9B87,78SR,,,F*J*M*P*,A*B*C*E*,1*4*5*T*.}

{\bf 30-16} {\tt TRSU,UFGI,IHML,LMJK,K3O4,418E,ECDT,,,NO*PQ,BG*PQ,9ABH*,5678*,1*2J*N,24*F*S*,23*7D*,C*K*NR*,8*BD*I*.}

{\bf 34-20} {\tt SRUT,TPYB,BYV1,1V7F,FEGK,KJML,LMHI,I8D9,94W2,2WXA,AXOS,,,V*W*X*Y*,NO*P*Q,CD*QU*,567*J*,7*9*CG*,\break 4*6D*F*,34*CH*,A*E*J*W*,8*B*H*V*.}

{\bf 38-19} {\tt VWYX,XYQP,PICB,BC87,78Z2,2c65,56KG,GMUT,TUSR,RSJF,FH43,34b1,1aA9,9AED,DELN,NOWV,,,Z*a*b*c*,K*L*M*O*,\break H*I*J*Q*.}

{\bf 38-24} {\tt aZcb,bcXY,YKNL,LACB,BCQ1,1DWI,IHJG,G5F8,87S9,96Oa,,,TUVW*,Q*RS*W*,O*PZ*a*,MN*X*Y*,F*J*Pa*,D*EH*I*,\break 45*6*R,3C*Vc*,2B*Ub*,1*6*9*T,5*8*A*K*,EI*MY*,348*b*,25*7*c*.}

{\bf 42-21} {\tt degf,f165,56CD,DERQ,QRba,abcZ,ZUKJ,JKGF,FG24,43IH,HINM,MNLP,POYX,XYWV,VWB9,9A87,78ed,,,STU*c*,\break C*E*L*O*,A*B*ST,1*2*3*g*.}

{\bf 46-22} {\tt ihkj,jkgf,fgTa,aYZb,bWXV,VPUN,N8OB,B9AC,C7MS,SQdc,cdei,,,U*X*e*h*,RS*W*Z*,O*P*T*Y*,KLM*Q*,IJM*R,\break FGHO*,DEHN*,567*8*,348*C*,127*B*,HP*Q*R.}

{\bf 46-23} {\tt hijk,k132,23HG,GHIJ,JKWV,VWYX,XYTU,ULNM,MNFE,EF46,65DC,CDPO,OPQa,aZbc,cfgS,SBRA,A987,78ih,,,def*g*,\break R*b*de,I*K*Q*Z*,9*B*L*T*,1*4*5*j*.}

{\bf 47-23} {\tt ilkj,jkaZ,ZaOP,PNcb,bced,deTW,WVgU,UfML,LMKF,F7GA,A89I,IHQR,RSYX,XYhi,,,f*g*h*l*,Q*S*T*V*,H*I*JK*,\break G*K*N*O*,DEG*J,BCF*J,567*H*,34A*H*,127*I*.} 

{\bf 52-26} {\tt onqp,pqed,deRS,SQgc,cZab,bLfM,M8FB,B9AC,C7JP,PNOT,ThiY,YWXm,mjkl,lUVo,,,f*g*h*i*,KS*X*k*,IJ*W*Y*,\break HM*V*j*,GKQ*U*,DEF*L*,567*8*,348*C*,127*B*,F*HKR*,HIL*O*,J*O*U*n*.}

{\bf 57-29} {\tt stvu,uvrq,qrlm,mYZc,cBHL,LIJK,KGdT,T7SA,A89i,ifhg,gOWR,RPQV,VUbj,jkps,,,nop*t*,d*eh*i*,ab*c*k*,\break W*Xf*g*,XZ*g*l*,T*eno,S*Y*ae,MNO*U*,EFG*H*,CDS*d*,567*h*,347*i*,12A*h*,H*K*R*U*,G*L*O*V*.}

{\bf 58-30} {\tt wtuv,vMsK,KIJb,bcde,eVXW,WCNB,Bpq7,76A8,8Efi,ijko,olnm,mYZa,aGHF,F9rw,,,p*q*r*s*,f*ghn*,ghj*k*,\break STUd*,PQRc*,LM*N*O,DE*OX*,9*A*E*H*,Z*l*s*w*,7*C*G*Y*,RUr*v*,58*DF*,34q*u*,12p*t*,QTq*t*,6*C*t*u*.} 

{\bf 60-33} {\tt rsut,tubc,ceiH,HfS3,3ST4,4TKO,ONPQ,Q9MU,U2jl,lxym,mCkD,DAWV,VWYX,XYaI,I157,7gh8,86nr,,,vwx*y*,n*opq,\break j*k*vw,f*g*h*i*,de*pq,Za*os*,RS*T*U*,JK*LM*,EFGH*,A*BC*D*,BD*b*d,5*6*FG,6*7*LP*,4*I*S*Z,2*9*RU*,2*5*8*R.}

{\bf 60-34} {\tt rsut,tuAB,B8HL,LNOM,M6Iy,y7aZ,Zabc,cEUX,XWih,higf,fgVe,e3dG,Gvon,noqp,pqST,TQRP,PFsr,,,v*wxy*,jklm,\break d*e*xy*,Yc*lm,U*b*jk,H*I*JK,G*KO*w,F*JN*R*,Q*S*Yc*,CDE*d*,8*9A*B*,6*7*DV*,45F*S*,3*H*L*v*,9B*U*c*,126*x,\break 57*P*x.}}

\subsection{\label{app:A-5}4-dim 300-675}

\small{The master set is given in the repository. Particular smaller
sets that also belong to the 60-74 class are in \ref{app:A-2}; See
Sec.~VII. The first three sets below are the sets obtained by Waegel a
nd Aravind [52], just translated into MMP hypergraphs. They do not
belong to the 60-74 class.}

\smallskip 
\small{{\bf 38-19} {\tt 4123,3C6L,LZSE,EDGF,FOVU,URbX,XNWJ,JIKH,HPQM,M794,,,56*7*8,9*ABC*,B5N*1*,8O*A2*,R*S*TI*,P*YZ*a,\break cb*Q*G*,V*K*cY,TW*D*a.}

{\bf 42-21} {\tt 2143,36CP,PIOe,eGcW,WVXU,ULdg,gFHa,aRZY,YbfT,TJNS,S8A2,,,56*78*,9A*BC*,DEF*G*,H*I*J*K,L*MN*O*,4*79Q,\break 1*5BR*,Qb*c*d*,EV*MZ*,DKf*X*.}

{\bf 48-25} {\tt 1243,36CP,Pcde,eVUX,XQYZ,ZmiE,EFGD,DabO,OMNL,LIhk,kjgR,R5B1,,,5*6*78,9AB*C*,HI*JK,4*79Q*,2*8AS,SJTU*,\break R*V*N*W,G*Hfg*,Sa*h*i*,F*j*d*l,c*TY*W,SfM*l,R*Km*b*.}

\smallskip
\begin{sloppypar}
{\bf 211-127} {\tt 1234, 5678, 9ABC, DEFG, HIJK, LMNO, PQR4, STUV, WXYZ, abc8, defg, hijk, lmno, pqoZ, rstY, uvwV, xyzc, !"\#g, \$\%\&', ()*-, /:;3, \textless =zR, \textgreater ?@7, [\textbackslash ]O, \textasciicircum \_`N, \{\textbar \}N, \textasciitilde ;yN, +1\}xQ, +2@fQ, +3+4+5:, +6+7+8+9, +A+BXU, +C+D+E-, +F+G+EG, +H`QF, +I+J+K+D, +L+M+Ne, +O+P+Qe, +R\textasciitilde TQ, +S+TtN, +U+V+WF, +X+Y+5k, +Z+a+b', +c+dqS, +e+f+g\textbar , +h+i+j=, +k+g+WP, +l+m+n6, +o+Q+Gn, +p+q+r+s, +t+u+v\textasciitilde , +w+o\&5, +x+s+Ej, +y+D+98, +z+!]K, +"+\#+PK, +1\textbackslash ?j, +\$+\#+2j, +\%+\&+'+k, +(+'+T*, +)+*+jP, +-+rcG, +/+:M2, +;+s\#b, +\textless +=+f\%, +\textgreater +?+@+\#, +[+\textbackslash +q\_, +]+\textasciicircum +'E, +\_+`+eW, +\{+\textbar +\}D, +\textasciitilde ++1+v+Q, ++2++3\textgreater E, ++4++3+nC, ++5+\&+!L, ++6+w+4J, ++7+\$+Ba, ++7+@\%1, ++8+u+8D, ++9++5+CB, ++A+*B5, ++B++C+oR, +"+8i3, ++D+b\textasciicircum m, +v+m+2w, ++E++F+p", ++D++9+i+B, ++G++9+(+z, ++H++I++6+q, ++A+l+T+K, ++6++2\textasciicircum s, ++2+\textbackslash /I, ++J[tA, ++4+'+J7, ++K++J+DL, ++I]eb, +\&+S+NC, +\}+?+:+S, ++L++FvE, ++M++C+Nu, +\%+EK9, +'ujH, ++N+S+D], ++O+=K4, ++P)L9, ++Q+hdO, ++4+y+e\$, +d+7[\textless , ++R++E64, ++F++B+?8, ++P++N++8H, ++E+b+MO, ++Q++H+d", +\textasciicircum +\textbackslash +e+K, ++R+:+h\{, ++D+\}h1, ++A++4+cI, ++N++CMA, ++K++H+\%(, +t+deJ, ++G(tH, +=+S+5a, +C+1t1, ++8+(+E[, ++S+h!J, ++T++G\textbackslash 9, ++U+m\{P, ++V++K+"a.}

\smallskip
{\bf 226-142} {\tt 1234, 5678, 9ABC, DEFG, HIJK, LMNO, PQRS, TUVW, XYZa, bcde, fghi, jklm, nopq, rstu, vwxy, 2zF!, "\#\$\%, 45\&9, '(B), *-A/, 4wE:, 1;\$\textless , "xF/, 3=\textgreater ?, v@[G, \textbackslash z(\textgreater , v\$A!, ];9y, \textasciicircum 6\_/, 5`F\{, @\textbar \%\}, w=B\textless , 2-\%y, \textasciicircum `C:, ]7=!, z[9\textasciitilde , +1+2+3+4, +5N+6+7, +8+9+A+B, IO+C+B, +D+5+E+4, +F+E+G+H, +I+JV+K, +L+MT+N, +LS+O+K, +P+Q+R+S, +J+TX+U, +P+VW+U, +I+W+R+N, +X+M+Y+Z, +a+b+c+d, +e+fcf, +e+gk+h, +i+j+k+l, +b+mi+n, +o+pei, +qbj+r, +s+tc+l, +ud+d+v, +q+g+w+x, +i+ydg, +s+z+ml, +!+f+k+n, +"e+\#+\$, +\%+\&p+', +(+)+*u, +-+/os, +\%+(+:+;, +\textless +=+\textgreater +?, +@+[+\textbackslash +], +\textasciicircum +\_+`+\{, +\textbar +\textgreater q+\{, +\textasciicircum +@pu, +\textless +(+\}+], +\textasciitilde ++1++2r, +/++3+?+', 1++4+!+\%, \textasciicircum +I++5+\textasciicircum , "J++6+-, ++7++8++9++A, ++BK++A+\textless , 4+D++5++C, ++7++D++E++F, 1++G++9+-, ++8+L+e+\textasciicircum , \textbackslash I+s+\textasciicircum , ++B++G+L++E, v++H+a+-, ++I++4+o+\textless , H++9++5++J, 4+L++6+\textbar , ++7J++4++C, w++K++L+t, ++M++N+p+=, *+5+Q+g, 6++M+t+\_, ++KP+z+(, 8+Q+z++O, \#+X+p++P, w++QS++O, \&++R++S+\}, ++T+M+w+\textgreater , =+9e++U, (++V++W+c, -++X++Y+), \$+Te+[, ++T++Vc+`, xL+w+[, (+E+m++1, ++Z+O++a+`, M++We+\textgreater , 9++bk++c, Fa+x+?, ++d++em+\textbackslash , D+Gi++2, C++fXf, G+AVl, ++gYjn, ++hT+\#++i, E++jg++c, \%+3a++i, D++f+So, +CXkp, ++eW+k+*, \_++fj+?, ++d+AY++c, \%Uin, A+Sl+*, +6V+x++2, \_Vgp, C+Ym++i, DY+k+:, A++jT+?, ++d+C+S+x, 9++hZi, F++gf+*, ++fU+d++c, y+H+nr, \}+7+l+\{, +B+Z+n+', ?+U+\$+], /+K++ku.}

\smallskip
{\bf 226-143} {\tt 1234, 5678, 9ABC, DEFG, HIJK, LMNK, OPQR, STUV, WXYV, ZabY, cdef, ghiR, jkli, mnoN, pqfJ, rste, uvwx, yzKG, !tFC, "\#K8, \$\%q8, \&'\#o, ()*M, -/:4, 'lbF, ;\textless =\textgreater , ?@[:, \textbackslash ]\textasciicircum \_, `\{\textbar \}, \textasciitilde +1+2B, +3+4+5x, +6+7+8+9, +A+B+C\%, +D'zq, +EnbJ, +F+G\textgreater *, +H+I+J[, +K+2[), +L+M+N+G, +O+P+QQ, +R+S+T+U, +V+W+G+9, +X+Y+Z), +a+b+c+K, +d+e:U, +f[h7, +g+hhX, +i+j+kp, +l+m+nh, +o+p+q!, +r+s+t+n, +u+v+w+x, +y+z+!/, +"+\#+\$+\%, +\&+'+(+C, +)+*-A, +-+NdX, +/+5\_W, +:+;+Kd, +\textless +;+ma, +=+E\#E, +\textgreater +\%=9, +Q+J+4\}, +*+!+e3, +?+@+[+), +\textbackslash +]+=9, +\textasciicircum +(\textgreater w, +\_+`)R, +\{qTN, +\textbar +DmI, +/+-+x+d, +\}\textbar \textasciicircum L, +w+n:s, +\textasciitilde ++1+Us, ++2+[+dY, \textasciicircum *\$I, ++3++4++5J, ++6++7++8", ++9+@+I3, ++A++B++CP, ++D++E+\textgreater D, ++F++8+\$\textgreater , +]+dtm, ++Gpl6, ++H++I+km, ++J++K++L\%, ++M+'+xV, +\{+f+cc, +t+haC, ++N++1+\textgreater N, ++O+Z]i, ++E+\textless +T', ++G+\}\textless (, +\textbar +\textasciicircum +/@, ++P+j+D2, ++2+\&+\#+Y, +`+:+\$k, ++7+[vo, ++Q++R+\textless +4, ++S+?+!+H, ++T+\textbackslash ug, ++L++5\textgreater V, ++N+\textbackslash \textgreater j, ++C+P*z, +B*ZO, ++U+b+Yk, ++R+qle, ++V+W\textasciicircum j, ++C+\textasciitilde +3\&, ++G+S+8:, ++U+"[1, ++W++X+\$h, +d+U+Ic, +b\&pL, ++X+\{)l, ++Y++ZWG, ++a++PtP, ++b++IyC, ++c++Z++L5, ++d++S+FA, +\textasciitilde mf5, ++e++f++X+\&, ++W+\_j1, ++gSOG, ++T++Q++4\&, ++V+??2, ++T++M+\textasciicircum +8, ++g++S++I1, +\}+\{+E+7, ++7+k+6O, ++B+\&+QZ, +bTHF, ++V+bge, ++O+bz7, +\textasciicircum +MpA, ++O+\textasciitilde "b, ++U+q\textasciicircum R, ++G+\textasciitilde \$D, ++h+Ai2, ++i++QnU, ++j)n8, ++k*T5.}

\smallskip
{\bf 240-156} {\tt 5678, 9ABC, HIJK, LMNO, PQRS, TUVW, bcde, nopq, vwxy, z!"\#, ()*-, /:;\textless , =\textgreater ?@, \_`\{\textbar , \}1\textasciitilde +1, +4+59+6, \_3+7+6, +8+9C+A, \}+B+C+D, +E+F+G\textbar , +E2+CE, +H+I+5+J, +K+LG+D, +M`9+N, 17C\textbar , +M+O\textasciitilde +J, +PM+Q+R, I+S+T+U, +VL+W+X, +c+d+e+f, +gN+h+i, +j+S+e+k, +V+lR+m, +n+T+a+o, K+p+q+r, +Y+s+t+u, ILS+v, +wP+W+f, +P+w+v+m, J+s+z+o, +l+Q+a+i, L+d+!+r, +YR+!+U, HOP+a, I+l+\#+f, K+w+t+i, N+T+\%+m, +\&Y+'+(, WXk+*, Wiw+-, +/Zu", +:o+*+(, +;hnt, +\textless jv", U+=u+\textgreater , +\textless ior, +?+@+[+\textbackslash , +]-\textless +\textasciicircum , ++5/++6++7, *++8\textgreater ++9, \$+\textasciitilde \textgreater ++A, +?++1++6++D, ++E++F/?, +`++G++2+\textasciicircum , ++H++I\textless @, +]++J\textgreater ++K, \&+`++L++7, '++5=+\textasciicircum , \$:++M+\textbar , ++E++8++N+\textasciicircum , ++O++5++P++Q, \&+@?++R, ++a[++b++c, ++d++e++f++c, ++S++g++b++h, ++i++j[++k, ++d++n++o++p, ++W++q++r++c, ++s++a++t++u, ++x++U++y++c, ++S++q]++z, ++!++"++m++\#, \}+;+?++i, +H+\&++H++d, +8IV++i, +8+n+\&+], ++'+/\&++(, +VW++H++s, +H+YV\$, +K+g\&++i, I+\textless '++*, 3MY++q, +L+Zb++a, +INc++F, ++/d+@++j, +SX++I++", ++:++/Y-, `+S*++q, +9X++;++a, 1OZ++\textless , Mf)++=, +B+w(++\textgreater , 1++/++F++a, 3++?i+@, 2d*++e, +se++;++q, 2NX++\textgreater , 3+l++I++T, +CQk\textasciicircum , +F+Q+'++r, 8R++[:, ++\textbackslash +tl++G, So+[++n, ++]j;++\textasciicircum , +O+p+=++g, ++\textbackslash Q++8++\textasciicircum , 8k++L++n, +Fn;[, Rp++J++g, +5+'++J++\textasciicircum , 6+Tn++\textbar , \{++]p++L, +C+t;++g, +d++[++8[, Pq++G++r, ++\}+hs++Z, +G++\textasciitilde ++3++m, +7+\#++N++k, G+zt++Q, +Ww++2++t, +7+++2w?, Br=++y, ++\}+++3++2++f, +G+Wt+++4, +hv++3++k, 9+e++2++o, +G+!++6++y, ++\}+ay\textgreater , B+zw++m, \textasciitilde +qu++b, Cs@++t, +\#++\textasciitilde =++Z, +ux++M++z, 9t\textgreater ++k, A+!v++f, +G+vx@, +\%+++3?+++4, +1+U"+++7, +A\#++D++u, +J+r+\textgreater ++h, +6+r+-++R, +N+U+++D+\textbar , +++8z+\textasciicircum ++h, +o+\textgreater ++D+++H, +N!++A++p, +k\#+\textbackslash ++\#.}

\smallskip
{\bf 253-165} {\tt 1234, DEFG, HIJK, LMNO, PQRS, TUVW, XYZa, bcde, fghi, jklm, rstu, vwxy, \$\%\&', ()*-, /:;\textless , =\textgreater ?@, [\textbackslash ]\textasciicircum , \_`\{\textbar , +3+4+5+6, +7+8+9+A, 1+B+C+D, +F8+G+A, +H9F+D, 5+I+J+K, +8+L+M+N, +O8+PF, 1+QD+K, +F+T+L+U, +W+X+Y+D, +Z+B+G+a, +bQX+c, +dPe+e, +f+g+hh, HZdg, +j+kc+l, +mRc+n, +b+r+i+e, +fMTd, +u+ve+w, +m+gZ+i, +xNWe, KQ+k+q, H+jU+e, LX+h+n, JXef, +g+sb+c, +yo+z+!, +"+\#+\$+\%, +\&+'+(+), +*+-+/+:, jp+;+\textless , +=+\#+\textgreater +?, +y+@+[+\textbackslash , +*+]+\textasciicircum +\_, m+`+z+/, o++2++3+), j+'+\textasciitilde +[, ++4+`++2+\textbackslash , +\#+@+\}+:, +\&++6++7+\%, ++8+z++9+?, lo+\}++A, m++B+@+), +=n+\textasciitilde ++C, +*++D+\textbar +\textless , +'++E+\textasciicircum +!, v++J;++K, "++L@++M, *++N:++O, x(++P++M, t++Q++R++S, )++T\textless ++U, ++V(++J++R, w++X:++Y, t)@++K, s++N++Z++a, ++d\%;++M, y)++e++I, ++k++l++m++n, ++o++p`++q, ++r\textasciicircum ++s++q, ++k++t\textbar +5, ++u]++v\{, ++u++o++"+1, ++\%++w\_+3, 3+fk++u, +7H++\&++r, 1+=w++k, +O+xu++(, +F+b+\&++), 4I++4r, ++*Kx++u, +7+bs++k, ++-jr++r, +W++8++'++/, +F+d+*++\%, +OI+y++f, 1++:++)++r, 4++\&y++\%, ++*++:++4++(, ++-+f+y++V, +7+dmv, 3+m++d++f, +x+*x++r, Ik++'++k, +8L++B*, Op\%++w, S++D(++\textgreater , +p++?"], 6+g++D\&, +T++?\$++o, +j+-!\textasciicircum , +H+u(++t, +Q+po++z, +X+'"++@, +H++B\%++z, 8n-], 5+p+`', No\&++@, ++[+\$++N++\textbackslash , +IZ++X++p, +BY+z++T, a+\textbar ++P++l, ++[a++J++], +P+s++\textasciicircum ++p, +r+z++L++\_, 9V+\textbar ++L, +B+v+\$++`, +k+\textasciitilde ++X++\textbackslash , +Y++\textasciicircum ++Q++\_, 9+]++F++], ++[W+z++!, CZ++E++\textbar , +I+r++6++e, +R++7;++\}, +G++9\textgreater ++y, G++\textasciitilde +/?, e+\textgreater =++", +Cb++9++W, Ed:++y, +M+;=++s, D+q+\}++\}, ++\textasciitilde +\textasciicircum \textless \{, D+++1++3\textgreater , +J+[\textless ++", +i++9/\textbar , +R+++2++1++s, E+i++7=, +++1+;;\_, b++5?+++3, +q+[++Z++y, +N+l+?+++4, +A+c++U\}, +w+\textbackslash ++a+5, +e+?+++6+1, +c++A++I+6, +A+e+:+3, +n+\_++c+2, h+!++U++n, +l+\%++S+4, +D+++7++C+4, +++8+w+!+++6, +U+e++a++n, +++B+++C++Y\}, +++D+\textless ++M+3.}

\smallskip
{\bf 257-169} {\tt 1234, 5678, DEFG, HIJK, LMNO, PQRS, TUVW, XYZa, bcde, fghi, jklm, nopq, rstu, vwxy, z!"\#, \$\%\&', ()*-, /:;\textless , =\textgreater ?@, +3+4+5+6, 1+B+C+D, +F8+G+A, +H9F+D, 5+I+J+K, +8+L+M+N, +O8+PF, 1+QD+K, +7+H+R+S, +F+T+L+U, +W+X+Y+D, +bQX+c, +dPe+e, +f+g+hh, HZdg, ST+if, +j+kc+l, +mRc+n, +b+r+i+e, +fMTd, IP+s+t, +u+ve+w, Iach, JXef, +g+sb+c, +yo+z+!, +"+\#+\$+\%, +\&+'+(+), +*+-+/+:, +=+\#+\textgreater +?, +y+@+[+\textbackslash , +*+]+\textasciicircum +\_, m+`+z+/, +\&+\{+\textbar +\}, +-+\textasciitilde ++1+\%, o++2++3+), j+'+\textasciitilde +[, ++4+`++2+\textbackslash , +\#+@+\}+:, j++2++5+:, ++8+z++9+?, lo+\}++A, m++B+@+), +=n+\textasciitilde ++C, +*++D+\textbar +\textless , +'++E+\textasciicircum +!, uz++F/, v++J;++K, "++L@++M, *++N:++O, x(++P++M, t++Q++R++S, )++T\textless ++U, ++V(++J++R, u++T++W++M, w++X:++Y, t)@++K, s++N++Z++a, ++d\%;++M, u\$++J++O, ++k++l++m++n, ++o++p`++q, ++r\textasciicircum ++s++q, ++k++t\textbar +5, ++f++w++p++m, ++x++l++y+1, ++f++z++!+6, ++u++o++"+1, ++\%++w\_+3, 3+fk++u, +7H++\&++r, 1+=w++k, 2+d++'++\#, +O+xu++(, 4I++4r, ++*Kx++u, +7+bs++k, ++-jr++r, +W++8++'++/, +F+d+*++\%, +OI+y++f, ++*+o++\&++', +O++;jw, +Fku++\textless , +7++4++d++=, 4++\&y++\%, ++*++:++4++(, ++-+f+y++V, +7+dmv, +FHw++/, 3+m++d++f, +x+*x++r, Ik++'++k, +ZP'\textbackslash , +8L++B*, Op\%++w, S++D(++\textgreater , +p++?"], 6+g++D\&, +T++?\$++o, +j+-!\textasciicircum , +H+u(++t, +Q+po++z, +H++B\%++z, 8n-], 5+p+`', ++[+\$++N++\textbackslash , +IZ++X++p, a+\textbar ++P++l, ++[a++J++], +P+s++\textasciicircum ++p, +r+z++L++\_, 9V+\textbar ++L, +B+v+\$++`, +k+\textasciitilde ++X++\textbackslash , 9+]++F++], ++[W+z++!, +L+k+\textbar ++\{, +I+r++6++e, +R++7;++\}, G++\textasciitilde +/?, e+\textgreater =++", +Cb++9++W, Ed:++y, +M+;=++s, +J+i++5++g, D+q+\}++\}, ++\textasciitilde +\textasciicircum \textless \{, Gc;++g, D+++1++3\textgreater , +J+[\textless ++", +i++9/\textbar , E+i++7=, +++1+;;\_, b++5?+++3, +q+[++Z++y, +N+l+?+++4, +A+c++U\}, +e+?+++6+1, +c++A++I+6, +Sg++Y+5, +A+e+:+3, +N++C++a\}, h+!++U++n, +l+\%++S+4, +D+++7++C+4, +++8+w+!+++6, +U+e++a++n, +++B+++C++Y\}, +++D+\textless ++M+3.}

\smallskip
{\bf 257-169} {\tt 1234, 5678, 9ABC, HIJK, LMNO, PQRS, TUVW, XYZa, bcde, fghi, jklm, rstu, vwxy, z!"\#, \$\%\&', ()*-, /:;\textless , =\textgreater ?@, [\textbackslash ]\textasciicircum , \_`\{\textbar , \_3+7+6, +8+9C+A, \}+B+C+D, +E+F+G\textbar , +H+I+5+J, +K+LG+D, +M`9+N, 17C\textbar , +M+O\textasciitilde +J, +PM+Q+R, I+S+T+U, +VL+W+X, +Y+Z+a+b, +c+d+e+f, +gN+h+i, +j+S+e+k, +V+lR+m, +n+T+a+o, K+p+q+r, +Y+s+t+u, +wP+W+f, +x+y+T+W, +P+w+v+m, J+s+z+o, +l+Q+a+i, L+d+!+r, +SQ+u+", +YR+!+U, HOP+a, +jL+\$+o, K+w+t+i, N+T+\%+m, +\&Y+'+(, Ufx+), WXk+*, Wiw+-, +/Zu", +:o+*+(, +;hnt, +\textless jv", U+=u+\textgreater , +\textless ior, +/e+'x, +?+@+[+\textbackslash , +]-\textless +\textasciicircum , +\}+\textasciitilde ++1++2, ';++3++4, ++5/++6++7, *++8\textgreater ++9, \$+\textasciitilde \textgreater ++A, \%++B++8++C, +?++1++6++D, +`++G++2+\textasciicircum , +\_+@++8++3, \&+`++L++7, '++5=+\textasciicircum , \$:++M+\textbar , \&+@?++R, +\textasciitilde \textless +\{++C, ++W++X++Y++Z, ++a[++b++c, ++d++e++f++c, ++i++j[++k, ++d++n++o++p, ++W++q++r++c, ++s++a++t++u, ++x++U++y++c, ++!++"++m++\#, \}+;+?++i, +H+\&++H++d, +8IV++i, +8+n+\&+], ++\%V+\}++\&, ++'+/\&++(, +VW++H++s, ++)T+\_++x, +H+YV\$, +K+g\&++i, I+\textless '++*, +B+yh++e, 3MY++q, +L+Zb++a, 1e)++-, +INc++F, ++/d+@++j, ++:++/Y-, `+S*++q, +9X++;++a, 1OZ++\textless , 1++/++F++a, +IY++5++@, 2d*++e, +se++;++q, +CQk\textasciicircum , +F+Q+'++r, 8R++[:, ++\textbackslash +tl++G, So+[++n, ++]j;++\textasciicircum , ++\_+Tk++J, +5+d/++n, +O+p+=++g, ++\textbackslash Q++8++\textasciicircum , 6S+'\textless , 7+to++U, 8k++L++n, +Fn;[, Rp++J++g, +5+'++J++\textasciicircum , \{++]p++L, +d++[++8[, Pq++G++r, ++\}+hs++Z, +G++\textasciitilde ++3++m, G+zt++Q, +Ww++2++t, +vr+\{+++1, A+v++Q++Z, Br=++y, ++\}+++3++2++f, +G+Wt+++4, +hv++3++k, 9+e++2++o, ++\}+ay\textgreater , B+zw++m, \textasciitilde +qu++b, Cs@++t, +ux++M++z, 9t\textgreater ++k, A+!v++f, +G+vx@, +\%+++3?+++4, +1+U"+++7, +++8+o+-++\#, +A\#++D++u, +++9+b+++A++A, +++B+)++C+++C, +J+r+\textgreater ++h, +6+r+-++R, +N+U+++D+\textbar , +++E+X!+++F, +A+f+++A+++G, +++8z+\textasciicircum ++h, +1+\textgreater ++C+++G, +N!++A++p, +J+f"+++I, +k\#+\textbackslash ++\#.}

\smallskip
{\bf 283-188} {\tt 1234, 5678, 9ABC, DEFG, HIJK, LMNO, PQRS, TUVW, XYZG, abKF, cdef, gfSJ, hijW, klmn, opqn, rsje, tuvd, wxyZ, z!OI, "\#!V, \$\%\&', ('\#N, )*-S, /:;C, \textless =\textgreater 8, ?@[\textbackslash , ]\textasciicircum \_`, \{\textbar \textbackslash -, \}\textasciitilde \&7, +1+2zB, +3+4+5y, +6+7+8v, +9+A+Bm, +8+2\textasciitilde l, +C+D\%G, +EqjM, +BpJE, +F+G+HY, +I+J+K+L, +M+N+L+A, +O+P+N+5, +Q+R+Ss, +T+UR6, +V+W+XY, +Y+Z+a+9, +b+c+d+a, +e+f+gx, +h+i+jb, +k+l-d, +mmVK, qfZV, +n+o+p+q, +r+s\}!, +t+uwH, +v+w+xp, +y+z+Xk, +!+"+\#c, +\$+\%+\#+g, +\&+'+ur, +(+)+PU, +*+-+/I, +:+;+\textless +=, +\textgreater +j+Mu, +?+@+Kw, +[+Hxk, +\textbackslash +]+\textasciicircum i, +\_+`+G`, +\{+\textbar +\}w, +\textasciitilde ++1++2k, ++3+z+O', ++4++5+S;, ++6+Zjg, ++7++8+yG, ++9++A+44, ++B++C+qA, ++D++E+W+J, ++F++G+)+1, ++H++I+\}o, ++I+R\textbar h, ++J+/+ad, ++K+@+'7, +-+'\_B, ++L++M++A+', ++N++O++P+', \textbackslash aZO, ++Q+t+G:, ++R+p+FJ, ++S++6+\textgreater Q, ++T++U+p+V, ++V++W++X+\textgreater , ++Y+\textgreater +\%+7, ++Z+m+Dp, ++P+o+Hh, ++a++C+UN, ++b++5nY, ++a+\textasciicircum +aq, ++c+gif, +B!bY, ++d++e+"v, ++f++g++J3, ++h++i+CP, ++j++g++MT, ++k++l++L\textgreater , ++m++G+Y+E, ++n++o++p++S, ++q++r+[+=, ++s+x+lM, ++b++X+\#+8, ++t++Z+K\textasciicircum , ++u+H\_s, ++v++a+l=, ++w+p+g\textless , ++x++y+D\textgreater , ++z++!++"[, ++\#++\$++S", ++\%+]P3, ++\&++'+3U, +\textless +s\textasciicircum R, ++(++)+9S, +?+\&+fn, ++R++Q+Xy, ++*++-+dC, ++/++l;4, ++:++c+4\textasciicircum , ++O+m+i8, ++W++3+!+3, ++\%++V+\$+6, ++;++\textless +B2, ++\textless ++F+k+j, +e\}@/, ++9++4aW, ++H+`+wL, ++p++L++E+;, ++=++G++C+5, ++\textgreater +\%eb, ++M++2]Y, ++:++H+a', +?+v+aO, ++Vtlb, ++1+\textbar +g+1, ++?++@+Kr, +u+Q?i, ++[]rQ, ++'++D:k, ++\textbackslash ++-+G5, ++X+c+H3, ++r++a+n9, ++8+\#+d@, ++)zpX, ++U\{\&u, ++E+q+G", ++]++;++/a, ++e+\textgreater \textasciitilde f, +2uiD, ++\textasciicircum +r+C\textless , ++\_+b[9, ++`++\textgreater ++"+6, ++\{++\textbar ++U+k, ++\}++\textasciitilde ++*1, +++1++\textasciitilde +';, +++2++a+w\%, +++3+++4+K\textbackslash , ++o++D+2L, ++\textbar ++N\textbar v, ++=++p++N\#, +++4++[*E, ++\textbackslash ++T+T=, ++@++i\{`, ++@++r++B?, +g+E:T, ++K]fF, ++\textbackslash ++"++Yi, +!ib9, ++(++bE1, ++1!o1, ++\textasciitilde +=+6w, ++;+p@5, +++4++T+\#7, +++5++@++\&m, +++6++7+ky, +++7++\textbackslash +h\textasciitilde , +++8+7\#X, ++b+Azf, +++9++P*e, +++A+\textbar +3t, +++B++b\%K, +++CuC8, +++D\textless -r.}
\end{sloppypar}  
}

\subsection{\label{app:A-6}4-dim 148-265}

\small{The master set is given in the repository.}

\smallskip 
\small{{\bf 49-28} {\tt 4132,2gNf,fhYL,LFKJ,JBml,lDPe,ecdV,VTU8,8567,7QRS,SZab,bCk4,,,5*9AB*,5*C*D*E,F*GHI,MN*OP*,W6*XY*,\break Z*9GN*,Z*EL*V*,c*C*J*N*,h*D*Gi,3*a*J*i,3*EHO,jd*K*O,4*AL*P*,l*f*Hk*,na*IP*,nEg*m*.}

\small{{\bf 49-29} {\tt 3142,2McC,CDE5,5FHG,GZaX,XWYP,PQSR,Rn6K,KILJ,JlT8,8fNh,hBd3,,,5*6*78*,5*9AB*,I*M*N*O,P*T*UV,Z*D*N*U,\break b9J*c*,b7d*U,2*8*eS*,f*AgQ*,f*F*iY*,3*6*OV,3*D*J*Q*,j9eV,jE*K*h*,4*9N*R*,4*H*kW*,l*AOm.}

{\bf 52-29} {\tt 4231,1675,5hUK,KIJL,LaeZ,ZXYW,WNRA,A9B8,8CED,Dmnk,kfVl,lGo4,,,8*FG*H,I*MN*O,PQR*S,PTU*V*,a*C*bc,\break a*G*dU*,2*f*J*c,2*HgZ*,3*FL*V*,6*e*ij,k*HK*j,7*D*L*o*,7*HbT,pFiq,ph*go*,4*C*K*q,4*e*m*T.}

{\bf 53-32} {\tt qorp,pKML,LGPX,X34S,SQRT,TOkD,DCWZ,ZYga,aIJd,dchb,bAHB,BNi7,75e6,612q,,,lmnr*,i*jk*n,e*fg*h*,X*d*k*q*,\break UVW*c*,N*O*P*r*,H*J*W*f,EFG*e*,9H*Y*j,8M*e*k*,4*Fjp*,3*Ea*n,6*L*R*n,2*K*S*i*,7*VX*m,B*EQ*q*,2*D*Fm,I*S*e*r*.}

  {\bf 43-32} {\tt rpqs,sLMK,K5Ej,j7Jc,c4Cm,mAGk,kVWU,U6In,n1FQ,QPdR,ROSX,XYei,ighf,fDbr,,,m*n*os*,j*k*lr*,b*c*d*e*,\break Zalo,S*Tj*n*,NO*c*h*,I*J*M*Y*,G*HL*T,F*ac*q*,E*F*Hg*,BC*Tp*,9A*S*g*,8S*W*q*,7*G*Q*Z,7*V*g*s*,C*I*g*l,23O*r*,\break 8E*Y*m*.}

{\bf 55-33} {\tt 4132,2CZT,Tk9M,Mjfp,pBKU,UqlN,Nt6h,hgbH,HGIF,FPQR,RsAV,Vr7O,OonY,YSXW,WEc4,,,56*7*8,59*A*B*,5C*DE*,\break F*JK*L,F*M*N*O*,S*T*U*V*,ab*c*T*,adI*e,3*DG*f*,g*dZ*W*,ij*Q*X*,iB*O*e,k*l*R*e,mn*P*T*,r*9*Q*h*,r*l*JW*,\break s*n*K*h*,t*j*LV*.}

{\bf 55-34} {\tt 3241,1576,69pn,nmKC,COak,kHfl,ltgs,sJeq,qcrB,BNWh,hGji,i8bo,oITQ,QRS3,,,2*8*9*A,2*B*C*D,2*EFG*,\break 2*H*I*J*,3*K*LM,3*N*O*P,4*T*UV,4*W*XY,4*Za*b*,5*Fc*a*,5*de*f*,5*g*S*V,m*EPb*,m*Ac*X,6*B*Mb*,6*G*O*X,q*9*PY,\break 7*ELr*,t*8*O*r*,t*dR*U.}

{\bf 56-33} {\tt ghqi,iMYN,NFlG,GCHp,pVmW,WbsS,SUnT,TDaQ,QOPR,RKLe,edfc,cAjB,B9Zt,tJoI,IXkg,,,rs*t*u,n*o*p*q*,\break j*k*l*m*,Z*a*b*f*,X*Y*p*u,U*i*m*t*,T*h*l*u,H*J*L*h*,Ed*k*n*,789*X*,56L*m*,6EP*u,4V*g*r,3B*O*V*,2M*T*V*,\break 2P*j*q*,5F*q*s*,1F*U*X*.}

{\bf 56-34} {\tt 4132,2Dgc,cWbH,HGIF,FJLK,Kdsr,rZpA,A59B,BteN,NMPO,OQST,T8li,ihnm,mCq4,,,5*678*,5*C*D*E,Q*6J*R,Q*EUV,\break W*9*XR,W*YZ*a,2*d*e*R,2*A*fO*,2*h*I*V,i*d*jP*,i*A*J*k,3*S*Xk,3*B*K*l*,3*EH*n*,o9*L*P*,4*7jp*,4*YJ*N*,\break r*8*XN*,t*S*L*p*,t*ug*q*.}

{\bf 60-35} {\tt 2341,19A8,8pqQ,QEsr,rDPV,Vkno,ovCS,Sxyd,d6Wa,aZbB,Befc,c5YG,GutR,RgLJ,JHI2,,,1*5*6*7,2*B*C*D*,\break 2*E*FG*,3*KL*M,3*NOP*,4*Q*R*S*,4*TUV*,4*W*XY*,hH*Mi,jI*KT,k*B*NQ*,k*Flm,r*d*t*m,9*E*o*i,9*d*NT,v*G*Nw,\break x*E*OU,x*D*q*w,A*B*t*U,A*FP*S*.}

{\bf 61-36} {\tt 3241,1576,6UOV,VxR9,9Itw,wFrZ,ZYab,bghf,fJei,inol,ljdk,kCup,pcq8,8HQK,KLM3,,,1*8*9*A,1*BC*D,1*EF*G,\break 2*H*I*J*,4*NO*P,4*Q*R*S,5*TK*N,7*WXP,Y*c*d*e*,j*J*a*m,Bo*r*Q*,BI*L*s,Ek*t*Q*,EJ*M*q*,C*J*vR*,F*H*Xu*,\break 9*o*M*u*,Gx*K*y,Ak*L*z,DH*V*z,Dc*t*y.}

{\bf 61-37} {\tt 3241,1576,6VeH,H8WZ,ZpEQ,QAcJ,JuUa,aqR9,9ITt,tCLo,onkl,lGmi,ifhg,gBrP,PNO3,,,1*8*9*A*,1*B*C*D,2*E*FG*,\break 2*H*I*J*,2*KL*M,3*Q*R*S,3*T*U*V*,3*W*XY,4*Z*a*b,4*c*de*,f*jT*k*,p*q*V*r*,5*sT*Z*,9*E*Ye*,A*FV*b,vwYb,vn*h*x,\break 7*I*Sb,DMN*c*,yKO*z,yg*o*c*,yjm*x.}

\smallskip
\begin{sloppypar}  
{\bf 115-78} {\tt 1234, 1567, 189A, 1BCD, EFGH, EIJK, ELMN, EOPQ, ERST, UVWX, UYZa, bcde, bfgh, bijk, blmn, opqr, ostu, vwtk, vxqy, vMz!, vK"m, \#\$Y!, \%L\&', \%J(!, 2)*j, 2O-l, /:;u, /R\textless l, =\textgreater ?r, =G@j, =[\textbackslash u, ]\textgreater \textasciicircum k, ]QWm, \_`;k, \_H\{\textbar , \_\}\textless !, 3:@\textbar , \textasciitilde :+1i, \textasciitilde P+2n, 4SXn, +3FY+4, +3+5+6+7, 5\textgreater (+4, 5GZ+7, 5[\&+8, 5Q+9h, +AF\&+B, +A+5(+C, 7+5a+D, 7:z+E, 7Pte, +F\textgreater +6+G, +F[Y+E, +H`Z+G, +HH(+E, +H\}+9e, 8N\textasciicircum h, B+I\textless +4, B+JX+8, BJ\{+K, +LM@+K, Cw+2+B, Cx-+C, +MpX+B, +MI;+N, +MN\{d, 9+OV+C, 9L\textbackslash +N, 9J?f, +P+IW+G, +P+O+2+D, +P+J-+E, +PL*g, AxX+D, AM;e, DsV+E, DI@g, D\$\textbackslash e.}

\smallskip
{\bf 130-80} {\tt 1234, 1567, 89AB, 8CDE, 8FGH, 8IJK, 8LMN, 8OPQ, RSTU, RVWX, YZab, Ycde, Yfgh, Yijk, 2lmc, 2noi, pqrg, 3sth, 3uvw, xyzd, x!Wj, "l\#\$, "Gvk, \%q\&e, \%'mf, \%()k, 4*r\$, 4y-f, 4F/:, 4;Xk, 4!\textless =, \textgreater ?@g, \textgreater ISi, \textgreater P[\textbackslash , \textgreater N]\textasciicircum , 5A\_g, 5L`i, \{B\textbar \}, \{O\textasciitilde i, \{M+1\textbackslash , 6O]w, +29+3h, +2I+4w, +2NTj, +5LUw, +5Q+1j, +6L[:, +6J]k, +7+8\_\$, +7OT:, +7K`=, 7I+1:, +9?+A+B, +9Ev+C, +DBX+C, +E+8v+F, +EDV+G, +EOtb, +HI\&b, +IA\textless +G, +IQr+J, +KC+A+L, +MD)+N, +O9/+L, +OE\textless +P, +OP+Q+R, +ON-Z, +S*S+T, +S;@a, +Us\textasciitilde +T, +Uu+V+W, +U(+Xa, +YqS+G, +ZF+3b, +Z!\_+J, +alU+F, +bl[+L, +b+c+dZ, +e*+1+L, +ey\textasciitilde +P.}

\smallskip
{\bf 131-90} {\tt 1234, 1567, 189A, 1BCD, 1EFG, 1HIJ, 2KLM, 2NOP, 2QRS, 3TUV, 3WXY, 3Zab, 4cde, 4fgh, 4ijk, 4lmn, 4opq, rKsf, tuvl, tOwx, yz!", yM\#g, yu\$\%, \&K'", \&()\%, \&P*m, -/wm, :;\#j, :L!h, :\textless s=, :\textgreater \$n, ?@'h, ?O)n, [\textbackslash ]n, [Pv\textasciicircum , 5\_bk, 5`Tl, 8R\{x, 8\textbar V\}, \textasciitilde +1+2\%, \textasciitilde RT+3, \textasciitilde \textbar Xm, 9`Y\%, 6+4+5", 6+6+7i, 6QV\%, 6+8W+3, 7+1U+9, 7RYn, 7\textbar +A\textasciicircum , +B+C+5h, +B\_+7=, BS!+D, E+4+Eo, E+F)+G, +H+Ivo, +H+J]+G, F+I\$+K, F+J+Lp, F+1!+M, FRs+N, +O+C+E+P, +O`'+M, +OS+Qd, C+4]+K, C+6v+P, +R+8!e, D+S)q, D+1+Q+T, DR'e, G+C]q, G\_v+U, GS+V+W, +X\textbackslash +Y+Z, +XPb+D, H\textless Y+G, H\textgreater +a+Z, +bu+7c, +bOa+Z, Iz+2+K, I@T+P, IMXp, IO+c+N, +dKY+K, +d\textbackslash Z+N, +e\textless \{p, +e/ad, +f\textless T+U, +f/+c+W, J(a+T.}

\smallskip
{\bf 146-97} {\tt 1234, 1567, 189A, 1BCD, 1EFG, 1HIJ, 1KLM, 2NOP, 2QRS, 2TUV, 3WXY, 3Zab, 3cde, 3fgh, 3ijk, 4lmn, 4opq, 4rst, 8ucv, 8Owx, 8yz!, BP"\#, 9\$\%\&, 9P'(, CN)*, C-c\&, C/w(, C:;\textless , 6S=\textless , D\textgreater \%?, DS@o, A/[?, A\textbackslash =], A:Zo, \textasciicircum \_k\#, \textasciicircum `\{p, \textasciicircum VY\textbar , E\}iv, E\textasciitilde +1\#, E+2+3p, E+4+5\textbar , +6+7jx, +6+8h\#, +6+9+A\textbar , F+7f\&, F+8+B(, FU+C+D, F+9\{\textless , +E\_g(, +E+F+5q, +G+2+H+D, +G+4X\textless , +I+Jk+K, +IT\{], +L+Mg+N, +L+O+3], G`+A+P, +Q+Ras, +Q+F\%l, +Q`'+S, +QVe+T, +U+J;s, +U\textasciitilde b+V, +W+8+X+V, +W+O[l, +Y+7@+Z, +YU\%+a, +b+F)+c, +b`c+a, +d\textasciitilde =+e, +d+2[+a, +f\}a+g, +f\textasciitilde @r, +h+7b+i, +h+8zr, +hU)n, +h+9c+j, +k+l=+g, +k+R+X+i, +m:k+T, HuW+n, Hyil, KQ+ol, I+p+Bm, L\textbackslash +1+c, LRi+a, +qO+H+Z, +qyh+c, +q+r+o+a, +qSjm, Mu\{+g, +s\$+5+i, +sPWr, +sQg+t, +s+u+1n, J-+A+i, JRhn.}

\smallskip
{\bf 147-96} {\tt 1234, 1567, 189A, 1BCD, 1EFG, 2HIJ, 2KLM, 2NOP, 2QRS, 3TUV, 3WXY, 3Zab, 3cde, 3fgh, 3ijk, 4lmn, 4opq, 4rst, uvwx, yzc!, yL"\#, y\$\%l, \&Pb', (M)*, (-/:, (;Z\textless , =\textgreater c?, =v\%\textless , =@am, [\$b:, [\textbackslash ]m, \textasciicircum z)\_, \textasciicircum Le`, \textasciicircum O/n, \{\textbar \}`, \{P\%\textasciitilde , +1N]+2, +1vbn, +1@+3\textasciitilde , +4+5Xx, +4+6+7', 5H+8!, 5+9g\#, +A+Bh+C, 6+D+E+F, 6+Bj*, +G+H+8?, +GQY:, +IHh+F, +I+J+Km, +LHj\_, +L+9+E`, +L+Mf+N, +O+B+8+N, +O+PU+2, +O+QYn, 7Ji+N, 7+6W\textasciitilde , 8+Hw+R, 8J/t, 8Q+S+T, +U+9a+R, +U+M+Vt, +UR"+W, +X+P+Y+T, +X+Qd+W, +XS+Z+a, +b+D/+c, +b+Pe+d, +bS)+e, +f+5c+g, 9+9+3r, 9R+Y+g, +hHws, +h+MZ+i, +h+j)+k, +h+J+S+l, A+Q\}+m, ASc+l, +nQd+k, +n+5+Z+m, +oKTo, +ovj+W, B+pYt, E;+q+W, EPh+a, F\textgreater +rr, +szV+c, +s+p+Kp, +s\$h+d, +sOi+g, +s\textbackslash +q+e, +t+uUs, +t\textbar Yq, +tM+r+i, D+vW+i.}
\end{sloppypar}

}

\subsection{\label{app:A-7}6-dim 236-1216}

\small{The master set is given in the repository.}

\smallskip
\small{{\bf 34-16} {\tt XTUVWY,YDEIRF,F34AN9,9587J6,67J2CX,,,QR*SW*X*Y*,KLMN*OP,GHI*J*OP,J*MN*PSV*,BC*F*U*X*Y*,A*J*SV*X*Y*,\break C*E*I*J*OY*,2*4*8*F*HM,1D*LQX*Y*,1BKT*X*Y*,3*5*E*GJ*O.}

{\bf 35-16} {\tt LYMNTO,OTCDSI,ISGJKH,H67VWQ,QW1PR5,54BUYL,,,U*V*W*XY*Z,P*Q*R*S*T*Z,EFN*O*S*Y*,D*J*K*R*S*X,89AB*M*Y*,\break 7*EFG*S*U*,239AFY*,6*C*N*P*V*W*,1*34*EU*Y*,28U*W*Y*Z.}

{\bf 36-18} {\tt RSTQaU,UFGHZP,P13KL2,2L5CEX,X7OWYB,B9ASTR,,,VW*X*Y*Z*a*,MNO*P*T*U*,IJK*L*S*a*,C*DE*H*L*Y*,\break 7*8A*B*S*T*,69*E*Q*X*Y*,8B*DH*K*Y*,3*45*JU*Z*,4DG*IU*Z*,6C*O*VX*Y*,8A*NVX*Y*,1*F*K*L*MP*.}

{\bf 37-16} {\tt aWXYZb,bMNOPV,VQTUSR,RSDIJH,H89ABC,C567Ka,,,I*J*K*La*b*,FGH*LO*P*,D*ET*U*Z*b*,EFGJ*N*Y*,347*B*H*I*,\break 246*B*H*X*,13A*B*H*L,129*K*W*a*,K*M*Q*W*X*b*,5*8*J*Y*Z*a*.}

{\bf 37-17} {\tt WXbYZa,a56BUA,ABU789,934CLF,FEGHRT,TDNXbW,,,U*VY*Z*a*b*,PQR*ST*V,MN*OST*V,IJKL*T*b*,C*D*G*H*OT*,\break 24*6*U*Y*Z*,E*JKQT*U*,123*IKL*,15*8*IJL*,7*IMPT*V,7*C*F*L*T*V.}

{\bf 37-18} {\tt aYZbXW,WXEFGL,LIJKVH,HVCRSA,A13MNT,T2PQU4,4578Oa,,,Q*R*S*T*U*V*,P*V*Y*Z*a*b*,M*N*O*X*a*b*,\break G*J*K*O*V*b*,DE*F*U*W*X*,BC*I*W*X*Z*,8*9A*T*W*b*,67*9A*S*T*,BDP*Q*V*Y*,2*3*5*A*N*T*,1*6A*N*R*T*.}

{\bf 37-19} {\tt YabWXZ,ZHIJLK,KLABVP,P48MNQ,QDEFGC,C7RTUS,SU6abY,,,R*S*T*U*V*b*,M*N*OP*Q*a*,B*F*G*Q*V*b*,9E*OP*Q*Y*,\break 7*8*F*G*N*Q*,6*9E*OX*Y*,A*I*J*T*U*b*,5G*H*V*W*b*,234*N*P*Q*,136*9D*X*,126*S*U*V*,5F*R*U*V*b*.}

{\bf 39-17} {\tt UdTWXV,VRcJNI,IJNMbH,HLPQS1,1S26BA,A4589E,EDFaZG,GZKYdU,,,Y*Z*a*b*c*d*,OP*Q*R*S*X*,K*L*M*N*S*W*,\break B*CF*G*Y*a*,6*78*9*A*C,2*35*7A*d*,1*4*D*G*N*a*,1*G*OT*a*d*,3N*S*a*b*c*.}

{\bf 39-18} {\tt SWTXVU,UV6CR7,7R34Qd,d12BZP,PZEacJ,JMNLOK,KLOIb5,58DYWS,,,Y*Z*a*b*c*d*,P*Q*R*W*X*d*,GHI*N*O*c*,\break D*E*FM*O*c*,9AB*C*U*V*,8*FO*T*V*c*,2*3*4*R*a*d*,6*AO*Y*b*d*,1*9S*U*V*X*,GHQ*R*U*W*.}

{\bf 44-18} {\tt MNcPQO,OPQ9iL,LiVghW,WHIJbK,KbEFGa,aUBCDA,ABCD52,236YZ4,4YZ1XM,,,defg*h*i*,X*Y*Z*a*b*c*,RSTU*V*W*,\break 8D*I*J*V*b*,B*C*F*G*L*f,7E*H*U*b*e,5*6*9*N*Y*Z*,1*3*deg*h*,78A*U*b*d.}

{\bf 44-19} {\tt XYacbZ,Z123g6,6g78VU,UVRSTW,WEFKLG,GBCJhD,D4AIiN,NiOPQM,M5YacX,,,defg*h*i*,HI*J*K*L*W*,9A*F*L*Q*W*,\break 5*8*O*P*S*T*,3*7*N*Q*ef,4*HL*a*b*c*,9E*W*Y*Z*c*,1*2*L*R*b*g*,L*W*X*c*dh*,B*C*M*O*P*d.}

{\bf 48-18} {\tt bcdefg,g8YZhA,A7BCTH,HIJKUL,LFGWQP,PQOVNM,MNRSE6,6E459b,,,h*ijklm,XY*Z*af*g*,R*S*T*U*V*W*,W*ae*g*lm,\break B*C*DE*jk,9*A*c*d*g*i,7*A*G*O*W*g*,4*5*6*8*DX,6*DF*I*J*K*,123H*Q*W*.}

{\bf 48-19} {\tt XYagmZ,Z29ckj,jBMNOL,LMNOfP,PIKTeJ,J5DEd6,637bVU,UVQRlW,W14YaX,,,hij*k*l*m*,b*c*d*e*f*g*,ST*U*V*W*l*,\break CD*E*FGH,B*FGHK*k*,7*89*Aj*k*,3*4*Aj*k*m*,6*R*a*b*c*g*,1*2*8Q*a*l*,5*CI*Sc*g*.}

{\bf 51-19} {\tt bajdec,cdeYZX,X78RSL,LGJKIH,HI6QhV,VTUWio,o123nD,D9ACpF,FpEfgb,,,klmn*o*p*,f*g*h*i*j*p*,R*S*V*W*h*o*,\break MNOPQ*W*,BC*D*Z*a*p*,8*A*D*E*Y*p*,7*9*BL*T*U*,56*D*J*K*i*,47*G*OPi*,45MNh*p*.}

{\bf 54-20} {\tt ecdfmg,gMNOPs,s5BCrR,RQVWXU,UVWXST,T6D897,7894Ak,k123iF,FEGHLK,KLIJle,,,nopqr*s*,hi*jk*l*m*,YZabf*g*,\break A*B*C*D*br*,G*H*J*d*m*q,I*e*f*jm*r*,4*5*9*Q*hi*,2*3*e*g*hq,1*E*bc*r*s*,1*6*9*F*S*e*.}

{\bf 57-21} {\tt lmnovp,pUVZat,tJKLYq,q12349,95678u,uMNPSc,cbdejk,kOQTgr,rBDsml,,,q*r*s*t*u*v*,fg*hij*k*,WXY*Z*a*u*,\break RS*T*hik*,P*Q*d*e*k*s*,HIK*L*V*o*,EFGWXo*,CD*HIJ*n*,9*AB*CS*U*,AEFGo*t*,M*N*Rn*t*u*,O*b*fk*n*v*.}

{\bf 59-22} {\tt wstuvx,xUVWXD,D9AkCB,BC5IJn,nmoqrp,pqregf,f8KQaZ,ZaYbdc,cMSTON,NOLPiw,,,hi*jk*lx*,U*V*W*X*g*o*,\break P*Q*RS*T*d*,K*L*b*f*lw*,EFGHI*J*,678*L*Rj,5*b*d*m*q*r*,67L*M*Y*h,9*A*K*e*r*v*,Y*b*c*k*q*x*,1234k*w*,\break 5*Hln*v*w*.}

{\bf 60-23} {\tt tuvwyx,xyKMTL,LCDFSR,RSPQsr,rsopqn,n34AEN,NOcdeb,bcdefg,gWYZaX,XaGHIi,ihklmj,jklmVU,UJuvwt,,,\break T*V*Y*Z*a*f*,E*F*M*T*U*q*,A*BI*J*O*h*,BD*K*Q*o*p*,6789J*V*,H*K*T*U*V*a*,5C*H*P*T*U*,3*4*C*G*P*W*,12BD*F*V*,\break 3*4*5A*N*a*.}

{\bf 61-23} {\tt klmnpo,opghij,jMOPVN,N45Krq,q3RSvA,A678y9,912uzW,WxYZaX,XYZasf,ftcdeb,bcdeUH,HFGLlk,,,u*v*wx*y*z*,\break q*r*s*t*y*z*,QR*S*TU*V*,K*L*O*P*V*n*,IJV*h*i*m*,BCDETz*,A*H*V*wy*z*,3*CDEr*z*,A*F*G*M*g*w,4*5*9*IJu*,\break 1*2*9*BQz*.}

}

\subsection{\label{app:A-8}8-dim 120-2024}

\small{The master set is given in the repository.}
\smallskip

\small{{\bf 37-11} {\tt 789A56CB,BCDEFIHG,GHWXYVRP,PROQ3487,,,123*4*5*6*7*8*,JKLMNI*A*C*,STUV*Q*R*MN,ZaY*ULF*28*,ZaX*TKE*17*,\break bJD*I*9*A*B*C*,bW*SO*R*I*A*B*.}

{\bf 38-11} {\tt IMHKLG8J,JaX9A35Z,Z5YPBC4T,TUQRE7VS,SVcbWOMI,,,123*4*5*67*8*,9*A*B*C*DE*FG*,NO*P*Q*R*M*F8*,\break W*X*U*V*K*L*D6,b*a*Y*V*O*I*25*,c*NO*H*I*M*18*.}

{\bf 39-11} {\tt 26134578,8OPJDFGC,CDFG9ABE,EcdbUQHZ,ZaXYTM62,,,H*IJ*KF*G*7*8*,LM*NKE*G*5*6*,Q*RST*P*N4*6*,\break U*VWX*Y*I3*8*,b*WS9*A*B*E*1*,c*d*a*VRO*LF*.}

{\bf 40-11} {\tt FGHIDEKJ,JKLMPQON,NOdeX26Z,Z6YTU4ba,abcRS38W,W8VBC7GF,,,12*3*4*56*7*8*,9AB*C*D*E*F*G*,\break R*S*T*U*P*Q*H*I*,X*Y*9AF*G*56*,d*e*c*V*L*M*18*.}

{\bf 41-11} {\tt 27135684,48WXYUMV,VZQNcdPb,bcdPJKCB,BC9DEFGA,AefaSI72,,,HI*J*K*LM*F*G*,N*OP*E*5*6*7*8*,Q*RS*TU*OP*D*,\break Z*a*X*Y*TL3*7*,e*f*W*RH9*1*8*.}

{\bf 42-11} {\tt 9ADEFGCB,BCdYU25a,a5bZV4LK,KLIJMNOH,HcPQfg6e,efg6T1A9,,,1*2*34*5*6*78,P*Q*RSO*F*G*8,T*U*V*WXN*E*7,\break Y*Z*WXSM*D*6*,c*d*b*RI*J*35*.}

{\bf 43-11} {\tt 23451876,67YWXRLZ,ZghJdefc,cdefABC9,9ABCEFGD,DabK3452,,,HIJ*K*L*MNG*,OPQR*SNF*8*,TUVW*X*SME*,\break h*f*Y*VPQI1*,g*f*a*b*TUOH.}

{\bf 49-13} {\tt 12346785,56789ABC,CTVWXLMU,UgeDijkh,hijkabcZ,ZabcYdGS,SNOPQR21,,,D*EFG*HI4*8*,JKL*M*HI3*7*,\break e*fd*X*R*EFB*,g*fV*W*P*Q*JK,lmnI9*A*7*8*,lmnY*T*N*O*I.}

{\bf 52-13} {\tt DE9ABCGF,FGRSTUQP,PQ12IJKH,HIJKLMON,NOjklhim,mopgde8n,nq56WXYV,VWXYZaED,,,1*2*345*6*78*,\break bcd*e*T*U*B*C*,fg*h*i*Z*a*9*A*,q*o*p*l*fbc7,j*k*R*S*L*M*34.}

{\bf 52-16} {\tt 56123487,789ABCDE,EFGJKLIH,HIpfgOih,hijkobRn,nobReaSX,XSTUVW65,,,MNO*PQR*S*L*,YZa*b*X*R*J*K*,\break cde*b*PQI*D*,f*g*h*i*j*k*V*W*,lmO*H*B*C*5*6*,p*o*e*a*b*I*D*E*,lmf*g*j*k*9*A*,qmYZU*F*G*A*,qlcdT*MNA*.}

\begin{sloppypar}
\smallskip
{\bf 100-48} {\tt 12345678, 9ABCDEF8, GHIJKLMF, NOPQRSM7, TUVWXYZ6, abcdefKL, ghijkfZS, lmnopqSE, rstuvkeZ, wxyeXYLF, z!vqcdS6, "\#\$!ujR6, \%\&'(eYPQ, )*ikQJ45, -/!yvkWD, :;\textless =WD37, \textgreater  tbdeQL2, ?=ybeWZC, @[\textbackslash ]\textless pqK, \textasciicircum \textgreater  *qiQJE, \_\$ujZC68, `\{\textbar \}xdLB, \textasciitilde \}];=\$zI, +1\textless (opVZ5, +2+3](nQM2, \textbackslash sHKMCF7, +4+5\textbar \#rhRA, +5)mnokUS, +6sniaKM7, +3\textbackslash =slqdT, \textasciicircum \textbackslash fUYKF7, +7'"mgOBE, +1\{*"OBE1, \{\}\textasciicircum \textbackslash /uvE, \textbackslash ]\textgreater  \textless wvqP, +8\_:;\textless wo3, +5\textbar ;sna17, +7rsvVZOM, \_:\&moM38, `\{\textbar ?WZBC, \textless =qhiZSA, +9+4\textbackslash ]=AC4, +6xYP4678, +A+4:=\#\$iT, +6[)oYI48, \textbar (!xsdRB, +7+4+5\$nL26, +7;\textless =hbWZ.}

\smallskip
{\bf 110-54} {\tt 12345678, 9ABCDEFG, HIJKLMN8, OPQRSNFG, TUVWXYSN, ZabYSEG7, cdefbXSM, ghijefD6, klmnijf6, opqrstgh, uvwxfW57, yz!xtVWL, "\#\$\%\&'nC, ()*-dMBF, /:cUXAC4, ;\textless ndfXRS, =\textgreater  WQKLN8, ?@[!xjfb, \textbackslash ]\textasciicircum 'xN67, \_@[rsAB3, `\{\textbar ;\textless nRL, \}\textasciitilde \{\textbar [bPM, +1+2\_:*-zw, +3+4+1+2zwmW, \textbar *-qtlMF, +5+6+7+8J9G8, +9+7+8\&jG48, +A+5+6\textasciicircum )zQ7, +B+C+A[iYQB, +4@['xsbB, +Dzprkmij, +E+D`\textless \%v92, \textasciitilde ythVWXY, ?pqmIF15, +E+8+2\textasciitilde vhVW, +AwnjXOKN, +F+E\textbar cOAF3, +2\}zpmiWQ, +F\}\textbar niaOA, +G+F+9\textgreater  y!18, [\textgreater  y!r9C8, +H+6+3?=qm1, +7]('psmJ, +I+H+6\}yXO9, +A\textbar oqstln, +J\textasciitilde ]\textasciicircum \$'D6, +D\$zoprt9, ]=-'xVON, +I+HjXOK34, +K+4[\textless ywxs, \textbackslash pqthWXK, +K+D+4(\#\$eM, +I`/"y!WL, \textbar \textgreater  seabL5.}

\smallskip
{\bf 120-58} {\tt 12345678, 9ABCDEF8, GHIJKLMF, NOPQRSTE, UVWXYSTM, ZabcQRLD, defghijk, lmncJKLF, opqr4567, stuvwYCF, xyznjkKB, !"riXIAF, \#\$\%yzWPT, \&'(yzORD, )*-/:ghP, ;\textless =\textgreater  wqhH, ?=\textgreater  "zO38, @[\textbackslash ]GHJK, \textasciicircum \_vbUV38, `\{]\textless (whH, \textbar (OQRLDE, \}*-/:xu2, \textasciitilde \textgreater  yvqrVO, +1?:hXMF7, +2+3+4\textless \textgreater  va3, +5+6+7+4?:yg, +8+9+A\textbackslash /ufG, \textbar ;\%twWYC, +B+C+D+3\{'V9, +E['(wmgJ, +F`-/xn26, +7*:yrgZP, +G+H+I+J+4!pI, +K+J+DxkVM7, +L@]?XRS6, +A+3@\textgreater  *sph, +M+6)gkaGL, +5+2\textasciitilde /vJAC, +C\{\_\textless zofT, +N+B+5+6+2+4?", +K+D\%teWY5, @ycURTK4, +I+2+3+4\&\$uv, +I+J=\&zOD8, +O+I+2*eIAD, +3'\$!sphF, +O)oqOLDE, +K+D\textgreater  VWYHM, +B+C+D)P9E1, +H+B!ZXY12, +P+7\textasciitilde tqeZ9, +H\%QGIJL5, \}deikO23, +Q+O[\textbackslash 'sfD, +E+9+A+6(peO, \textasciicircum [/\#vbP3, +R+S+T+N+G['", +U+8+2(lUE3.}

\end{sloppypar}

}

\subsection{\label{app:A-9}16-dim 80-265}

\small{The master set is given in the repository.}
\smallskip

\begin{sloppypar}
  
\smallskip  
\small{{\bf 72-12} {\tt 123456789ABCDEFG, HIJKLMNOPQRSTEFG, UVWXYZabcdefghST, ijklmnopfghOPQRT, qrstuvnopdeMNPQR, wxyz!uvpceNRABCD, "\#\$!stvpZabcdeNQ, \%\&'(yzrlmYhP89CD, )*-\#\$vXbcdKL567G, *-'(qjkWhIJP347G, )\%\&"wxiUVbHT12CD, tuopHJLMNOQRST7F.}

{\bf 72-13} {\tt 123456789ABCDEFG, HIJKLMNOPQRSTUVW, XYZabcdefgVWDEFG, hijklmnefgUW9ABC, opqrstuvcdST78FG, wxyz!tuvmnRT56BC, "\#\$\%rsabNOPQ34FG, \&'(!qvjklXYZfgMW, )*-(\$\%!qulZgKLPQ, )*-\#\%z!quiJLPQ2C, *-("\$yhlZgIKPQ1B, )-'\%x!pqulZfHLQT, )*\&(\$w!oquvlZgKP.}

\smallskip  
{\bf 73-13} {\tt 123456789ABCDEFG, HIJKLMNOPQRSTUVW, XYZabcdefghijklm, nopqrshijklmTUVW, tuvwrsNOPQRSVWFG, xyzwqsbcdefglmDE, !"\#zfgKLMQRSUWCE, \$\%\#vegIJLMOPRSBG, \&'()yupsHQUV789A, *\%"xtoZadjkm56EG, -/()\$!Yacgik349A, -/*\&'nbmJMPS129A, XabcdefgijlmJLOS.}

\smallskip  
{\bf 74-19} {\tt 123456789ABCDEFG, HIJKLMNOPQRSDEFG, TUVWXYZaPQRS9ABC, bcdefghijkZaNORS, lmnopghijkZaNOQS, qrstuefWXY678ABC, tupdVXYaMOQR78BC, vwxyz!udjkVYMR8C, "z!snocfUXYZPSBC, \#\$\%\&y!rtmoTY58AB, tulmnopXYaOQ78BC, '(\%\&xyz!lmnoiULS, y!rmoTUXYZPS5ABC, xyz!lmnoUXYZPSBC, )*-/wscehkJK34FG, /"qspUVXYZMPQSBC, :-(\$\&vbcdefgkaOR, :)*'\#\&x!mnHI12FG, y!rmoghijkTZNS5A.}
  
\smallskip  
{\bf 77-14} {\tt 123456789ABCDEFG, HIJKLMNOPQRSTUVW, XYZabcdefghijklm, nopqrstujklmTUVW, vwxyz!"\#9ABCDEFG, \$\%\&'"\#rstuhiRSFG, ()\&'!\#pqtufgPQEG, *-/:delmNOVW5678, ;\textless =/:)\%'oqsu3478, \textless =*-bcgiLMQS1278, ;(\$'yznqstaiMRCD, wxyzZekmKOUWABCD, xz!"YeklKNTW9BDG, xy!"XegkKLPW9ADG.}

\smallskip  
{\bf 80-23} {\tt 123456789ABCDEFG, HIJKLMNOPQRSTUVW, XYZaNOPQRSTUVWFG, bcdeHIJKLMRSTUVW, fghijaMPQUVWBCDE, klmnojLMOQTW789A, pqrsojbcdeZaSTVW, tuvwfghi3456BCDE, xyz!vwrshiJK56DE, "\#\$\%\&'()z!ndeYAG, *-/:'()qsmceIK29, ;:pslYIJNPSV12AF, \textless =\textgreater /xylnojbcLM8A, ;*-:()qsIKNPSV12, ?-\&)jZLNOPQSTUVW, =\textgreater \$\%xyz!lnojLM8A, \textgreater \%y!XYZaOPRSUWFG, @?=*-\#\$\&()xzNQTV, ?*/"\#)tuhibe56BC, \textgreater :\$x!klnXY189AFG, \textless *-/"'()qsbcdeIK, "\$\%'z!deXYZaRUFG, @?*-\#\&()bcdeSTVW.}
}      
\end{sloppypar}

\subsection{\label{app:A-10}32-dim 160-661}

\small{The master set is given in the repository.}
\smallskip

\begin{sloppypar}

\small{{\bf 135-21} {\tt 123456789ABCDEFGHIJKLMNOPQRSTUVW, XYZabcdefghijklmnopqrstuvwRSTUVW, xyz!"\#\$\%\&'()*-/:;\textless =\textgreater ?@[\textbackslash ]\textasciicircum vwPQVW, \_`\{\textbar \}\textasciitilde +1+2+3+4]\textasciicircum mnopqrstuJKLMNOPQUVW, +5+6\&'()*-/:;\textless =\textgreater ?@[\textbackslash ABCDEFGHIMNOPQ,\break +7+8+9+A+B+1+2+3+4efghijklpqrstu6789GHIST, +C+D+E+F+G+H+A+B\textasciitilde +4?@[\textbackslash kl2345DEFGHIKLOQTU, +I+J+K+L+M+N+O+P+Q+E+F+G+H+8+9+B+5+6cdjl1589BCFIST, +R+S+T+U+V+W+X+Y+O+P+Q+H+8+9+A\}bdiklw13457CILUW, +Z+a+N+H+9+A+B`\{\textbar \}+1+2+3*-/:;\textless =\textgreater \textasciicircum bikow67JW, +K+L+M+G+8\_ablw12345789BCFHIKLORSTUVW,\break  +J+M+7+8+A+1+2+3+4Zaiklnorstu34689GILNPQS, +Z+a+I+L+M+N+E+F+H+9+B`\{\textbar +1+2+3*-/:;\textless =\textgreater \textasciicircum jo689J, +V+W+X+Y+7+8+9+B;\textless =\textgreater XYcdfghijlpqrstu79GH, +J+M+7+9+B\_+1+2+3+4]\textasciicircum mrstuv34679HIJLNPQRV, +Z+a+M+C+D+7+8+9+A\_`\{\textbar \textasciitilde +1+2+3+4]\textasciicircum rstu4HIJLNPQ, +Z+a`\{\textbar +1+2+3\&'()*-/:;\textless =\textgreater ?@[\textbackslash ADEGMNPQ, +b+M+Q+Bxyz!"\#\$\%Ybdehijklmvw3HRSTUVW, +c+d+K+L+M+N+P+Q+D+G+H+8+9+B`\{\textbar \textasciitilde +3/:=\textgreater blw4789UW, +R+S+T+U+V+W+X+Y+1+2+3+4*-/:;\textless =\textgreater Zafgnopqrstu, +e+f+g+h+i+j+b+J+K+N+F+G+B+6ZeghijknqtuvBFHKST.}

\smallskip
{\bf 144-11} {\tt 123456789ABCDEFGHIJKLMNOPQRSTUVW, 1XYZa5bcdef9AgBChiGjklmnopqNrstu, 1v3mwxpyzfcdFH!"l\#N\$B\%\&'4()seA*-, /!1XY5b:;c\#df\textless A=\textgreater C?h\&klmn@\$s[\textbackslash ]\textasciicircum , \_;`f\{:\textbar d\}\textasciitilde +1+2!+3+4k+5+6+7+8?h\&\textasciicircum +9+A+B\textgreater ]+C+D+E,\break  \_4(5+9+F+B+G+H+I8gB+JC+K\%+LF+3i+Mm+N+5Ns+O+8V+2+D, Zj+P+Qxz+R+SPQ+T+U"+Vaq+WOE+X+Y+ZI+aKM*+b+c+d2+e, +f+Z+Q+V+g+H+I\textasciitilde +K+7+h+dE\}+3+i+j+MJK+kL+l+5\textbar OQR+mT+2+D, jZ+Pw@`yz+RrD\textbackslash F+n+o+W7+N+6+X\textasciicircum 4)o+p+O*+q]-+r+s, /!+f+n6c\textless =+b+C+c+hE+Uh\&+okl+aK+k+A[OQS+mU]W\textasciicircum , +E+r+s+4Z+p+F+g+G\{+T+q+S+J+L+e+3+1+ij+j+Yv'+l+5zt*+2+Du.}

\smallskip
{\bf 144-13} {\tt 123456789ABCDEFGHIJKLMNOPQRSTUVW, XYZabcdefghijklmnopqrstuPQRSTUVW, vwxyz!"\#\$\%\&'()*-/:;\textless rstuHIJKLMNO, =\textgreater ?@[\textbackslash ]\textasciicircum \_`\{\textbar \}\textasciitilde +1+2/:;\textless nopqDEFGLMNO, +3+4+5+6+7+8+9+A+B+C+D+E\}\textasciitilde +1+2()*-jklm9ABCTUVW, +F+G+H+I+J+K+L+M+N+O+P+Q+R+S+T+U\_`\{\textbar \$\%\&'5678LMNO, +V+W+X+Y+Z+a+b+c+N+O+P+Q+R+S+T+U+B+C+D+Efghi1234TUVW, +d+e+f+g+h+i+j+k+Z+a+b+c+J+K+L+M+R+S+T+U+9+A]\textasciicircum +2"\#-dequ,\break  +l+m+n+o+h+i+j+k+X+Y+b+c+H+I+L+M+P+Q+T+U+7+8+9+A\textbackslash \textasciicircum +1+2!\#*-,\break  +p+q+l+m+n+o+j+k+X+Y+b+c+H+I+L+M+P+Q+T+U+9+A[\textasciicircum +1z\#*bcpt, +r+s+q+n+o+f+g+i+5+6+8+9+A?@[\textbackslash \textasciicircum \textasciitilde +1+2xy\#)*Zacpqu, +p+q+l+m+d+e+j+k+V+W+X+Y+b+c+F+G+H+I+L+M+P+Q+R+S+T+U\textgreater wXYpt,\break  +r+s+q+l+m+n+o+i+b+c+L+M+P+Q+3+4+8+9+A=[\textbackslash \textasciicircum \textasciitilde +1+2v\#)*cp.}

\smallskip
{\bf 146-14} {\tt 123456789ABCDEFGHIJKLMNOPQRSTUVW, XYZabcdefghijklmnopqrstuPQRSTUVW, vwxyz!"\#\$\%\&'()*-/:;\textless =\textgreater ?@HIJKLMNO, [\textbackslash ]\textasciicircum \_`\{\textbar \}\textasciitilde +1+2+3+4+5+6+7+8+9+A;\textless =\textgreater ?@DEFGNO, +B+C+D+E+F+G+H+I+J+K+5+6+7+8+9+A*-/:rstu9ABCTUVW, +L+M+N+O+P+Q+R+S+T+U+V+I+J+K+2+3+4+8+9+A()/:jklmnopq, +W+X+Y+Z+T+U+V+H+K\}\textasciitilde +1+3+4+6+7+9+A\&')-:@hipq8GMO, +a+b+c+d+e+P+Q+R+S+F+G+I+J\{\textbar +2+5+8\#\$\%(*/fgno7FKL, +f+g+h+i+j+k+l+m+e+Z+R+S+V+F+G+H+J+K\textbar +1+2+3+4+5+6+7+8+A\%'/:, +d+S+G+I\textbackslash ]\textasciicircum \_`\textbar \textasciitilde +1+2+4+5+7+8+A"\$')-:\textless =\textgreater ?EGJL, +n+o+p+q+r+s+j+k+l+m+b+c+X+Y+L+M+N+O\_`z!\textgreater ?bcdenopq, +t+u+p+q+r+s+h+i+l+m+B+C+D+E]\textasciicircum xy\textgreater ?XYZa3456TUVW, +t+u+n+o+f+g+l+m+a+e+W+Z+J+K[\textbackslash +2+4+5+7vw\$@12FGJKMN, +e+Z+R+V+F+H+J+K[\{\textbar +1+2+4+5+7+8+A"\$\%(*/;@DFJLNO.}

\smallskip
{\bf 160-29} {\tt 123456789ABCDEFGHIJKLMNOPQRSTUVW, XYZabcdefghijklmnopqrstuvwRSTUVW, xyz!"\#\$\%\&'()*klmnopqrstuvwOPQUVW, -/:;\textless =\textgreater ?@[\textbackslash ]\textasciicircum \_`\{\textbar \}\textasciitilde +1+2+3+4+5\#\$\%\&'()*, +6+7+8+9+A+B+C+D+E+F+G+3+4+5!"()*ghijtuvwMNPQW, +H+I+J+K+F+G+2+5yz\%\&'()*efijlmnopqrstuvw, +L+M+N+O+P+Q+J+K+F+G\textasciitilde +1+2+3+4+5\%\&'()*defhijpqrs, +R+S+T+U+V+W+X+Y+Z+a+b+c+H+I+G\}!"\$\&'*cflmnovwKL, +d+e+f+g+h+i+j+k+O+P+Q+I+F\_`\{\textbar \%'()bejknoGHIJT, +l+m+k+W+X+Y+Z+a+b+c+D+E+2+5yz!"bcmuvwCDEFIJLS, +j+V+c+L+M+N+O+P+Q+C\textasciicircum \}\textasciitilde +1+3+4!"\#\$XYZacdghpqrs, +n+o+p+m+h+i+j+U+a+b+K+8+9+A+B+D+E+2+5yz\&ahiuBEFORT, +o+p+h+i+V+c+M+N+P+Q+H+I+8+9+A+B\}+1+2+4yz!"\#'Zcpqrs, +q+l+k+c+K+7+A+B+D+E]\{\textbar +3+4+5z!"*cfhtuvwDFJNQ, +r+s+n+m+j+S+T+Y+Z+H+I+6+7\textasciicircum \textasciitilde +1+3y\$'YalmBKLMNQRT, +t+s+m+g+Z+J+6+7+8+9+A+B+C+E+F\textasciicircum +3+4+5"\#agjtwAHKMNQ, +u+v+t+s+d+e+f+R-/:;\textless =\textgreater ?@[\textbackslash ]\textasciicircum \_`\{\textbar \textasciitilde +3+5\$\&Ya, +w+x+l+j+k+L+O+J+K+C+F+GXdefghij56789CDGIJMN, +y+z+!+"+\#+\$+u+v+n+e+f+i+U+Z+b+6+9+A+B+D\textless =\textgreater ?@[\textbackslash `34EL, +\%+\&+'+\#+\$+x+X+6+9+B+G:;[\textbackslash +2+5!cfhsuw23489FNO, +L+N+Q+K:;\textgreater ?@\textbackslash \textasciicircum \_\textbar \}+2+5\#\$\&'YZacdioqrtvw, +"+r+R+Y+c+H+K+6+7+D+G\textasciitilde +3+4z!*ZcfhilmtuvFMNOQ, +\&+'+y+z+!+\$+w+x+v+f+g+L+M+N+O+P+C=x()egrs14789HV, +(+\%+\&+'+y+z+!+"+\#+w+x+r+n+l+g+k+T+U+W+X+Zxbk689GHQSW, +(+\%+!+q+o+p+l+i+j+k+S+V+I/;=]\_`\{\$bnop89CJRST, +u+s+p+d+f+8+9+A+C+F-/:;\textless =\textgreater ?@[\textbackslash ]\_`\{\textbar \%)gjov, +(+\%+\&+'+!+w+o+l+i+j+V+I/;\textless =@\textbackslash ]`\{\$bopr289CJR, +j+a+c+I+J+C+E+F-/:;\textless ?[\textbackslash \textasciitilde +1+3+4y"XYZacegjmu, +)+t+s+R+T+Y+Z+H+6+7+E+G\textasciicircum +1+3+4+5"\#*YflmtwKLMNQU.}
}

\end{sloppypar}

}


\begin{thebibliography}{70}
\expandafter\ifx\csname natexlab\endcsname\relax\def\natexlab#1{#1}\fi
\expandafter\ifx\csname bibnamefont\endcsname\relax
  \def\bibnamefont#1{#1}\fi
\expandafter\ifx\csname bibfnamefont\endcsname\relax
  \def\bibfnamefont#1{#1}\fi
\expandafter\ifx\csname citenamefont\endcsname\relax
  \def\citenamefont#1{#1}\fi
\expandafter\ifx\csname url\endcsname\relax
  \def\url#1{\texttt{#1}}\fi
\expandafter\ifx\csname urlprefix\endcsname\relax\def\urlprefix{URL }\fi
\providecommand{\bibinfo}[2]{#2}
\providecommand{\eprint}[2][]{\url{#2}}

\bibitem[{\citenamefont{Hein et~al.}(2004)\citenamefont{Hein, Eisert, and
  Briegel}}]{heinbriegel04}
\bibinfo{author}{\bibfnamefont{M.}~\bibnamefont{Hein}},
  \bibinfo{author}{\bibfnamefont{J.}~\bibnamefont{Eisert}}, \bibnamefont{and}
  \bibinfo{author}{\bibfnamefont{H.~J.} \bibnamefont{Briegel}},
  \bibinfo{journal}{{\it Phys. Rev. A}} \textbf{\bibinfo{volume}{{\bf 69}}},
  \bibinfo{pages}{062311} (\bibinfo{year}{2004}).

\bibitem[{\citenamefont{Cabello and Moreno}(2010)}]{cabello-moreno-09}
\bibinfo{author}{\bibfnamefont{A.}~\bibnamefont{Cabello}} \bibnamefont{and}
  \bibinfo{author}{\bibfnamefont{P.}~\bibnamefont{Moreno}},
  \bibinfo{journal}{{\it Phys. Rev. A}} \textbf{\bibinfo{volume}{{\bf 81}}},
  \bibinfo{pages}{042110} (\bibinfo{year}{2010}).

\bibitem[{\citenamefont{Cabello}(2010)}]{cabello-10}
\bibinfo{author}{\bibfnamefont{A.}~\bibnamefont{Cabello}},
  \bibinfo{journal}{{\it Phys. Rev. Lett.}} \textbf{\bibinfo{volume}{{\bf
  104}}}, \bibinfo{pages}{2210401} (\bibinfo{year}{2010}).

\bibitem[{\citenamefont{Bartlett}(2014)}]{bartlett-nature-14}
\bibinfo{author}{\bibfnamefont{S.~D.} \bibnamefont{Bartlett}},
  \bibinfo{journal}{{\it Nature}} \textbf{\bibinfo{volume}{{\bf 510}}},
  \bibinfo{pages}{345} (\bibinfo{year}{2014}).

\bibitem[{\citenamefont{Howard et~al.}(2014)\citenamefont{Howard, Wallman,
  Veitech, and Emerson}}]{magic-14}
\bibinfo{author}{\bibfnamefont{M.}~\bibnamefont{Howard}},
  \bibinfo{author}{\bibfnamefont{J.}~\bibnamefont{Wallman}},
  \bibinfo{author}{\bibfnamefont{V.}~\bibnamefont{Veitech}}, \bibnamefont{and}
  \bibinfo{author}{\bibfnamefont{J.}~\bibnamefont{Emerson}},
  \bibinfo{journal}{{\it Nature}} \textbf{\bibinfo{volume}{{\bf 510}}},
  \bibinfo{pages}{351} (\bibinfo{year}{2014}).

\bibitem[{\citenamefont{Delfosse et~al.}(2015)\citenamefont{Delfosse, Guerin,
  Bian, and Raussendorf}}]{delfosse-raussendorf-15}
\bibinfo{author}{\bibfnamefont{N.}~\bibnamefont{Delfosse}},
  \bibinfo{author}{\bibfnamefont{P.~A.} \bibnamefont{Guerin}},
  \bibinfo{author}{\bibfnamefont{J.}~\bibnamefont{Bian}}, \bibnamefont{and}
  \bibinfo{author}{\bibfnamefont{R.}~\bibnamefont{Raussendorf}},
  \bibinfo{journal}{{\it Phys. Rev. X}} \textbf{\bibinfo{volume}{{\bf 5}}},
  \bibinfo{pages}{021003} (\bibinfo{year}{2015}).

\bibitem[{\citenamefont{Raussendorf}(2013)}]{raussendorf-13}
\bibinfo{author}{\bibfnamefont{R.}~\bibnamefont{Raussendorf}},
  \bibinfo{journal}{{\it Phys. Rev. A}} \textbf{\bibinfo{volume}{{\bf 88}}},
  \bibinfo{pages}{022322} (\bibinfo{year}{2013}).

\bibitem[{\citenamefont{Waegell and Aravind}(2011)}]{waeg-aravind-jpa-11}
\bibinfo{author}{\bibfnamefont{M.}~\bibnamefont{Waegell}} \bibnamefont{and}
  \bibinfo{author}{\bibfnamefont{P.~K.} \bibnamefont{Aravind}},
  \bibinfo{journal}{{\it J. Phys. A}} \textbf{\bibinfo{volume}{{\bf 44}}},
  \bibinfo{pages}{505303} (\bibinfo{year}{2011}).

\bibitem[{\citenamefont{Pavi{\v c}i{\'c}
  et~al.}(2010{\natexlab{a}})\citenamefont{Pavi{\v c}i{\'c}, Mc{K}ay, Megill,
  and Fresl}}]{bdm-ndm-mp-fresl-jmp-10}
\bibinfo{author}{\bibfnamefont{M.}~\bibnamefont{Pavi{\v c}i{\'c}}},
  \bibinfo{author}{\bibfnamefont{B.~D.} \bibnamefont{Mc{K}ay}},
  \bibinfo{author}{\bibfnamefont{N.~D.} \bibnamefont{Megill}},
  \bibnamefont{and} \bibinfo{author}{\bibfnamefont{K.}~\bibnamefont{Fresl}},
  \bibinfo{journal}{{\it J. Math. Phys.}} \textbf{\bibinfo{volume}{{\bf 51}}},
  \bibinfo{pages}{102103} (\bibinfo{year}{2010}{\natexlab{a}}).

\bibitem[{\citenamefont{Megill and Pavi{\v c}i{\'c}}(2011)}]{mp-7oa}
\bibinfo{author}{\bibfnamefont{N.~D.} \bibnamefont{Megill}} \bibnamefont{and}
  \bibinfo{author}{\bibfnamefont{M.}~\bibnamefont{Pavi{\v c}i{\'c}}},
  \bibinfo{journal}{{\it Ann. {H}enri {P}oinc.}} pp.
  \bibinfo{pages}{1417--1429} (\bibinfo{year}{2011}).

\bibitem[{\citenamefont{Cabello et~al.}(2011)\citenamefont{Cabello,
  {D'A}mbrosio, Nagali, and Sciarrino}}]{cabello-dambrosio-11}
\bibinfo{author}{\bibfnamefont{A.}~\bibnamefont{Cabello}},
  \bibinfo{author}{\bibfnamefont{V.}~\bibnamefont{{D'A}mbrosio}},
  \bibinfo{author}{\bibfnamefont{E.}~\bibnamefont{Nagali}}, \bibnamefont{and}
  \bibinfo{author}{\bibfnamefont{F.}~\bibnamefont{Sciarrino}},
  \bibinfo{journal}{{\it Phys. Rev. A}} \textbf{\bibinfo{volume}{{\bf 84}}},
  \bibinfo{pages}{030302(R)} (\bibinfo{year}{2011}).

\bibitem[{\citenamefont{Nagata}(2005)}]{nagata-05}
\bibinfo{author}{\bibfnamefont{K.}~\bibnamefont{Nagata}},
  \bibinfo{journal}{{\it Phys. Rev. A}} \textbf{\bibinfo{volume}{{\bf 72}}},
  \bibinfo{pages}{012325–} (\bibinfo{year}{2005}).

\bibitem[{\citenamefont{Simon et~al.}(2000)\citenamefont{Simon, Weinfurter,
  ${\dot{\rm Z}}$ukowski, and Zeilinger}}]{simon-zeil00}
\bibinfo{author}{\bibfnamefont{C.}~\bibnamefont{Simon}},
  \bibinfo{author}{\bibfnamefont{H.}~\bibnamefont{Weinfurter}},
  \bibinfo{author}{\bibfnamefont{M.}~\bibnamefont{${\dot{\rm Z}}$ukowski}},
  \bibnamefont{and}
  \bibinfo{author}{\bibfnamefont{A.}~\bibnamefont{Zeilinger}},
  \bibinfo{journal}{{\it Phys. Rev. Lett.}} \textbf{\bibinfo{volume}{{\bf
  85}}}, \bibinfo{pages}{1783} (\bibinfo{year}{2000}).

\bibitem[{\citenamefont{Michler et~al.}(2000)\citenamefont{Michler, Weinfurter,
  and ${\dot{\rm Z}}$ukowski}}]{michler-zeil-00}
\bibinfo{author}{\bibfnamefont{M.}~\bibnamefont{Michler}},
  \bibinfo{author}{\bibfnamefont{H.}~\bibnamefont{Weinfurter}},
  \bibnamefont{and} \bibinfo{author}{\bibfnamefont{M.}~\bibnamefont{${\dot{\rm
  Z}}$ukowski}}, \bibinfo{journal}{{\it Phys. Rev. Lett.}}
  \textbf{\bibinfo{volume}{{\bf 84}}}, \bibinfo{pages}{5457}
  (\bibinfo{year}{2000}).

\bibitem[{\citenamefont{Amselem et~al.}(2009)\citenamefont{Amselem,
  R{\aa}dmark, Bourennane, and Cabello}}]{amselem-cabello-09}
\bibinfo{author}{\bibfnamefont{E.}~\bibnamefont{Amselem}},
  \bibinfo{author}{\bibfnamefont{M.}~\bibnamefont{R{\aa}dmark}},
  \bibinfo{author}{\bibfnamefont{M.}~\bibnamefont{Bourennane}},
  \bibnamefont{and} \bibinfo{author}{\bibfnamefont{A.}~\bibnamefont{Cabello}},
  \bibinfo{journal}{{\it Phys. Rev. Lett.}} \textbf{\bibinfo{volume}{{\bf
  103}}}, \bibinfo{pages}{160405} (\bibinfo{year}{2009}).

\bibitem[{\citenamefont{Liu et~al.}(2009)\citenamefont{Liu, Huang, Gong, Sun,
  Zhang, Li, and Guo}}]{liu-09}
\bibinfo{author}{\bibfnamefont{B.~H.} \bibnamefont{Liu}},
  \bibinfo{author}{\bibfnamefont{Y.~F.} \bibnamefont{Huang}},
  \bibinfo{author}{\bibfnamefont{Y.~X.} \bibnamefont{Gong}},
  \bibinfo{author}{\bibfnamefont{F.~W.} \bibnamefont{Sun}},
  \bibinfo{author}{\bibfnamefont{Y.~S.} \bibnamefont{Zhang}},
  \bibinfo{author}{\bibfnamefont{C.~F.} \bibnamefont{Li}}, \bibnamefont{and}
  \bibinfo{author}{\bibfnamefont{G.~C.} \bibnamefont{Guo}},
  \bibinfo{journal}{{\it Phys. Rev. A}} \textbf{\bibinfo{volume}{{\bf 80}}},
  \bibinfo{pages}{044101} (\bibinfo{year}{2009}).

\bibitem[{\citenamefont{{D'A}mbrosio et~al.}(2013)\citenamefont{{D'A}mbrosio,
  Herbauts, Amselem, Nagali, Bourennane, Sciarrino, and
  Cabello}}]{d-ambrosio-cabello-13}
\bibinfo{author}{\bibfnamefont{V.}~\bibnamefont{{D'A}mbrosio}},
  \bibinfo{author}{\bibfnamefont{I.}~\bibnamefont{Herbauts}},
  \bibinfo{author}{\bibfnamefont{E.}~\bibnamefont{Amselem}},
  \bibinfo{author}{\bibfnamefont{E.}~\bibnamefont{Nagali}},
  \bibinfo{author}{\bibfnamefont{M.}~\bibnamefont{Bourennane}},
  \bibinfo{author}{\bibfnamefont{F.}~\bibnamefont{Sciarrino}},
  \bibnamefont{and} \bibinfo{author}{\bibfnamefont{A.}~\bibnamefont{Cabello}},
  \bibinfo{journal}{{\it Phys. Rev. X}} \textbf{\bibinfo{volume}{{\bf 3}}},
  \bibinfo{pages}{011012} (\bibinfo{year}{2013}).

\bibitem[{\citenamefont{Huang et~al.}(2003)\citenamefont{Huang, Li, Zhang, Pan,
  and Guo}}]{ks-exp-03}
\bibinfo{author}{\bibfnamefont{Y.-F.} \bibnamefont{Huang}},
  \bibinfo{author}{\bibfnamefont{C.-F.} \bibnamefont{Li}},
  \bibinfo{author}{\bibfnamefont{Y.-S.} \bibnamefont{Zhang}},
  \bibinfo{author}{\bibfnamefont{J.-W.} \bibnamefont{Pan}}, \bibnamefont{and}
  \bibinfo{author}{\bibfnamefont{G.-C.} \bibnamefont{Guo}},
  \bibinfo{journal}{{\it Phys. Rev. Lett.}} \textbf{\bibinfo{volume}{{\bf
  90}}}, \bibinfo{pages}{250401} (\bibinfo{year}{2003}).

\bibitem[{\citenamefont{Hasegawa et~al.}(2006)\citenamefont{Hasegawa, Loidl,
  Badurek, Baron, and Rauch}}]{h-rauch06}
\bibinfo{author}{\bibfnamefont{Y.}~\bibnamefont{Hasegawa}},
  \bibinfo{author}{\bibfnamefont{R.}~\bibnamefont{Loidl}},
  \bibinfo{author}{\bibfnamefont{G.}~\bibnamefont{Badurek}},
  \bibinfo{author}{\bibfnamefont{M.}~\bibnamefont{Baron}}, \bibnamefont{and}
  \bibinfo{author}{\bibfnamefont{H.}~\bibnamefont{Rauch}},
  \bibinfo{journal}{{\it Phys. Rev. Lett.}} \textbf{\bibinfo{volume}{{\bf
  97}}}, \bibinfo{pages}{230401} (\bibinfo{year}{2006}).

\bibitem[{\citenamefont{Cabello et~al.}(2008)\citenamefont{Cabello, Filipp,
  Rauch, and Hasegawa}}]{cabello-fillip-rauch-08}
\bibinfo{author}{\bibfnamefont{A.}~\bibnamefont{Cabello}},
  \bibinfo{author}{\bibfnamefont{S.}~\bibnamefont{Filipp}},
  \bibinfo{author}{\bibfnamefont{H.}~\bibnamefont{Rauch}}, \bibnamefont{and}
  \bibinfo{author}{\bibfnamefont{Y.}~\bibnamefont{Hasegawa}},
  \bibinfo{journal}{{\it Phys. Rev. Lett.}} \textbf{\bibinfo{volume}{{\bf
  100}}}, \bibinfo{pages}{130404} (\bibinfo{year}{2008}).

\bibitem[{\citenamefont{Bartosik et~al.}(2009)\citenamefont{Bartosik, Klep,
  Schmitzer, Sponar, Cabello, Rauch, and Hasegawa}}]{b-rauch-09}
\bibinfo{author}{\bibfnamefont{H.}~\bibnamefont{Bartosik}},
  \bibinfo{author}{\bibfnamefont{J.}~\bibnamefont{Klep}},
  \bibinfo{author}{\bibfnamefont{C.}~\bibnamefont{Schmitzer}},
  \bibinfo{author}{\bibfnamefont{S.}~\bibnamefont{Sponar}},
  \bibinfo{author}{\bibfnamefont{A.}~\bibnamefont{Cabello}},
  \bibinfo{author}{\bibfnamefont{H.}~\bibnamefont{Rauch}}, \bibnamefont{and}
  \bibinfo{author}{\bibfnamefont{Y.}~\bibnamefont{Hasegawa}},
  \bibinfo{journal}{{\it Phys. Rev. Lett.}} \textbf{\bibinfo{volume}{{\bf
  103}}}, \bibinfo{pages}{040403} (\bibinfo{year}{2009}).

\bibitem[{\citenamefont{Kirchmair et~al.}(2009)\citenamefont{Kirchmair,
  Z{\"a}hringer, Gerritsma, Kleinmann, G{\"u}hne, Cabello, Blatt, and
  Roos}}]{k-cabello-blatt-09}
\bibinfo{author}{\bibfnamefont{G.}~\bibnamefont{Kirchmair}},
  \bibinfo{author}{\bibfnamefont{F.}~\bibnamefont{Z{\"a}hringer}},
  \bibinfo{author}{\bibfnamefont{R.}~\bibnamefont{Gerritsma}},
  \bibinfo{author}{\bibfnamefont{M.}~\bibnamefont{Kleinmann}},
  \bibinfo{author}{\bibfnamefont{O.}~\bibnamefont{G{\"u}hne}},
  \bibinfo{author}{\bibfnamefont{A.}~\bibnamefont{Cabello}},
  \bibinfo{author}{\bibfnamefont{R.}~\bibnamefont{Blatt}}, \bibnamefont{and}
  \bibinfo{author}{\bibfnamefont{C.~F.} \bibnamefont{Roos}},
  \bibinfo{journal}{{\it Nature}} \textbf{\bibinfo{volume}{{\bf 460}}},
  \bibinfo{pages}{494} (\bibinfo{year}{2009}).

\bibitem[{\citenamefont{Moussa et~al.}(2010)\citenamefont{Moussa, Ryan, Cory,
  and Laflamme}}]{moussa-09}
\bibinfo{author}{\bibfnamefont{O.}~\bibnamefont{Moussa}},
  \bibinfo{author}{\bibfnamefont{C.~A.} \bibnamefont{Ryan}},
  \bibinfo{author}{\bibfnamefont{D.~G.} \bibnamefont{Cory}}, \bibnamefont{and}
  \bibinfo{author}{\bibfnamefont{R.}~\bibnamefont{Laflamme}},
  \bibinfo{journal}{{\it Phys. Rev. Lett.}} \textbf{\bibinfo{volume}{{\bf
  104}}}, \bibinfo{pages}{160501} (\bibinfo{year}{2010}).

\bibitem[{\citenamefont{Lison{\v e}k
  et~al.}(2014{\natexlab{a}})\citenamefont{Lison{\v e}k, Badzi{\c a}g,
  Portillo, and Cabello}}]{lisonek-14}
\bibinfo{author}{\bibfnamefont{P.}~\bibnamefont{Lison{\v e}k}},
  \bibinfo{author}{\bibfnamefont{P.}~\bibnamefont{Badzi{\c a}g}},
  \bibinfo{author}{\bibfnamefont{J.~R.} \bibnamefont{Portillo}},
  \bibnamefont{and} \bibinfo{author}{\bibfnamefont{A.}~\bibnamefont{Cabello}},
  \bibinfo{journal}{{Phys. Rev. A}} \textbf{\bibinfo{volume}{{\bf 89}}},
  \bibinfo{pages}{042101} (\bibinfo{year}{2014}{\natexlab{a}}).

\bibitem[{\citenamefont{Ca{\~n}as
  et~al.}(2014{\natexlab{a}})\citenamefont{Ca{\~n}as, Arias, Etcheverry,
  G{\'o}mez, Cabello, Xavier, and Lima}}]{canas-cabello-14}
\bibinfo{author}{\bibfnamefont{G.}~\bibnamefont{Ca{\~n}as}},
  \bibinfo{author}{\bibfnamefont{M.}~\bibnamefont{Arias}},
  \bibinfo{author}{\bibfnamefont{S.}~\bibnamefont{Etcheverry}},
  \bibinfo{author}{\bibfnamefont{E.~S.} \bibnamefont{G{\'o}mez}},
  \bibinfo{author}{\bibfnamefont{A.}~\bibnamefont{Cabello}},
  \bibinfo{author}{\bibfnamefont{G.~B.} \bibnamefont{Xavier}},
  \bibnamefont{and} \bibinfo{author}{\bibfnamefont{G.}~\bibnamefont{Lima}},
  \bibinfo{journal}{{\it Phys. Rev. Lett.}} \textbf{\bibinfo{volume}{{\bf
  113}}}, \bibinfo{pages}{090404} (\bibinfo{year}{2014}{\natexlab{a}}).

\bibitem[{\citenamefont{Ca{\~n}as
  et~al.}(2014{\natexlab{b}})\citenamefont{Ca{\~n}as, Etcheverry, G{\'o}mez,
  Xavier, Lima, and Cabello}}]{canas-cabello-8d-14}
\bibinfo{author}{\bibfnamefont{G.}~\bibnamefont{Ca{\~n}as}},
  \bibinfo{author}{\bibfnamefont{S.}~\bibnamefont{Etcheverry}},
  \bibinfo{author}{\bibfnamefont{E.~S.} \bibnamefont{G{\'o}mez}},
  \bibinfo{author}{\bibfnamefont{G.~B.} \bibnamefont{Xavier}},
  \bibinfo{author}{\bibfnamefont{G.}~\bibnamefont{Lima}}, \bibnamefont{and}
  \bibinfo{author}{\bibfnamefont{A.}~\bibnamefont{Cabello}},
  \bibinfo{journal}{{\it Phys. Rev. A}} \textbf{\bibinfo{volume}{{\bf 90}}},
  \bibinfo{pages}{012119} (\bibinfo{year}{2014}{\natexlab{b}}).

\bibitem[{\citenamefont{Megill et~al.}(2011)\citenamefont{Megill, Fresl,
  Waegell, Aravind, and Pavi{\v c}i{\'c}}}]{mfwap-s-11}
\bibinfo{author}{\bibfnamefont{N.~D.} \bibnamefont{Megill}},
  \bibinfo{author}{\bibfnamefont{K.}~\bibnamefont{Fresl}},
  \bibinfo{author}{\bibfnamefont{M.}~\bibnamefont{Waegell}},
  \bibinfo{author}{\bibfnamefont{P.~K.} \bibnamefont{Aravind}},
  \bibnamefont{and} \bibinfo{author}{\bibfnamefont{M.}~\bibnamefont{Pavi{\v
  c}i{\'c}}}, \bibinfo{journal}{{\it Phys. Lett. A}}
  \textbf{\bibinfo{volume}{{\bf 375}}}, \bibinfo{pages}{3419–}
  (\bibinfo{year}{2011}), \bibinfo{note}{{\it Supplementary Material}}.

\bibitem[{\citenamefont{Gleason}(1957)}]{gleason}
\bibinfo{author}{\bibfnamefont{A.~M.} \bibnamefont{Gleason}},
  \bibinfo{journal}{{\it J. Math. Mech.}} \textbf{\bibinfo{volume}{{\bf 6}}},
  \bibinfo{pages}{885} (\bibinfo{year}{1957}).

\bibitem[{\citenamefont{Kochen and Specker}(1967)}]{koch-speck}
\bibinfo{author}{\bibfnamefont{S.}~\bibnamefont{Kochen}} \bibnamefont{and}
  \bibinfo{author}{\bibfnamefont{E.~P.} \bibnamefont{Specker}},
  \bibinfo{journal}{{\it J. Math. Mech.}} \textbf{\bibinfo{volume}{{\bf 17}}},
  \bibinfo{pages}{59} (\bibinfo{year}{1967}).

\bibitem[{\citenamefont{Zimba and Penrose}(1993)}]{zimba-penrose}
\bibinfo{author}{\bibfnamefont{J.}~\bibnamefont{Zimba}} \bibnamefont{and}
  \bibinfo{author}{\bibfnamefont{R.}~\bibnamefont{Penrose}},
  \bibinfo{journal}{{\it Stud. Hist. Phil. Sci.}} \textbf{\bibinfo{volume}{{\bf
  24}}}, \bibinfo{pages}{697} (\bibinfo{year}{1993}).

\bibitem[{\citenamefont{Pavi{\v c}i{\'c}
  et~al.}(2010{\natexlab{b}})\citenamefont{Pavi{\v c}i{\'c}, Megill, and
  Merlet}}]{pmm-2-10}
\bibinfo{author}{\bibfnamefont{M.}~\bibnamefont{Pavi{\v c}i{\'c}}},
  \bibinfo{author}{\bibfnamefont{N.~D.} \bibnamefont{Megill}},
  \bibnamefont{and} \bibinfo{author}{\bibfnamefont{J.-P.}
  \bibnamefont{Merlet}}, \bibinfo{journal}{{\it Phys. Lett. A}}
  \textbf{\bibinfo{volume}{{\bf 374}}}, \bibinfo{pages}{2122}
  (\bibinfo{year}{2010}{\natexlab{b}}).

\bibitem[{\citenamefont{Swift and Wright}(1980)}]{anti-shimony}
\bibinfo{author}{\bibfnamefont{A.~R.} \bibnamefont{Swift}} \bibnamefont{and}
  \bibinfo{author}{\bibfnamefont{R.}~\bibnamefont{Wright}},
  \bibinfo{journal}{{\it J. Math. Phys.}} \textbf{\bibinfo{volume}{{\bf 21}}},
  \bibinfo{pages}{77} (\bibinfo{year}{1980}).

\bibitem[{\citenamefont{Pavi{\v c}i{\'c}
  et~al.}(2005{\natexlab{a}})\citenamefont{Pavi{\v c}i{\'c}, Merlet, Mc{K}ay,
  and Megill}}]{pmmm04a-arXiv}
\bibinfo{author}{\bibfnamefont{M.}~\bibnamefont{Pavi{\v c}i{\'c}}},
  \bibinfo{author}{\bibfnamefont{J.-P.} \bibnamefont{Merlet}},
  \bibinfo{author}{\bibfnamefont{B.~D.} \bibnamefont{Mc{K}ay}},
  \bibnamefont{and} \bibinfo{author}{\bibfnamefont{N.~D.}
  \bibnamefont{Megill}}, \bibinfo{journal}{{\it J. Phys. A}}
  \textbf{\bibinfo{volume}{{\bf 38}}}, \bibinfo{pages}{1577}
  (\bibinfo{year}{2005}{\natexlab{a}}), \bibinfo{note}{{\it
  {A}r{X}iv:quant-ph/0409014}}.

\bibitem[{\citenamefont{Ruuge}(2012)}]{ruuge12}
\bibinfo{author}{\bibfnamefont{A.~E.} \bibnamefont{Ruuge}},
  \bibinfo{journal}{{\it J. Phys. A}} \textbf{\bibinfo{volume}{{\bf 45}}},
  \bibinfo{pages}{465304} (\bibinfo{year}{2012}).

\bibitem[{\citenamefont{Mc{K}ay et~al.}(2000)\citenamefont{Mc{K}ay, Megill, and
  Pavi{\v c}i{\'c}}}]{bdm-ndm-mp-1}
\bibinfo{author}{\bibfnamefont{B.~D.} \bibnamefont{Mc{K}ay}},
  \bibinfo{author}{\bibfnamefont{N.~D.} \bibnamefont{Megill}},
  \bibnamefont{and} \bibinfo{author}{\bibfnamefont{M.}~\bibnamefont{Pavi{\v
  c}i{\'c}}}, \bibinfo{journal}{{\it Int. J. Theor. Phys.}}
  \textbf{\bibinfo{volume}{{\bf 39}}}, \bibinfo{pages}{2381}
  (\bibinfo{year}{2000}).

\bibitem[{\citenamefont{Pavi{\v c}i{\'c}
  et~al.}(2005{\natexlab{b}})\citenamefont{Pavi{\v c}i{\'c}, Merlet, Mc{K}ay,
  and Megill}}]{pmmm05a}
\bibinfo{author}{\bibfnamefont{M.}~\bibnamefont{Pavi{\v c}i{\'c}}},
  \bibinfo{author}{\bibfnamefont{J.-P.} \bibnamefont{Merlet}},
  \bibinfo{author}{\bibfnamefont{B.~D.} \bibnamefont{Mc{K}ay}},
  \bibnamefont{and} \bibinfo{author}{\bibfnamefont{N.~D.}
  \bibnamefont{Megill}}, \bibinfo{journal}{{\it J. Phys. A}}
  \textbf{\bibinfo{volume}{{\bf 38}}}, \bibinfo{pages}{1577}
  (\bibinfo{year}{2005}{\natexlab{b}}).

\bibitem[{\citenamefont{Pavi{\v c}i{\'c} et~al.}(2011)\citenamefont{Pavi{\v
  c}i{\'c}, Megill, Aravind, and Waegell}}]{mp-nm-pka-mw-11}
\bibinfo{author}{\bibfnamefont{M.}~\bibnamefont{Pavi{\v c}i{\'c}}},
  \bibinfo{author}{\bibfnamefont{N.~D.} \bibnamefont{Megill}},
  \bibinfo{author}{\bibfnamefont{P.~K.} \bibnamefont{Aravind}},
  \bibnamefont{and} \bibinfo{author}{\bibfnamefont{M.}~\bibnamefont{Waegell}},
  \bibinfo{journal}{{\it J. Math. Phys.}} \textbf{\bibinfo{volume}{{\bf 52}}},
  \bibinfo{pages}{022104} (\bibinfo{year}{2011}).

\bibitem[{\citenamefont{Lison{\v e}k
  et~al.}(2014{\natexlab{b}})\citenamefont{Lison{\v e}k, Raussendorf, and
  Singh}}]{lisonek-raus-14}
\bibinfo{author}{\bibfnamefont{P.}~\bibnamefont{Lison{\v e}k}},
  \bibinfo{author}{\bibfnamefont{R.}~\bibnamefont{Raussendorf}},
  \bibnamefont{and} \bibinfo{author}{\bibfnamefont{V.}~\bibnamefont{Singh}}
  (\bibinfo{year}{2014}{\natexlab{b}}), \bibinfo{note}{{\tt ar{X}iv:1401.3035v1
  [quant-ph]}}.

\bibitem[{\citenamefont{Pavi{\v c}i{\'c}}(2002)}]{mporl02}
\bibinfo{author}{\bibfnamefont{M.}~\bibnamefont{Pavi{\v c}i{\'c}}}, in
  \emph{\bibinfo{booktitle}{{SCI} 2002/{ISAS} 2002 {P}roceedings, {T}he 6th
  {W}orld {M}ulticonference on {S}ystemics, {C}ybernetics, and {I}nformatics}},
  edited by \bibinfo{editor}{\bibfnamefont{N.}~\bibnamefont{Callaos}},
  \bibinfo{editor}{\bibfnamefont{Y.}~\bibnamefont{He}}, \bibnamefont{and}
  \bibinfo{editor}{\bibfnamefont{J.~A.} \bibnamefont{Perez-{P}eraza}}
  (\bibinfo{publisher}{SCI}, \bibinfo{address}{Orlando, Florida},
  \bibinfo{year}{2002}), vol. \bibinfo{volume}{{\bf XVII}, {SCI} in {P}hysics,
  {A}stronomy and {C}hemistry}, pp. \bibinfo{pages}{65--70}.

\bibitem[{\citenamefont{Pavi{\v c}i{\'c} et~al.}(2004)\citenamefont{Pavi{\v
  c}i{\'c}, Merlet, and Megill}}]{pmmm04b}
\bibinfo{author}{\bibfnamefont{M.}~\bibnamefont{Pavi{\v c}i{\'c}}},
  \bibinfo{author}{\bibfnamefont{J.-P.} \bibnamefont{Merlet}},
  \bibnamefont{and} \bibinfo{author}{\bibfnamefont{N.~D.}
  \bibnamefont{Megill}}, \bibinfo{journal}{{\it The {F}rench {N}ational
  {I}nstitute for {R}esearch in {C}omputer {S}cience and {C}ontrol {R}esearch
  {R}eports}} \textbf{\bibinfo{volume}{{\bf {RR-5388}}}}
  (\bibinfo{year}{2004}).

\bibitem[{\citenamefont{Pavi{\v c}i{\'c}}(2005)}]{pavicic-book-05}
\bibinfo{author}{\bibfnamefont{M.}~\bibnamefont{Pavi{\v c}i{\'c}}},
  \emph{\bibinfo{title}{Quantum Computation and Quantum Communication: {T}heory
  and Experiments}} (\bibinfo{publisher}{Springer}, \bibinfo{address}{New
  York}, \bibinfo{year}{2005}).

\bibitem[{\citenamefont{Peres}(1991)}]{peres}
\bibinfo{author}{\bibfnamefont{A.}~\bibnamefont{Peres}}, \bibinfo{journal}{{\it
  J. Phys. A}} \textbf{\bibinfo{volume}{{\bf 24}}}, \bibinfo{pages}{L175}
  (\bibinfo{year}{1991}).

\bibitem[{\citenamefont{Kernaghan and Peres}(1995)}]{kern-peres}
\bibinfo{author}{\bibfnamefont{M.}~\bibnamefont{Kernaghan}} \bibnamefont{and}
  \bibinfo{author}{\bibfnamefont{A.}~\bibnamefont{Peres}},
  \bibinfo{journal}{{\it Phys. Lett. A}} \textbf{\bibinfo{volume}{{\bf 198}}},
  \bibinfo{pages}{1} (\bibinfo{year}{1995}).

\bibitem[{\citenamefont{Kernaghan}(1994)}]{kern}
\bibinfo{author}{\bibfnamefont{M.}~\bibnamefont{Kernaghan}},
  \bibinfo{journal}{{\it J. Phys. A}} \textbf{\bibinfo{volume}{{\bf 27}}},
  \bibinfo{pages}{L829} (\bibinfo{year}{1994}).

\bibitem[{\citenamefont{Bub}(1996)}]{bub}
\bibinfo{author}{\bibfnamefont{J.}~\bibnamefont{Bub}}, \bibinfo{journal}{{\it
  Found. Phys.}} \textbf{\bibinfo{volume}{{\bf 26}}}, \bibinfo{pages}{787}
  (\bibinfo{year}{1996}).

\bibitem[{\citenamefont{Cabello et~al.}(1996)\citenamefont{Cabello, Estebaranz,
  and {Garc{\'\i}a-Alcaine}}}]{cabell-est-96a}
\bibinfo{author}{\bibfnamefont{A.}~\bibnamefont{Cabello}},
  \bibinfo{author}{\bibfnamefont{J.~M.} \bibnamefont{Estebaranz}},
  \bibnamefont{and}
  \bibinfo{author}{\bibfnamefont{G.}~\bibnamefont{{Garc{\'\i}a-Alcaine}}},
  \bibinfo{journal}{{\it Phys. Lett. A}} \textbf{\bibinfo{volume}{{\bf 212}}},
  \bibinfo{pages}{183} (\bibinfo{year}{1996}).

\bibitem[{\citenamefont{Aravind and Lee-{E}lkin}(1998)}]{aravind-600}
\bibinfo{author}{\bibfnamefont{P.~K.} \bibnamefont{Aravind}} \bibnamefont{and}
  \bibinfo{author}{\bibfnamefont{F.}~\bibnamefont{Lee-{E}lkin}},
  \bibinfo{journal}{{\it J. Phys. A}} \textbf{\bibinfo{volume}{{\bf 31}}},
  \bibinfo{pages}{9829} (\bibinfo{year}{1998}).

\bibitem[{\citenamefont{Peres}(1993)}]{peres-book}
\bibinfo{author}{\bibfnamefont{A.}~\bibnamefont{Peres}},
  \emph{\bibinfo{title}{Quantum Theory: {C}oncepts and Methods}}
  (\bibinfo{publisher}{Kluwer}, \bibinfo{address}{Dordrecht},
  \bibinfo{year}{1993}).

\bibitem[{\citenamefont{Waegell and Aravind}(2010)}]{aravind10}
\bibinfo{author}{\bibfnamefont{M.}~\bibnamefont{Waegell}} \bibnamefont{and}
  \bibinfo{author}{\bibfnamefont{P.~K.} \bibnamefont{Aravind}},
  \bibinfo{journal}{{\it J. Phys. A}} \textbf{\bibinfo{volume}{{\bf 43}}},
  \bibinfo{pages}{105304} (\bibinfo{year}{2010}).

\bibitem[{\citenamefont{Waegell et~al.}(2011)\citenamefont{Waegell, Aravind,
  Megill, and Pavi{\v c}i{\'c}}}]{waeg-aravind-megill-pavicic-11}
\bibinfo{author}{\bibfnamefont{M.}~\bibnamefont{Waegell}},
  \bibinfo{author}{\bibfnamefont{P.~K.} \bibnamefont{Aravind}},
  \bibinfo{author}{\bibfnamefont{N.~D.} \bibnamefont{Megill}},
  \bibnamefont{and} \bibinfo{author}{\bibfnamefont{M.}~\bibnamefont{Pavi{\v
  c}i{\'c}}}, \bibinfo{journal}{{\it Found. Phys.}}
  \textbf{\bibinfo{volume}{{\bf 41}}}, \bibinfo{pages}{883}
  (\bibinfo{year}{2011}).

\bibitem[{\citenamefont{Pavi{\v c}i{\'c}}(2013)}]{pavicic-book-13}
\bibinfo{author}{\bibfnamefont{M.}~\bibnamefont{Pavi{\v c}i{\'c}}},
  \emph{\bibinfo{title}{Companion to Quantum Computation and Communication}}
  (\bibinfo{publisher}{Wiley-VCH}, \bibinfo{address}{Weinheim},
  \bibinfo{year}{2013}).

\bibitem[{\citenamefont{Waegell and Aravind}(2014)}]{waeg-aravind-fp-14}
\bibinfo{author}{\bibfnamefont{M.}~\bibnamefont{Waegell}} \bibnamefont{and}
  \bibinfo{author}{\bibfnamefont{P.~K.} \bibnamefont{Aravind}},
  \bibinfo{journal}{{\it Found. Phys.}} \textbf{\bibinfo{volume}{{\bf 44}}},
  \bibinfo{pages}{1085} (\bibinfo{year}{2014}).

\bibitem[{\citenamefont{Waegell and Aravind}(2017)}]{waeg-aravind-17-arXiv}
\bibinfo{author}{\bibfnamefont{M.}~\bibnamefont{Waegell}} \bibnamefont{and}
  \bibinfo{author}{\bibfnamefont{P.~K.} \bibnamefont{Aravind}},
  \emph{\bibinfo{title}{The {P}enrose dodecahedron and the {W}itting polytope
  are identical in $\mathbb{CP}^{3}$}}, \bibinfo{howpublished}{{\it
  {A}r{X}iv:1701.06512}} (\bibinfo{year}{2017}).

\bibitem[{\citenamefont{Aravind and
  Waegell}(2013)}]{aravind-waegell-6dim-private}
\bibinfo{author}{\bibfnamefont{P.~K.} \bibnamefont{Aravind}} \bibnamefont{and}
  \bibinfo{author}{\bibfnamefont{M.}~\bibnamefont{Waegell}},
  \emph{\bibinfo{title}{6-dim {KS} master set, {P}rivate communication}}
  (\bibinfo{year}{2013}).

\bibitem[{\citenamefont{Ruuge and {van Oystaeyen}}(2005)}]{ruuge05}
\bibinfo{author}{\bibfnamefont{A.~E.} \bibnamefont{Ruuge}} \bibnamefont{and}
  \bibinfo{author}{\bibfnamefont{F.}~\bibnamefont{{van Oystaeyen}}},
  \bibinfo{journal}{{\it J. Math. Phys.}} \textbf{\bibinfo{volume}{{\bf 46}}},
  \bibinfo{pages}{052109} (\bibinfo{year}{2005}).

\bibitem[{\citenamefont{Planat}(2012)}]{planat-12}
\bibinfo{author}{\bibfnamefont{M.}~\bibnamefont{Planat}},
  \bibinfo{journal}{{\it Eur. Phys. J. Plus}} \textbf{\bibinfo{volume}{{\bf
  127}}}, \bibinfo{pages}{86} (\bibinfo{year}{2012}).

\bibitem[{\citenamefont{Waegell and Aravind}(2015)}]{waeg-aravind-jpa-15}
\bibinfo{author}{\bibfnamefont{M.}~\bibnamefont{Waegell}} \bibnamefont{and}
  \bibinfo{author}{\bibfnamefont{P.~K.} \bibnamefont{Aravind}},
  \bibinfo{journal}{{\it J. Phys. A}} \textbf{\bibinfo{volume}{{\bf 48}}},
  \bibinfo{pages}{225301} (\bibinfo{year}{2015}).

\bibitem[{\citenamefont{Waegell and Aravind}(2012)}]{waegel-aravind-12}
\bibinfo{author}{\bibfnamefont{M.}~\bibnamefont{Waegell}} \bibnamefont{and}
  \bibinfo{author}{\bibfnamefont{P.~K.} \bibnamefont{Aravind}},
  \bibinfo{journal}{{\it J. Phys. A}} \textbf{\bibinfo{volume}{{\bf 45}}},
  \bibinfo{pages}{405301} (\bibinfo{year}{2012}).

\bibitem[{\citenamefont{Harvey and
  Chryssanthacopoulos}(2012)}]{harv-cryss-aravind-12a}
\bibinfo{author}{\bibfnamefont{C.}~\bibnamefont{Harvey}} \bibnamefont{and}
  \bibinfo{author}{\bibfnamefont{J.}~\bibnamefont{Chryssanthacopoulos}},
  \bibinfo{type}{Tech. Rep.} \bibinfo{number}{PH-PKA-JC08},
  \bibinfo{institution}{Worcester Polytechnic Institute}
  (\bibinfo{year}{2012}), \bibinfo{note}{{\tt
  https://web.wpi.edu/Pubs/E-project/Available/E-project-042108-171725/unrestricted/MQPReport.pdf}}.

\bibitem[{\citenamefont{Aravind}(2002)}]{aravind-02}
\bibinfo{author}{\bibfnamefont{P.~K.} \bibnamefont{Aravind}},
  \bibinfo{journal}{{\it Found. Phys. Lett.}} \textbf{\bibinfo{volume}{{\bf
  15}}}, \bibinfo{pages}{397} (\bibinfo{year}{2002}).

\bibitem[{\citenamefont{{DiV}incenzo and Peres}(1997)}]{divinc-peres}
\bibinfo{author}{\bibfnamefont{D.~P.} \bibnamefont{{DiV}incenzo}}
  \bibnamefont{and} \bibinfo{author}{\bibfnamefont{A.}~\bibnamefont{Peres}},
  \bibinfo{journal}{{\it Phys. Rev. A}} \textbf{\bibinfo{volume}{{\bf 55}}},
  \bibinfo{pages}{4089} (\bibinfo{year}{1997}).

\bibitem[{\citenamefont{Planat and Saniga}(2012)}]{planat-saniga-12}
\bibinfo{author}{\bibfnamefont{M.}~\bibnamefont{Planat}} \bibnamefont{and}
  \bibinfo{author}{\bibfnamefont{M.}~\bibnamefont{Saniga}},
  \bibinfo{journal}{{\it Phys. Lett. A}} \textbf{\bibinfo{volume}{{\bf 376}}},
  \bibinfo{pages}{3485} (\bibinfo{year}{2012}).

\bibitem[{\citenamefont{Pavi{\v c}i{\'c}
  et~al.}(2005{\natexlab{c}})\citenamefont{Pavi{\v c}i{\'c}, Merlet, Mc{K}ay,
  and Megill}}]{pmmm05a-corr}
\bibinfo{author}{\bibfnamefont{M.}~\bibnamefont{Pavi{\v c}i{\'c}}},
  \bibinfo{author}{\bibfnamefont{J.-P.} \bibnamefont{Merlet}},
  \bibinfo{author}{\bibfnamefont{B.~D.} \bibnamefont{Mc{K}ay}},
  \bibnamefont{and} \bibinfo{author}{\bibfnamefont{N.~D.}
  \bibnamefont{Megill}}, \bibinfo{journal}{{\it J. Phys. A}}
  \textbf{\bibinfo{volume}{{\bf 38}}}, \bibinfo{pages}{1577}
  (\bibinfo{year}{2005}{\natexlab{c}}).

\bibitem[{\citenamefont{Larsson}(2002)}]{larsson}
\bibinfo{author}{\bibfnamefont{J.-{\AA}.} \bibnamefont{Larsson}},
  \bibinfo{journal}{{\it Europhys. Lett.}} \textbf{\bibinfo{volume}{{\bf 58}}},
  \bibinfo{pages}{799} (\bibinfo{year}{2002}).

\bibitem[{\citenamefont{Pavi{\v c}i{\'c}
  et~al.}(2005{\natexlab{d}})\citenamefont{Pavi{\v c}i{\'c}, Merlet, Mc{K}ay,
  and Megill}}]{pmmm04c}
\bibinfo{author}{\bibfnamefont{M.}~\bibnamefont{Pavi{\v c}i{\'c}}},
  \bibinfo{author}{\bibfnamefont{J.-P.} \bibnamefont{Merlet}},
  \bibinfo{author}{\bibfnamefont{B.~D.} \bibnamefont{Mc{K}ay}},
  \bibnamefont{and} \bibinfo{author}{\bibfnamefont{N.~D.}
  \bibnamefont{Megill}}, \bibinfo{journal}{{\it J. Phys. A}}
  \textbf{\bibinfo{volume}{{\bf 38}}}, \bibinfo{pages}{1577}
  (\bibinfo{year}{2005}{\natexlab{d}}), \bibinfo{note}{and {\bf 38}, 3709
  (2005) (corrigendum)}.

\bibitem[{\citenamefont{Held}(2009)}]{held-09}
\bibinfo{author}{\bibfnamefont{C.}~\bibnamefont{Held}}, in
  \emph{\bibinfo{booktitle}{Compendium of Quantum Physics}}, edited by
  \bibinfo{editor}{\bibfnamefont{D.}~\bibnamefont{Greenberger}},
  \bibinfo{editor}{\bibfnamefont{K.}~\bibnamefont{Hentschel}},
  \bibnamefont{and} \bibinfo{editor}{\bibfnamefont{F.}~\bibnamefont{Weinert}}
  (\bibinfo{publisher}{Springer}, \bibinfo{address}{New-York},
  \bibinfo{year}{2009}), pp. \bibinfo{pages}{331--335}.

\bibitem[{\citenamefont{Yu and Oh}(2012)}]{yu-oh-12}
\bibinfo{author}{\bibfnamefont{S.}~\bibnamefont{Yu}} \bibnamefont{and}
  \bibinfo{author}{\bibfnamefont{C.~H.} \bibnamefont{Oh}},
  \bibinfo{journal}{{\it Phys. Rev. Lett.}} \textbf{\bibinfo{volume}{{\bf
  108}}}, \bibinfo{pages}{030402} (\bibinfo{year}{2012}).

\bibitem[{\citenamefont{Kochen}(2015)}]{kochen-15}
\bibinfo{author}{\bibfnamefont{S.}~\bibnamefont{Kochen}},
  \emph{\bibinfo{title}{Private communication}} (\bibinfo{year}{2015}).

\bibitem[{\citenamefont{Gould and
  Aravind}(2010)}]{peres-penrose-gould-aravind09}
\bibinfo{author}{\bibfnamefont{E.}~\bibnamefont{Gould}} \bibnamefont{and}
  \bibinfo{author}{\bibfnamefont{P.~K.} \bibnamefont{Aravind}},
  \bibinfo{journal}{{\it Found. Phys.}} \textbf{\bibinfo{volume}{{\bf 40}}},
  \bibinfo{pages}{1096} (\bibinfo{year}{2010}).

\bibitem[{\citenamefont{Waegell and Aravind}(2013)}]{waeg-aravind-pra-13}
\bibinfo{author}{\bibfnamefont{M.}~\bibnamefont{Waegell}} \bibnamefont{and}
  \bibinfo{author}{\bibfnamefont{P.~K.} \bibnamefont{Aravind}},
  \bibinfo{journal}{{\it Phys. Rev. A}} \textbf{\bibinfo{volume}{{\bf 88}}},
  \bibinfo{pages}{012102} (\bibinfo{year}{2013}).

\end{thebibliography}
\end{document}